\documentclass[12pt]{article}
\usepackage{amsmath}
\usepackage{graphicx}
\usepackage{enumerate}
\usepackage{natbib}
\usepackage{setspace} 
\usepackage{amsfonts}
\usepackage{bm}
\usepackage{amsthm}
\usepackage{url} 
\usepackage[colorlinks,linkcolor=red,anchorcolor=black,citecolor=blue]{hyperref}
\newcommand{\blind}{1}

\newtheorem{theorem}{Theorem}
\newtheorem{condition}{Condition}
\newtheorem{lemma}{Lemma}[section]
\newtheorem{definition}{Definition}

\AtBeginDocument{%
	\setlength{\abovedisplayskip}{0pt}%
	\setlength{\belowdisplayskip}{0pt}%
	\setlength{\abovedisplayshortskip}{0pt}%
	\setlength{\belowdisplayshortskip}{0pt}%
}

\begin{document}

\def\spacingset#1{\renewcommand{\baselinestretch}%
{#1}\small\normalsize} \spacingset{1}


\if1\blind
{
  \title{\bf Unbiased estimation and asymptotically valid inference in multivariable Mendelian randomization with many weak instrumental variables}
  \author{Yihe Yang, Noah Lorincz-Comi, Xiaofeng Zhu \thanks{Email: xxz10@case.edu. This work was supported by grant HG011052 (to X.Z.) from the National Human Genome Research Institute (NHGRI), USA}\hspace{.2cm} \\
    Department of Population and Quantitative Health Science\\ Case Western Reserve University}
  \date{}
  \maketitle
} \fi

\if0\blind
{
  \title{\LARGE\bf Unbiased estimation and asymptotically valid inference in multivariable Mendelian randomization with many weak instrumental variables}
   \date{}
  \maketitle
} \fi

\bigskip
\begin{abstract}
Mendelian randomization (MR) is a popular epidemiological approach that utilizes genetic variants as instrumental variables (IVs) to infer the causal relationships between exposures and an outcome in the genome-wide association studies (GWAS) era. It is well-known that the inverse-variance weighted (IVW) estimate of causal effect suffers from bias caused by the violation of valid IV conditions, however, the quantitative degree of this bias has not been well characterized and understood. This paper contributes to the theoretical investigation and practical solution of the causal effect estimation in multivariable MR. First, we prove that the bias of IVW estimate is a product of the weak instrument and estimation error biases, where the estimation error bias is caused linearly by measurement error and confounder biases with a trade-off due to the sample overlap among exposure and outcome GWAS cohorts. Second, we demonstrate that our novel multivariable MR approach, MR using Bias-corrected Estimating Equation (MRBEE), can estimate the causal effect unbiasedly in the presence of many weak IVs. Asymptotic behaviors of IVW and MRBEE are investigated under moderate conditions, where MRBEE is shown superior to IVW in terms of unbiasedness and asymptotic validity. Simulations exhibit that only MRBEE can provide a strongly asymptotically unbiased estimate of causal effect in comparison with existing MR methods. Applied to data from the UK Biobank, MRBEE can eliminate weak instrument and estimation error biases and provide valid causal inferences. R package \texttt{MRBEE} and supplementary materials are available online.
\end{abstract}

\noindent%
\vfill

\spacingset{1.05} 
\section{Introduction}
\label{sec:intro}

A genome-wide association study (GWAS) refers to the identification of genetic variants statistically associated with complex traits or diseases across the whole genome using large population cohorts \citep{visscher201710}. GWAS typically examines associations between single-nucleotide polymorphisms (SNPs) and a trait but can also handle other genetic variants such as insertion and deletions (indels) and structural variations (SVs) \citep{gresham2008comparing}. The first example of a successful GWAS was the 2005 GWAS which revealed two genetic variants significantly associated with age-related macular degeneration \citep{klein2005complement}. To date, over 5,000 human GWAS have investigated approximately 2,000 diseases and traits and have identified more than 400,000 genetic associations \citep{wijmenga2018importance}. This groundbreaking work has uncovered numerous compelling associations with human complex traits and diseases, shedding light on the disease mechanisms and enhancing clinical care and personalized medicine \citep{tam2019benefits}.

Mendelian randomization (MR) is an epidemiological method that utilizes genetic variants as instrumental variables (IVs) to infer whether an exposure causally influences an outcome \citep{burgess2021mendelian}. Since the genotypes of individuals are randomly inherited from their parents and generally do not change during their lifetime, genetic variants are considered to be independent of underlying confounders and hence can be used as IVs to eliminate confounding bias. Early MR studies progressed slowly because individual-level data simultaneously including genotypes and phenotypes were rarely available \citep{ebrahim2008mendelian}. Recently, many large GWAS have been published and the corresponding summary statistics are available in databases such as the GWAS Catalog \citep{macarthur2017new} (\url{https://www.ebi.ac.uk/gwas/}), dbGaP (\url{https://www.ncbi.nlm.nih.gov/gap/}), and UK biobank (UKBB, \citet{sudlow2015uk}) (\url{https://www.ukbiobank.ac.uk/}). The accuracy of causal effect estimation is improved and valuable insights into the causal relationships between common risk factors and diseases are uncovered by utilizing MR with GWAS summary data \citep{wang2022mendelian}.

The inverse-variance weighted (IVW) method is the most popular approach used to perform MR with GWAS summary data. A causal effect estimate yielded by the IVW method is supposedly unbiased if three so-called valid IV conditions are satisfied: (IV1) the genetic variants are strongly associated with the exposure; (IV2) the genetic variants are associated with the outcome only through the exposure; and (IV3) the genetic variants are independent of confounders  \citep{bowden2015mendelian}. The directed acyclic graph (DAG) of valid IV conditions is shown in panel (a) in Figure \ref{fig1}. Due to the complexity of genetic architecture, conditions IV2 and IV3 are often difficult to validate \citep{zhu2020mendelian}. Meanwhile, it is challenging to quantify the instrument strength and define a universal criterion for concluding that an IV satisfies condition IV1, although the F statistic can be utilized as a rough measure of weak instrument bias \citep{burgess2011avoiding}. Thus, quantifying and eliminating the bias of IVW estimate in MR analysis will lead to valid causal inference and help to understand disease etiology.

\begin{figure}[t]
	\begin{center}
		\includegraphics[width=5in]{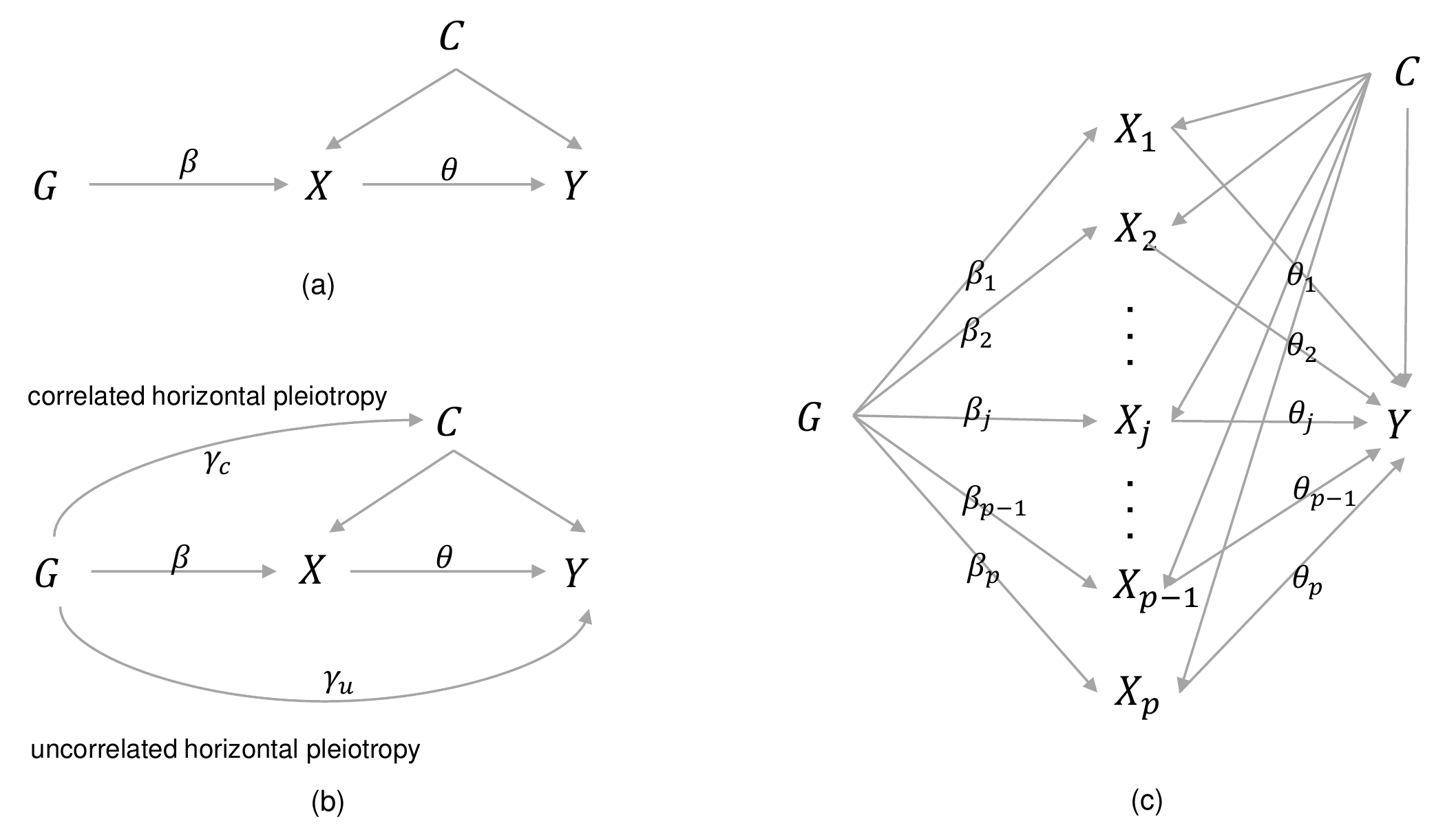}
	\end{center}
	\caption{\footnotesize DAG of MR and multivariable MR. Panel (a): causal path digram with valid genetic IVs. Panel (b): causal path digram with UHP and CHP. Panel (c):  causal path digram for multivariable MR methods. $G$: genetic IVs; $X$: exposure; $Y$: outcome; $C$: confounders; $\beta$: association between $G$ and $X$; $\theta$: causal effect of $X$ on $Y$; $\gamma_c$: direct correlation between $G$ and $C$; $\gamma_u$: direct correlation between $G$ and $Y$.
		 \label{fig1}}
\end{figure}

A genetic variant is termed a pleiotropic variant or pleiotropy if it simultaneously affects multiple traits through different pathways. There are two types of pleiotropy: vertical and horizontal pleiotropy, where the former refers to the genetic variant associated with one trait through the mediation of another trait (as described in panel (a) in Figure \ref{fig1}), while the latter refers to the genetic variant independently associated with both traits (as illustrated in panel (b) in Figure \ref{fig1}). IVs with evidence of horizontal pleiotropy should be removed before applying IVW because it violates either the (IV2) condition or the (IV3) condition; otherwise, a biased causal effect estimate is likely obtained. In the literature, there are three strategies to remove the effect of horizontal pleiotropy: 1) Identifying and excluding horizontally pleiotropic IVs by using hypothesis tests, such as the MR pleiotropy residual sum and outlier (MR-PRESSO, \citet{verbanck2018detection}) and the iterative MR and pleiotropy (IMRP, \citet{zhu2021iterative}); 2) Eliminating the effect of horizontal pleiotropy by applying robust tools; e.g., the MR-Egger \citep{bowden2015mendelian}, MR-Median \citep{bowden2016consistent}, and MR-Lasso/MR-Robust \citep{rees2019robust}; 3) Automatically separating vertical pleiotropy from horizontal pleiotropy through a mixture mode, among which the representative methods include MR-Mix \citep{qi2019mendelian} and MR contamination mixture (MR-ConMix, \citet{burgess2020robust}).

It has been gradually realized that horizontal pleiotropy can be divided into uncorrelated horizontal pleiotropy (UHP) and correlated horizontal pleiotropy (CHP). UHP violates the (IV2) condition and usually refers to a genetic variant that is directly associated with the outcome. In contrast, CHP violates the IV3 condition and may occur when a genetic variant indirectly affects the outcome through the mediation of unspecified exposures. The DAG of UHP and CHP is shown in panel (b) in Figure \ref{fig1}.  \citet{morrison2020mendelian} proposed causal analysis using summary effect (CAUSE), which is the first MR approach accounting for UHP and CHP simultaneously.  \citet{cheng2022mr} proposed MR-Corr to detect CHP by a Bayesian mixture model and \citet{cheng2022mendelian} extended MR-Corr to MR-CUE (MR with CHP Unraveling shared Etiology and confounding) to detect the UHP and CHP simultaneously. Alternatively, \citet{xue2021constrained} proposed the constrained maximum likelihood-based MR (cML-MR) method that identifies UHP and CHP through Bayesian information criterion (BIC, \citet{schwarz1978estimating}). In addition, \citet{yuan2022likelihood} derived MR with automated instrument determination (MRAID) to address UHP and CHP, which allows vertical pleiotropy to be in high linkage disequilibrium (LD).

A significant disadvantage of most existing approaches is that they assume both UHP and CHP to have similar properties as outliers in the traditional regression approach. However, there is substantial evidence that most traits have shared moderate or high genetic correlations \citep{bulik2015atlas}, violating this technical assumption required by most existing approaches. Consequently, it is challenging to remove the effect of horizontally pleiotropic variants by considering only one exposure in MR analysis. Multivariable MR, which simultaneously estimates the causal effects of multiple exposures on an outcome, is compelling in resolving this problem \citep{burgess2015multivariable}. Multivariable MR recognizes the bias caused by horizontal pleiotropy as an omitted-variable bias (OVB), which will disappear automatically if all the omitted exposures are specified in the multivariable MR model. The DAG of multivariable MR is exhibited in panel (c) in Figure \ref{fig1}. So far, the multivariable versions of the IVW method, MR-Egger, MR-Median, and MR-Lasso/MR-Robust have been developed \citep{burgess2015multivariable,rees2017extending,grant2021pleiotropy}. \citet{sanderson2019examination} showed that the multivariable MR is able to unbiasedly estimate the causal effects of a target exposure when the other exposure is confounder, collider, or mediator of this exposure. 

Weak instrument bias arises when the majority of IVs are weakly associated with the exposures, therefore violating the (IV1) condition and making conventional MR methods unreliable \citep{burgess2011avoiding}. It is widely recognized that a common trait is often polygenic affected by hundreds or even thousands of independent variants/genes with small effect sizes. With the increasing sample sizes of GWAS, more and more trait-associated variants are being identified. Thus, the weak instrument bias is likely to become a considerable problem in future MR studies.  \citet{burgess2011avoiding} and \citet{sanderson2021testing} suggested using the F and conditional F statistics to measure the weak instrument bias in MR and multivariable MR, respectively. \citet{burgess2016bias} illustrated that the weak instrument bias also depends on the sample overlap in two-sample MR. \citet{sadreev2021navigating} examined the impact of sample overlap and winner's curse when weak instrument bias exists and observed that the weak instrument bias grew dramatically in the presence of winner's curse. For the univariate MR model with no sample overlap, \citet{zhao2020statistical} proposed the robust adjusted profile score to estimate the causal effect unbiasedly, while \citet{ye2021debiased} provided the debiased IVW (DIVW) method to remove the weak instrument bias of IVW estimate. Overall, these aforementioned methods have neither provided a comprehensive theoretical analysis of weak instrument bias nor a general solution to remove the weak instrument bias in both univariable MR and multivariable MR.

As the first contribution of this paper, we theoretically characterize the bias in multivariable IVW causal estimate. Specifically, we demonstrate that the bias of IVW causal estimate is the product of weak instrument and estimation error biases. Meanwhile, we demonstrate that the estimation error bias is a linear combination of measurement error \citep{yi2017statistical} and confounder biases, and the sample overlaps among multiple GWAS cohorts trade off the proportions of these two biases. With moderate conditions on the MR model, we theoretically illustrate how the number of IVs, sample sizes of GWAS studies, and sample overlap among GWAS cohorts influence the asymptotic behavior of multivariable IVW estimate. These theoretical findings are summarized in Theorem \ref{theorem2} that to our best knowledge is the first comprehensive investigation of multivariable IVW estimate.

As the second contribution of this paper, we demonstrate our novel multivariable MR approach, MR using Bias-corrected Estimating Equations (MRBEE), can estimate causal effects unbiasedly in the presence of many weak IVs. Under moderate conditions, we investigate the asymptotic behaviors of IVW and MRBEE, revealing that MRBEE is superior to IVW in terms of strongly asymptotic unbiasedness. In particular, only when an estimate is strongly asymptotically unbiased, the inference made based on this estimate is asymptotically valid. Simulations show that only MRBEE can provide unbiased causal effect estimates in the presence of many weak IVs. Applied to data from the UK Biobank, MRBEE can successfully remove the weak instrument and estimation error biases and therefore make valid causal inferences.

This paper is arranged as follows. In section 2, we study the asymptotic behavior of multivariate IVW estimate. In section 3, we introduce MRBEE and examine its asymptotic properties. In section 4, simulations are conducted to compare MRBEE with the existing methods. In section 5, we apply MRBEE to estimate the causal effects of exposures on cardiovascular disease. Discussion is presented in section 6 and proofs of the related theorems are shown in Appendix. R package \texttt{MRBEE} (\url{https://github.com/noahlorinczcomi/MRBEE}) and supplementary materials are available online.
\section{Mendelian Randomization}
In this section, we introduce the notations, the model of the multivariable MR, and the bias of the multivariable IVW. Since univariable MR is a special case of multivariable MR, MR refers to the multivariable MR unless otherwise specified.
\subsection{Notation}
For a vector $\boldsymbol a=(a_j)_{p\times 1}$, $||\boldsymbol a||_q=(\sum_{j=1}^p|a_j|^q)^{1/q}$ with $q\in[0,\infty]$. For a symmetric matrix $\mathbf A=(A_{ij})_{p\times p}$, $\lambda_{\max}(\mathbf A)$ and $\lambda_{\min }(\mathbf A)$ is its the maximum and minimum eigenvalues, $\mathbf A^+$ is its Moore–Penrose inverse; and $||\mathbf A||_q=\max\{||\mathbf A\boldsymbol a||_q,\ ||\boldsymbol a||_q=1\}$. Let diag($\boldsymbol\alpha$) be the diagonalizing operator of vector $\boldsymbol\alpha$ and $\mathbf A\odot\mathbf B$ be the Hadamard product of matrices $\mathbf A$ and $\mathbf B$. For a set $\mathcal A$, $|\mathcal A|$ is the number of elements in $\mathcal A$.  Notations $O(\cdot)$ and $o(\cdot)$ are the infinitely large and small quantities, while $O_P(\cdot)$ and $o_P(\cdot)$ mean that such relationships hold in probability. 
\subsection{Mendelian randomization model}
The central aim of MR is to estimate causal effects between exposures and an outcome unbiasedly. Let $\boldsymbol g_i=(g_{i1},\dots,g_{im})^\top$ be an ($m\times 1$) genotype value vector of $m$ genetic variants, $\boldsymbol x_i=(x_{i1},\dots,x_{ip})^\top$ be an ($p\times 1$) vector representing $p$ exposures, and $y_i$ be an outcome. Here, $m$ is the number of specified IVs, which is usually the number of independent loci with $p$-values reaching the genome-wide significant level. Let $\mathbf B=(\boldsymbol\beta_1,\dots,\boldsymbol\beta_m)^\top$ be an ($m\times p$) matrix of genetic effects on exposures with $\boldsymbol\beta_j=(\beta_{j1},\dots,\beta_{jp})^\top$ being an $(p\times 1)$ vector, and  $\boldsymbol\theta=(\theta_{1},\dots,\theta_p)^\top$ be an $(p\times1)$ vector of causal effects of the $p$ exposures on the outcome. The MR model is 
\begin{align}
	y_i&=\boldsymbol x_i^\top\boldsymbol\theta+ v_i,\label{eq1}\\
	\boldsymbol x_i&=\mathbf B^\top\boldsymbol g_i+\boldsymbol u_i\label{eq2}
\end{align}
where $\boldsymbol u_i$ and $v_i$ are the noise terms. Substituting for $\boldsymbol x_i$ in (\ref{eq1}), we obtain the equation
\begin{align}
y_i=\boldsymbol\alpha^\top\boldsymbol g_i+\boldsymbol\theta^\top\boldsymbol u_i+ v_i,\label{eq3}
\end{align}
where $\boldsymbol\alpha=\mathbf B\boldsymbol\theta$. In the literature, (\ref{eq1}) - (\ref{eq3}) have been named as the structural form, first-stage, and reduced form, respectively \citep{stock2002survey}.

In this paper, we assume that the total number of exposures $p$ is fixed and the causal effect $\boldsymbol\theta$ is fixed and bounded. The genetic variant $g_{ij}$ is standardized so that E$(g_{ij})=0$ and var$(g_{ij})=1$, and all IVs are in linkage equilibrium (LE), i.e., cov$(g_{ij},g_{ik})=0$ for $j\neq k$. The genetic effect $\boldsymbol\beta_j$ is random with zero-mean, covariance matrix $\bm\Sigma_{\beta\beta}$, and cumulative covariance matrix $\bm\Psi_{\beta\beta}$
\begin{align*}
\bm\Sigma_{\beta\beta}=\text{E}(\boldsymbol\beta_j\boldsymbol\beta_j^\top),\quad\bm\Psi_{\beta\beta}=m\bm\Sigma_{\beta\beta}.
\end{align*}
The covariance matrix $\bm\Sigma_{\beta\beta}$ will vanish as $m$ increase, but the cumulative covariance matrix $\bm\Psi_{\beta\beta}$ is still a constant matrix, representing the total genetic covariance contributed from the $m$ IVs. Next, the noise terms $\boldsymbol u_i$ and $v_j$ have zero-means and joint covariance matrix
\[
\bm\Sigma_{u\times v}=\text{cov}((\boldsymbol u_i^\top,v_j)^\top)=\begin{pmatrix}
	\bm\Sigma_{uu}&\boldsymbol\sigma_{uv}\\\boldsymbol\sigma_{uv}^\top&\sigma_{vv}.
\end{pmatrix}
\]
Thus, the exposure vector $\boldsymbol x_i$ and outcome $y_i$ have zero-means and joint covariance matrix
\[
\bm\Sigma_{x\times y}=\text{cov}((\boldsymbol x_i^\top,y_j)^\top)=\begin{pmatrix}
	\bm\Sigma_{xx}&\boldsymbol\sigma_{xy}\\\boldsymbol\sigma_{xy}^\top&\sigma_{yy}.
\end{pmatrix}
\]
where $\mathbf\Sigma_{xx}=\bm\Psi_{\beta\beta}+\bm\Sigma_{uu}$, $\boldsymbol\sigma_{xy}=\bm\Psi_{\beta\beta}\boldsymbol\theta+\mathbf\Sigma_{uu}\boldsymbol\theta+\boldsymbol\sigma_{uv}$, and $\sigma_{yy}=\boldsymbol\theta^\top\bm\Psi_{\beta\beta}\boldsymbol\theta+\boldsymbol\theta^\top\bm\Sigma_{uu}\boldsymbol\theta+2\boldsymbol\theta^\top\boldsymbol\sigma_{uv}+\sigma_{vv}.$ Note that $\boldsymbol\sigma_{uv}\neq\mathbf 0$ means the confounders simultaneously affect $\boldsymbol x_i$ and $y_i$.

In genetics, the genetic effect $\beta_{js}$ can be treated as a random variable with mean 0 and variance $\psi_{\beta_s\beta_s}/m$, where $\psi_{\beta_s\beta_s}$ is the IV-heritability, i.e., the variance explained by additive effects of specified instrumental variants of the $s$th exposure \citep{bulik2015ld}. Since a complex trait is often polygenic with a contribution from thousands of independent variants, the number of causal variants can be regarded as a number approaching infinity. Subject to this principle, a random effect model can describe the variation of these effects more simply and essentially. Although the fixed effect model has also been adopted by some works to study the asymptotic properties of the corresponding MR approaches \citep{zhao2020statistical,ye2021debiased}, the random effect model is still the most commonly used at genome-wide level studies. Moreover, even if all causal variants were identified, the random effect model should still be more efficient than the fixed effect model to characterize the statistical property of the MR model \citep{diggle2002analysis}.

The existing univariable MR methods, such as CAUSE \citep{morrison2020mendelian} and MR-CUE \citep{cheng2022mendelian}, can successfully remove the effects of CHP only when a small fraction of IVs have CHP, which is easy to violate because common traits may share a large fraction of pleiotropic variants \citep{bulik2015atlas}. In contrast, multivariable MR resolves the pleiotropic variant problem by specifying all the relevant exposures in the model (\ref{eq1}), as the multivariable regression can automatically account for the pleiotropic variants shared by these exposures. This is one of the greatest advantages of multivariable MR over univariable MR. Hence, we assume that all the exposures can be included in the multivariable MR; therefore, the CHP effect is ignorable in (\ref{eq1}).  On the other hand, some IVs may still present strong UHP effects. To account for potential UHP, we propose using an iterative procedure to remove these invalid IVs, which is similar to detecting outliers in MR analysis \citep{verbanck2018detection,zhu2020mendelian}. Thus, the bias introduced by UHP and CHP can be greatly alleviated, as demonstrated in our proposed MRBEE.

\subsection{Mendelian randomization with GWAS summary data}
With the rapid development of GWAS, large GWASs of exposures and disease outcomes have been conducted and their summary statistics including effect sizes, SEs, and variant information are publicly available for download \citep{sudlow2015uk,macarthur2017new}. Thus, many recently developed MR methods are often designed based on GWAS summary statistics, as is in this paper. 

With GWAS summary statistics, MR is mainly based on the linear regression model 
\begin{align}
	\hat\alpha_j=\hat{\boldsymbol\beta}_j^\top\boldsymbol\theta+\varepsilon_j,
\end{align}
where $\hat\alpha_j$ and $\hat{\boldsymbol\beta}_j$ are estimated from outcome and exposure GWAS for $j$th IV, and $\varepsilon_j$ represents the residual of this regression model. Let $\mathbf y^{[0]}=(y^{[0]}_1,\dots,y^{[0]}_{n_0})^\top$ be the sample vector from outcome GWAS, $\boldsymbol x^{[1]}=(x^{[1]}_1,\dots,x^{[1]}_{n_1})^\top,\dots,\boldsymbol x^{[p]}=(x^{[p]}_1,\dots,x^{[p]}_{n_p})^\top$ be the sample vectors of the 1st$,\dots,p$th exposure GWAS cohorts, and $\mathbf G^{[0]}=(g^{[0]}_{ij})_{n_0\times m},\dots,\mathbf G^{[p]}=(g^{[p]}_{ij})_{n_p\times m}$ be the sample matrices of $m$ genetic variants of the outcome and 1st$,\dots,p$th exposure GWAS cohorts. The sample size of the $s$th cohort is $n_s$,  the overlapping sample size between the $s$th and the $k$th cohorts is $n_{sk}$, and the minimum sample size is $n_{\rm min}=\min\{n_0,\dots,n_p\}$.

The GWAS summary data are generated as follows. Suppose that $\boldsymbol y^{[0]}$, $\{\boldsymbol x^{[s]}\}$, and $\{\mathbf G^{[s]}\}$ are centered, and the $m$ genetic variants are in LE, i.e. ${\rm E}(\mathbf G^{[s]\top}\mathbf G^{[s]}/n_s)=\mathbf I_m$ for $j=0,1,\dots,p$. This orthogonality enables the following genetic effects to be estimated separately
\begin{align}
	\hat\alpha_j=\frac{\boldsymbol g_j^{[0]\top}\boldsymbol y^{[0]}}{n_0},\quad \hat{\beta}_{js}=\frac{\boldsymbol g_j^{[s]\top}\boldsymbol x^{[s]}}{n_s},
	\label{gwasstat}
\end{align}
where the corresponding variance estimates are given by
\begin{align}
	\text{var}(\hat\alpha_j)=\frac{\sigma_{yy}-\boldsymbol\theta^\top\bm\Sigma_{\beta\beta}\boldsymbol\theta}{n_0}\approx\frac{\sigma_{yy}}{n_0},\quad	\text{var}(\hat\beta_{js})=\frac{\sigma_{x_sx_s}-\sigma_{\beta_s\beta_s}}{n_s}\approx\frac{\sigma_{x_sx_s}}{n_s}.
\end{align}
Then the GWAS summary data are formed by $\hat{\boldsymbol\alpha}=(\hat\alpha_1,\dots,\hat\alpha_m)^\top$, $\hat{\boldsymbol\beta}_j=(\hat\beta_{j1},\dots,\hat\beta_{jp})^\top$, $\hat{\mathbf B}=(\hat{\boldsymbol\beta}_1,\dots,\hat{\boldsymbol\beta}_m)^\top$, the related SE estimates, the p-values, and sample sizes $n_0,n_1,\dots,n_p$, and SNPs information.
 
The IVW method is equivalent to a weighted regression which estimates $\boldsymbol\theta$ by 
\begin{align}
	\hat{\boldsymbol\theta}_{\rm IVW}=\arg\min_{\boldsymbol\theta}\bigg\{\frac1{2m}\sum_{j=1}^m\frac{(\hat\alpha_j-\hat{\boldsymbol\beta}_{j}^\top\boldsymbol\theta)^2}{{\text{var}}(\hat\alpha_j)}\bigg\}=(\hat{\mathbf B}^\top\mathbf V\hat{\mathbf B})^{-1}\hat{\mathbf B}^\top\mathbf V\hat{\boldsymbol\alpha},
	\label{ivw1}
\end{align}
where $\mathbf V=\text{diag}(1/\text{var}(\hat\alpha_1),\dots,1/\text{var}(\hat\alpha_m))$. In practice, we often standardize $\hat\alpha_j$ and $\hat{\beta}_{js}$ by $\hat{\alpha}_j/\text{se}(\hat{\alpha}_j)$ and $\hat{\beta}_{js}/\text{se}(\hat{\beta}_{js})$ to remove the minor allele frequency effect \citep{zhu2022genome}. Therefore, ${\text{var}}(\hat\alpha_j)=1$ for all $j$ and (\ref{ivw1}) reduces to  
\begin{align}
	\hat{\boldsymbol\theta}_{\rm IVW}=\arg\min_{\boldsymbol\theta}\bigg\{\frac1{2m}\|\hat{\boldsymbol\alpha}-\hat{\mathbf B}\boldsymbol\theta\|_\text{2}^2\bigg\}=(\hat{\mathbf B}^\top\hat{\mathbf B})^{-1}\hat{\mathbf B}^\top\hat{\boldsymbol\alpha}.
	\label{ivw2}
\end{align}

Here, we qualitatively show that the IVW estimate is biased due to the estimation errors of $\hat{\boldsymbol\alpha}$ and $\hat{\mathbf B}$, i.e., ${\boldsymbol w}_\alpha=\hat{\boldsymbol\alpha}-\boldsymbol\alpha$ and ${\mathbf W}_\beta=\hat{\mathbf B}-\mathbf B$, and meanwhile, the weak IVs can inflate th estimation error bias. Specifically, consider the estimating equation and Hessian matrix of $\hat{\boldsymbol\theta}_{\rm IVW}$:
\begin{align}
	\boldsymbol{S}_{\rm IVW}(\boldsymbol\theta)=-\frac{\hat{\mathbf B}^\top(\hat{\boldsymbol\alpha}-\hat{\mathbf B}\boldsymbol\theta)}{m}=&\frac1m\bigg(-\mathbf B^\top\boldsymbol w_\alpha-\mathbf B^\top\mathbf W_\beta\boldsymbol\theta+\mathbf W_\beta^\top\boldsymbol w_\alpha-\mathbf W_\beta^\top\mathbf W_\beta\boldsymbol\theta\bigg),\\
	\textbf{H}_{\rm IVW}=\frac{\hat{\mathbf B}^\top\hat{\mathbf B}}m=&\frac1m\bigg(\mathbf B^\top\mathbf B+\mathbf W_\beta^\top\mathbf W_\beta+\mathbf B^\top\mathbf W_\beta+\mathbf W_\beta^\top\mathbf B\bigg).
\end{align}
That is, $\boldsymbol{S}_{\rm IVW}(\boldsymbol\theta)$ is the score function of (\ref{ivw2}) and $\hat{\boldsymbol\theta}_{\rm IVW}$ is estimated by solving $\boldsymbol{S}_{\rm IVW}(\hat{\boldsymbol\theta}_{\rm IVW})=\mathbf 0$, and $\textbf{H}_{\rm IVW}$ is the second order derivative matrix of (\ref{ivw2}). In particular, since the third derivative of the quadratic loss function (\ref{ivw2}) is zero, we have $\hat{\boldsymbol\theta}_{\rm IVW}-\boldsymbol\theta=-\textbf{H}_{\rm IVW}^{-1}\boldsymbol{S}_{\rm IVW}(\boldsymbol\theta)$. As a result, the expectation of the bias of $\hat{\boldsymbol\theta}_{\rm IVW}$ is approximately:
\begin{align}
	{\rm E}(\hat{\boldsymbol\theta}_{\rm IVW}-\boldsymbol\theta)&\approx-{\rm E}(\textbf{H}_{\rm IVW})^{-1}{\rm E}(\boldsymbol{S}_{\rm IVW}(\boldsymbol\theta))\notag\\&=-\underbrace{\bigg\{\mathbf\Sigma_{\beta\beta}+\bm\Sigma_{W_\beta W_\beta}\bigg\}^{-1}}_{\text{weak instrument bias}}\underbrace{\bigg\{\bm\Sigma_{W_\beta W_\beta}\boldsymbol\theta-\boldsymbol\sigma_{W_\beta w_\alpha}\bigg\}}_\text{estimation error bias},
	\label{biasivw}
\end{align}
where ${\boldsymbol w}_{\beta_j}$ is the $j$th row of ${\mathbf W}_\beta$, $w_{\alpha_j}$ is the $j$th element of ${\boldsymbol w}_\alpha$, and
\[
\text{cov}(({\boldsymbol w}_{\beta_j}^\top,w_{\alpha_j})^\top)=\bm\Sigma_{W_\beta \times w_\alpha}=\begin{pmatrix}
	\bm\Sigma_{W_\beta W_\beta}&\boldsymbol\sigma_{W_\beta w_\alpha}\\
	\boldsymbol\sigma_{W_\beta w_\alpha}^\top&\sigma_{w_\alpha w_\alpha}.
\end{pmatrix},
\]
Intuitively, the bias of $\hat{\boldsymbol\theta}_{\rm IVW}$ has a product structure ``weak instrument bias $\times$ estimation error bias". We call $\{\bm\Sigma_{W_\beta W_\beta}\boldsymbol\theta-\boldsymbol\sigma_{W_\beta w_\alpha}\}$ the estimation error bias because it comes from the covariance matrix of estimation errors $\bm\Sigma_{W_\beta \times w_\alpha}$. We term $\{\mathbf\Sigma_{\beta\beta}+\bm\Sigma_{W_\beta W_\beta}\}$ the weak instrument bias because the bias of $\hat{\boldsymbol\theta}_{\rm IVW}$ is inflated if the covariance matrix of effect sizes $\mathbf\Sigma_{\beta\beta}$ is not considerably larger than the covariance matrix of estimation errors  $\bm\Sigma_{W_\beta W_\beta}$, which often happens if the majority of IVs used to infer the causal effect have weak effects.
\subsection{Asymptotic behavior of IVW estimate}
In this subsection, we investigate the asymptotic behavior of the IVW estimate as the number of IVs $m$ and the minimum sample size $n_{\min }$ go to infinity. To facilitate the theoretical derivation, we specify the following three definitions and four regularity conditions.
\begin{definition}[Sub-Gaussian variable]
	A random variable $x$ is sub-Gaussian distributed with sub-Gaussian parameter $\tau_x>0$ if for all $t>0$, $\Pr(|x-E(x)|\geq t)\leq 2e^{-t^2/\tau_x^2}$.
\end{definition}
\begin{definition}[Well-conditioned covariance matrix]
	A covariance matrix $\bm\Sigma$ is well-conditioned if there is a positive constant $d_0$ such that $
	0<d_0^{-1}\leq\lambda_{\min }(\bm\Sigma)\leq\lambda_{\max }(\bm\Sigma)\leq d_0<\infty.$
\end{definition}
\begin{definition}[Strongly asymptotically unbiased estimate] Let  $\hat{\boldsymbol\theta}$ be a consistent estimate of $\boldsymbol\theta$ with an asymptotic normal distribution $\sqrt s_n(\hat{\boldsymbol\theta}-\boldsymbol\theta)\stackrel{D}{\longrightarrow}\mathcal{N}(\boldsymbol\mu_\theta,\bm\Sigma_\theta)$,
	where $\boldsymbol\mu_\theta$ is a vector with a bounded $\ell_2$-norm, $\bm\Sigma_\theta$ is a well-conditioned covariance matrix, and $s_n$ is a sequence of $n$. Then $\hat{\boldsymbol\theta}$ is called a strongly asymptotically unbiased estimate of $\boldsymbol\theta$ if $\boldsymbol\mu_\theta=\mathbf0$.
\end{definition}
A sub-Gaussian variable is one of the basic concepts in modern statistics \citep{vershynin2018high}. It generalizes the scope of ordinary Gaussian variables to include all bounded discrete and common continuous variables. The well-conditioned covariance matrix is another important concept \citep{bickel2008regularized}. A well-conditioned covariance matrix will ensure that the related statistical optimization is nondegenerate. In addition, we define the strongly asymptotic unbiasedness to distinguish the consistent estimate whose bias square vanishes with an equal and a smaller rate than its variance, respectively. If an estimate is consistent but its bias square and variance vanish at the same rate, the classic confidence interval cannot cover the true parameter with a probability of 0.95, thus leading to invalid statistical inference. This problem widely exists in all fields of statistics, especially, in nonparametric statistics and high-dimensional statistics, and many novel methods are derived to reduce the bias such that the bias square vanishes faster than the variance \citep{hall1992effect,van2014asymptotically,jankova2018semiparametric,calonico2018effect}.
\begin{condition}[Regularity conditions for multivariable MR]~
	\begin{itemize}
		\item[(C1)] For $\boldsymbol g_i=(g_{i1},\dots,g_{im})^\top$, each entry $g_{ij}$ is a bounded sub-Gaussian variable with E$(g_{ij})=0$, var($g_{ij}$)=1, and sub-Gaussian parameter $\tau_g\in(0,\infty)$. For all $(i,j)\neq(t,s)$, $g_{ij}$ is independent of $g_{ts}$. 
		\item[(C2)] For $\boldsymbol u_i=(u_{i1},\dots,u_{ip})^\top$, each entry $u_{ij}$ is a sub-Gaussian variable with E$(u_{ij})=0$, var($u_{is})\in(0,\infty)$, and sub-Gaussian parameter $\tau_u\in(0,\infty)$; $v_i$ is a sub-Gaussian variable with E$(v_i)=0$, var($v_i)\in(0,\infty)$, and sub-Gaussian parameter $\tau_v\in(0,\infty)$. Besides, $(\boldsymbol u_i^\top,v_i)^\top$ is independent of $(\boldsymbol u_t^\top,v_t)^\top$  for all $i\neq t$. Furthermore, $\mathbf\Sigma_{u\times v}$ is a well-conditioned covariance matrix of $(\boldsymbol u_i^\top,v_i)^\top$. 
		\item[(C3)] For $\boldsymbol \beta_j=(\beta_{j1},\dots,\beta_{jp})^\top$, $\sqrt m\beta_{js}$ is a sub-Gaussian variable with E$(\sqrt m\beta_{js})=0$, var($\sqrt m\beta_{js})\in(0,\infty)$, and sub-Gaussian parameter $\tau_\beta\in(0,\infty)$. For all $j\neq t$, $\boldsymbol\beta_j$ is independent of $\boldsymbol\beta_t$. In addition, $\bm\Psi_{\beta\beta}$ is a well-conditioned covariance matrix of $\sqrt m\boldsymbol\beta_j$.
		\item[(C4)] The genetic variant $g_{ij}$, the genetic effect $\boldsymbol\beta_j$, the noise terms $\boldsymbol u_i$ and $v_i$, are three mutually independent groups. \end{itemize}
\end{condition}

Conditions (C1)-(C4) restrict that all variables involved in this paper are sub-Gaussian distributed. In practice, $g_{ij}$ is standardized from a binomial variable with status 0, 1, and 2. Hence, it is supposedly a bounded sub-Gaussian variable as long as its minor allele frequency is not rare. Besides, we assume $\sqrt m\boldsymbol \beta_j$ to be sub-Gaussian with a well-conditioned covariance matrix $\bm\Psi_{\beta\beta}$ because the covariance explained by each variant $\bm\Sigma_{\beta\beta}$ decreases as the number of instrumental variants $m$ increases. 
\begin{theorem}
Denote $w_{\alpha_j}=\hat\alpha_j-\alpha_j$ and $\omega_{js}=\hat\beta_{js}-\beta_{js}$, $s=1,\dots,p$. Then for all $j$,
\[
\begin{pmatrix}
	\sqrt n_0w_{\alpha_j}\\
	\sqrt n_1w_{\beta_{1j}}\\
	\vdots\\
	\sqrt n_pw_{\beta_{1p}}
\end{pmatrix}
\stackrel{D}{\longrightarrow}\mathcal{N}\begin{pmatrix}
\begin{pmatrix}
	0\\
	0\\
	\vdots\\
	0
\end{pmatrix},
\begin{pmatrix}
\sigma_{yy}&\frac{n_{01}}{\sqrt{n_0n_1}}\sigma_{yx_1}&\cdots&\frac{n_{01}}{\sqrt{n_0n_p}}\sigma_{yx_p}\\
\frac{n_{01}}{\sqrt{n_0n_1}}\sigma_{yx_1}&\sigma_{x_1x_1}&\cdots&\frac{n_{1p}}{\sqrt{n_1n_p}}\sigma_{x_1x_p}\\
\vdots&\vdots&\ddots&\vdots\\
\frac{n_{0p}}{\sqrt{n_0n_p}}\sigma_{yx_p}&	\frac{n_{1p}}{\sqrt{n_1n_p}}\sigma_{x_1x_p}&\cdots&\sigma_{x_px_p}	
\end{pmatrix}
\end{pmatrix},
\]
if $n_0,\dots,n_p$ and $m\to\infty$.
\label{theorem1}
\end{theorem}
Theorem \ref{theorem1} demonstrates the asymptotic normal distribution of the estimation errors, based on which we are able to obtain
\begin{align}
\bm\Sigma_{W_\beta W_\beta}=\mathbf\Delta_{xx}\odot\bm\Sigma_{xx},\quad \boldsymbol\sigma_{W_\beta w_\alpha}=\boldsymbol\delta_{xy}\odot\boldsymbol\sigma_{xy},\quad \sigma_{w_\alpha w_\alpha}=\sigma_{yy}/n_0,
\end{align}
where the $(j,s)$th element of $\mathbf\Delta_{xx}$ is $n_{js}/(n_jn_s)$ and the $j$th element of $\boldsymbol\delta_{xy}$ is $n_{j0}/(n_0n_j)$. As a result, the expectations of $\boldsymbol S_{\rm IVW}(\boldsymbol\theta)$ and $\mathbf H_{\rm IVW}$ are given by
\begin{align}
{\rm E}(\boldsymbol S_{\rm IVW}(\boldsymbol\theta))&=(\mathbf\Delta_{xx}\odot\bm\Sigma_{xx})\boldsymbol\theta-\boldsymbol\delta_{xy}\odot\boldsymbol\sigma_{xy},\\
{\rm E}(\mathbf H_{\rm IVW})&=\mathbf\Sigma_{\beta\beta}+\mathbf\Delta_{xx}\odot\bm\Sigma_{xx}.
\end{align}
By expressing $\boldsymbol\sigma_{xy}=\mathbf\Sigma_{xx}\boldsymbol\theta+\boldsymbol\sigma_{uv}$,
we obtain an alternative expectation of $\boldsymbol S_{\rm IVW}(\boldsymbol\theta))$:
\begin{align}
	\underbrace{\text{E}(\boldsymbol{S}_{\rm IVW}(\boldsymbol\theta))}_{\text{estimation bias}}=\underbrace{\{(\mathbf\Delta_{xx}-\boldsymbol\delta_{xy}\mathbf 1^\top)\odot\bm\Sigma_{xx}\}\boldsymbol\theta}_{\text{measurement error bias}}-\underbrace{\boldsymbol\delta_{xy}\odot\boldsymbol\sigma_{uv}}_\text{confounder bias}.
\end{align}
From this expectation, it is clear that there are two sources of the estimation error bias: $\{(\mathbf\Delta_{xx}-\boldsymbol\delta_{xy}\mathbf 1^\top)\odot\bm\Sigma_{xx}\}\boldsymbol\theta$ comes from the measurement error, while $\{\boldsymbol\delta_{xy}\odot\boldsymbol\sigma_{uv}\}$ is caused by the confounder. Here, we call  $\{(\mathbf\Delta_{xx}-\boldsymbol\delta_{xy}\mathbf 1^\top)\odot\bm\Sigma_{xx}\}\boldsymbol\theta$ the measurement error bias because it has the same statistical impact, i.e., shrinking the coefficient estimate toward zero, as in measurement error analysis \citep{yi2017statistical}. In contrast, we term  $\{\boldsymbol\delta_{xy}\odot\boldsymbol\sigma_{uv}\}$ the confounder bias because $\boldsymbol\sigma_{uv}\neq\mathbf0$ implies that there are underlying confounders simultaneously affecting both $\boldsymbol x_i$ and $y_i$.  In addition, the overlapping fraction vector $\boldsymbol\delta_{xy}$ trades off these two sources of biases. Generally, the measurement error bias is dominant when the elements of $\boldsymbol\delta_{xy}$ are small, while the confounder bias dominates when the elements of $\boldsymbol\delta_{xy}$ are large, and there may exist a special sample overlap such that $\boldsymbol\delta_{xy}\odot\boldsymbol\sigma_{uv}=\{(\mathbf\Delta_{xx}-\boldsymbol\delta_{xy}\mathbf 1^\top)\odot\bm\Sigma_{xx}\}\boldsymbol\theta$. In univariable MR, this special fraction is $n_{01}/n_0=\sigma_{xx}\theta/\sigma_{xy}$, which guarantees that E$(S_{\rm IVW}(\theta))=0$ and E$(\hat\theta_{\rm IVW})=\theta$. This theoretical result explains why in the empirical studies (e.g., Figures 1 and 2 in \citet{sadreev2021navigating}), $\hat{\theta}_{\rm IVW}$ has a negative bias when $n_{01}/n_0$ is small, positive bias when  $n_{01}/n_0$ is large, and is unbiased at this specific point.
\begin{theorem}
	Suppose conditions (C1)-(C4) hold and $m$, $n_{\rm min}\to\infty$. Then 
	\begin{enumerate}
		\item[(i)] if $m/\sqrt n_{\min }\to0$, 
		$\sqrt n_{\min }(\hat{\boldsymbol\theta}_{\rm IVW}-\boldsymbol\theta)
		\stackrel{D}{\longrightarrow}\mathcal N(\mathbf0,\psi_\theta\bm\Psi_{\beta\beta}^{-1});$
		\item[(ii)] if $m/\sqrt n_{\min }\to c_0$,
		$\sqrt n_{\min }(\hat{\boldsymbol\theta}_{\rm IVW}-\boldsymbol\theta)
		\stackrel{D}{\longrightarrow}\mathcal N(-c_0\bm\Psi_{\beta\beta}^{-1}(\bm\Psi_{W_\beta W_\beta}\boldsymbol\theta-\boldsymbol\psi_{W_\beta w_\alpha}),\psi_\theta\bm\Psi_{\beta\beta}^{-1});$
		\item[(iii)] if $m/n_{\min }\to c_0$, 
	$\hat{\boldsymbol\theta}_{\rm IVW}-\boldsymbol\theta\stackrel{P}{\longrightarrow}-c_0(\bm\Psi_{\beta\beta}+c_0\bm\Psi_{W_\beta W_\beta})^{-1}(\bm\Psi_{W_\beta W_\beta}\boldsymbol\theta-\boldsymbol\psi_{W_\beta w_\alpha});$
		\item[(iv)] if $m/n_{\rm min}\to\infty$, 
		$\hat{\boldsymbol\theta}_{\rm IVW}\stackrel{P}{\longrightarrow}\bm\Psi_{W_\beta W_\beta}^{+}\boldsymbol\psi_{W_\beta w_\alpha};$
	\end{enumerate}
where 
\[
	\mathbf\Psi_{W_\beta \times w_\alpha}=\begin{pmatrix}
		\mathbf\Psi_{W_\beta W_\beta}&\boldsymbol\psi_{W_\beta w_\alpha}\\
		\boldsymbol\psi_{W_\beta w_\alpha}^\top&\psi_{w_\alpha w_\alpha}
	\end{pmatrix}=\lim_{n_{\min }\to\infty}\begin{pmatrix}
		n_{\min }\bm\Sigma_{W_\beta W_\beta}&n_{\min }\boldsymbol\sigma_{W_\beta w_\alpha}\\
		n_{\min }\boldsymbol\sigma_{W_\beta w_\alpha}^\top&n_{\min }\sigma_{w_\alpha w_\alpha}
	\end{pmatrix},
\]
$\psi_\theta=\psi_{w_\alpha w_\alpha}+\boldsymbol\theta^\top\bm\Psi_{W_\beta W_\beta}\boldsymbol\theta-2\boldsymbol\theta^\top\boldsymbol\psi_{W_\beta w_\alpha}$, and $c_0$ is a positive constant.
\label{theorem2}
\end{theorem}
Theorem \ref{theorem2} is one of two main theorems in this paper and points out four scenarios. First, if  $m$ goes to infinity with a lower rate than $\sqrt{n}_{\rm min}$, $\hat{\boldsymbol\theta}_{\rm IVW}$ is strongly asymptotically unbiased. In other words, $\hat{\boldsymbol\theta}_{\rm IVW}$ is able to reliably infer causality only when the sample size of GWAS data is quadratically larger than the number of IVs. On the other hand, the asymptotic covariance matrix of $\hat{\boldsymbol\theta}_{\rm IVW}$ is the inverse of the cumulative covariance matrix $\mathbf\Psi_{\beta\beta}=\sum_{j=1}^m\text{cov}(\boldsymbol\beta_j)$, therefore, it is optimal to include as many associated variants as possible in order to have $\mathbf\Psi_{\beta\beta}$ large enough. In contrast, using a few top significant variants to perform MR analysis is not recommended. 

Second, if $m$ tends to infinity with the same rate as $\sqrt{n}_{\rm min}$, $\sqrt n_{\min }(\hat{\boldsymbol\theta}_{\rm IVW}-\boldsymbol\theta)$ converges to an asymptotic normal distribution with a non-zero asymptotic bias $\{-c_0\bm\Psi_{\beta\beta}^{-1}(\bm\Psi_{W_\beta W_\beta}\boldsymbol\theta-\boldsymbol\psi_{W_\beta w_\alpha})\}$. In this asymptotic bias, $\{-c_0(\bm\Psi_{W_\beta W_\beta}\boldsymbol\theta-\boldsymbol\psi_{W_\beta w_\alpha})\}$ is caused by $\boldsymbol S_{\rm IVW}(\boldsymbol\theta)$ and $\bm\Psi_{\beta\beta}^{-1}$ is resulted by $\mathbf H_{\rm IVW}^{-1}$. Since the asymptotic bias and asymptotic covariance matrix are of the same order in this scenario, the inference made is invalid although the bias of $\hat{\boldsymbol\theta}_{\rm IVW}$ is infinitesimal. Scenario $(iii)$ is more serious than $(ii)$ because the bias of $\hat{\boldsymbol\theta}_{\rm IVW}$ will not vanish even when $\sqrt n_{\rm min}$ goes to infinity. In the fourth scenario, $\hat{\boldsymbol\theta}_{\rm IVW}$ converges to a term irrelevant to $\boldsymbol\theta$. Scenarios (ii) - (iv) indicate that the IVW method is unlikely to make valid causal inference unless the sample sizes are quadratically larger than the number of IVs.

It is crucial to understand the asymptotic behaviors of $\hat{\boldsymbol\theta}_{\rm IVW}$ since the IVW method serves as the foundation for practically all MR techniques. Specifically, IMRP and MR-PRESSO use hypothesis tests to identify invalid IVs and then apply the IVW method to estimate causal effects based on valid IVs only. MR-Robust and MR-Median replace the quadratic loss function used in IVW by a robust loss function and absolute loss function, respectively. Although there have been literature studying the bias of $\hat{\boldsymbol\theta}_{\rm IVW}$ empirically \citep{burgess2011avoiding,burgess2016bias}, they could not explain what causes the bias and how it behaves asymptotically. In contrast, Theorem \ref{theorem2} points out the asymptotic properties of $\hat{\boldsymbol\theta}_{\rm IVW}$, representing a significant advance in understanding the IVW method and its extensions. 
\section{Bias-corrected Estimating Equation}
According to (\ref{biasivw}), it is possible to remove the bias of $\boldsymbol S_{\rm IVW}(\boldsymbol\theta)$ by subtracting the measurement error bias $\{\bm\Sigma_{W_\beta W_\beta}\boldsymbol\theta-\boldsymbol\sigma_{W_\beta w_\alpha}\}$. Motivated by this principle, we propose MRBEE that estimates the causal effect estimates by solving the new unbiased estimating equation. 
In this section, we introduce the estimation of MRBEE, investigate its asymptotic properties, and discuss three implementation issues including the estimations of the bias-correction terms, the estimation of sandwich formula of causal effect estimate, and the detection of potential pleiotropy.
\subsection{Estimation of causal effect}
There are many methods that can remove the measurement error bias, including maximum likelihood estimation, unbiased estimating functions, and simulation-extrapolation (SIMEX) methods; see, e.g., \citet{yi2017statistical}. MRBEE is a subtraction correction method belonging to the class of unbiased estimating function methods. Specifically, MRBEE estimates $\boldsymbol\theta$ by solving the following unbiased estimating equation:
\begin{align}
	\boldsymbol{S}_{\rm BEE}(\boldsymbol\theta)=\boldsymbol{S}_{\rm IVW}(\boldsymbol\theta)-(\bm\Sigma_{W_\beta W_\beta}\boldsymbol\theta-\boldsymbol\sigma_{W_\beta w_\alpha}),
\end{align}
where $\boldsymbol{S}_{\rm IVW}(\boldsymbol\theta)=-\hat{\mathbf B}^\top(\hat{\boldsymbol\alpha}-\hat{\mathbf B}\boldsymbol\theta)/m$. 
The solution  $\hat{\boldsymbol\theta}_{\rm BEE}$ such that  $\boldsymbol{S}_{\rm BEE}(\hat{\boldsymbol\theta}_{\rm BEE})=\mathbf0$ is 
\begin{align}
	\hat{\boldsymbol\theta}_{\rm BEE}=\bigg\{\frac{\hat{\mathbf B}^\top\hat{\mathbf B}}{m}-\bm\Sigma_{W_\beta W_\beta}\bigg\}^{-1}\bigg\{\frac{\hat{\mathbf B}^\top\hat{\boldsymbol\alpha}}m-\boldsymbol\sigma_{W_\beta w_\alpha}\bigg\}.
\end{align}
In practice, $\hat{\boldsymbol\theta}_{\rm BEE}$ is unreliable when the minimum eigenvalue of $\hat{\mathbf B}^\top\hat{\mathbf B}/m-\bm\Sigma_{W_\beta W_\beta}$ is negative, which is also a common problem for subtraction correction methods. In this case, we recommend first adjusting the negative eigenvalues to be 0 and then using the generalized inverse of this semi-positive matrix to yield $\hat{\boldsymbol\theta}_{\rm BEE}$.
\begin{theorem}
	Suppose conditions (C1)-(C4) hold and $m$, $n_{\rm min}\to\infty$. Then 
	\begin{enumerate}
		\item[(i)] if $m/n_{\min }\to0$, $\sqrt n_{\rm min}(\hat{\boldsymbol\theta}_{\rm BEE}-\boldsymbol\theta)\stackrel{D}{\longrightarrow}\mathcal N(\mathbf 0,\psi_\theta\bm\Psi_{\beta\beta}^{-1});$
		\item[(ii)] if $m/n_{\min }\to c_0$,
		$\sqrt n_{\rm min}(\hat{\boldsymbol\theta}_{\rm BEE}-\boldsymbol\theta)\stackrel{D}{\longrightarrow}\mathcal N(\mathbf 0,\psi_\theta\bm\Psi_{\beta\beta}^{-1}+c_0\bm\Psi_{\beta\beta}^{-1}\bm\Psi_{\rm BC}\bm\Psi_{\beta\beta}^{-1});$
		\item[(iii)] if $m/n_{\rm min}\to\infty$ and $m/n_{\min }^2\to0$, $\sqrt{n_{\rm min}^2/m}(\hat{\boldsymbol\theta}_{\rm BEE}-\boldsymbol\theta)\stackrel{D}{\longrightarrow}\mathcal N(\mathbf 0,\bm\Psi_{\beta\beta}^{-1}\bm\Psi_{\rm BC}\bm\Psi_{\beta\beta}^{-1});$
	\end{enumerate}
	where $\psi_\theta$ is defined in Theorem \ref{theorem2}, $c_0$ is a positive constant, and $\bm\Psi_{\rm BC}$ is a semi-positive symmetric matrix whose expression is shown in equation (\ref{SigmaBC}). 
	\label{theorem3}
\end{theorem}
Theorem \ref{theorem3} indicates the following three scenarios. First, if $m/n\to0$,$\sqrt n_{\rm min}(\hat{\boldsymbol\theta}_{\rm BEE}-\boldsymbol\theta)$ converges to a normal distribution with a zero mean and the covariance matrix being exactly the same as $\hat{\boldsymbol\theta}_{\rm IVW}$. In other words, $\hat{\boldsymbol\theta}_{\rm BEE}$ not only enjoys the strongly asymptotic unbiasedness but also loses no efficiency in comparison to $\hat{\boldsymbol\theta}_{\rm IVW}$. Second, if $m/n_{\min }\to c_0\in(0,\infty)$, there is an additional covariance matrix  $c_0\bm\Psi_{\beta\beta}^{-1}\bm\Psi_{\rm BC}\bm\Psi_{\beta\beta}^{-1}$ in the asymptotic normal distribution, where $\bm\Psi_{\rm BC}$ is introduced by the bias-correction terms:
\[
\bm\Psi_{\rm BC}=\lim_{n_{\rm min}\to\infty}\text{var}\bigg[\frac{n_{\rm min}}{\sqrt m}\bigg((\mathbf W_\beta^\top\mathbf W_\beta-m\mathbf\Sigma_{W_\beta W_\beta})\boldsymbol\theta-(\mathbf W_\beta^\top\boldsymbol w_\alpha-m\boldsymbol\sigma_{W_\beta w_\alpha})\bigg)\bigg].
\]
In this scenario, $\hat{\boldsymbol\theta}_{\rm BEE}$ is again strongly asymptotically unbiased with a convergence rate $\sqrt n_{\rm min}$, while $\hat{\boldsymbol\theta}_{\rm IVW}$ suffers from a bias not vanishing asymptotically. In the third scenario, $\hat{\boldsymbol\theta}_{\rm BEE}$ is still strongly asymptotically unbiased with a convergence rate $\sqrt{n_{\rm min}^2/m}$, and the asymptotic distribution is dominated by the bias correction term. In contrast,  $\hat{\boldsymbol\theta}_{\rm IVW}$ converges to a term irrelevant to $\boldsymbol\theta$. Note that $\hat{\boldsymbol\theta}_{\rm IVW}$ is not consistent unless $m/n\to0$ and the inference made by $\hat{\boldsymbol\theta}_{\rm IVW}$ is unreliable unless $m/\sqrt n_{\rm min}\to0$. Therefore, MRBEE is superior to IVW in terms of both unbiasedness and asymptotic validity.

Most previous works of MR introduced their methods from the perspective of empirical applications and have not discussed the asymptotic properties; see, e.g., \citet{bowden2015mendelian,bowden2016consistent,verbanck2018detection,morrison2020mendelian}. Some works \citep{zhao2020statistical,ye2021debiased} described the asymptotic behaviors of the causal effect estimates yielded by their  univariate MR methods, but the convergence rates and related conditions were not straightforward. For example, \citet{zhao2020statistical} showed that the convergence rate of their causal effect estimate is $O(V_1/\sqrt{V_2})$ where $V_1$ and $V_2$ are two $m$-concentrations, which may mislead that this estimate has a $O(\sqrt m)$ convergence rate. From Theorem \ref{theorem3}, it is easy to see that $\hat{\boldsymbol\theta}_{\rm BEE}$ is strongly asymptotically unbiased, the asymptotic covariance matrix is $\psi_\theta\bm\Psi_{\beta\beta}^{-1}$, $\psi_\theta\bm\Psi_{\beta\beta}^{-1}+c_0\bm\Psi_{\beta\beta}^{-1}\bm\Psi_{\rm BC}\bm\Psi_{\beta\beta}^{-1}$, and $\bm\Psi_{\beta\beta}^{-1}\bm\Psi_{\rm BC}\bm\Psi_{\beta\beta}^{-1}$, and the convergence rate is $\sqrt n_{\rm min}$, $\sqrt n_{\rm min}$, and $\sqrt{n_{\min}^2/m}$, with respect to scenarios $(i)$, $(ii)$, and $(iii).$ In addition, although our method focuses on the multivariable MR model, the theoretical results can be readily extended to the univariable MR model. To the best of our knowledge, this is the first theoretical work to demonstrate how the convergence rate and asymptotic normal distributions vary with the sample sizes of multiple GWAS cohorts and the number of IVs for univariable and multivariable MR.
\subsection{Estimation of bias-correction terms}
In this subsection, we discuss how to estimate the bias-correction terms $\bm\Sigma_{W_\beta W_\beta}$ and $\boldsymbol\sigma_{W_\beta w_\alpha}$ in practice. Specifically, we apply the method provided by \citet{zhu2015meta} to estimate the covariance matrix $\bm\Sigma_{W_\beta\times w_\alpha}$ of the vector $(\boldsymbol w_{\beta_j}^\top,w_{\alpha_j})^\top$ from insignificant GWAS summary statistics. Let $\mathbf G^{\{0\}}=(g^{\{0\}}_{ij})_{n_1\times M},\dots,\mathbf G^{\{p\}}=(g^{\{p\}}_{ij})_{n_s\times M}$ be the sample matrices of $M$ insignificant and independent genetic variants. The insignificance means that the $p$-value of the genetic variants are larger than 0.05 for all exposures and outcome, and independence means that these variants are in LE. The insignificant GWAS statistics are estimated by
\begin{align}
	\hat\alpha_j^*=\frac{\boldsymbol g_j^{\{0\}\top}\boldsymbol y^{[0]}}{n_0},\quad \hat{\beta}^*_{js}=\frac{\boldsymbol g_j^{\{s\}\top}\boldsymbol x^{[s]}}{n_s},
	\label{gwasinstat}
\end{align}
for $s=1,\dots,p$. With these insignificant effect sizes, $\bm\Sigma_{W_\beta \times w_\alpha}$ can be estimated by 
\begin{align}
\hat{\bm\Sigma}_{W_\beta \times w_\alpha}=\frac1M\sum_{j=1}^M(\hat{\beta}^*_{j1},\dots,\hat{\beta}^*_{jp},\hat\alpha^*_{j})^\top(\hat{\beta}^*_{j1},\dots,\hat{\beta}^*_{jp},\hat\alpha^*_{j}),
\label{gwascorrelation}
\end{align}
because $\hat\alpha_j^*$ and $\hat{\beta}^*_{js}$ follow the same distributions of $w_{\alpha_j}$ and $w_{\beta_{js}}$, respectively. Here, $\hat{\bm\Sigma}_{W_\beta W_\beta}$ is the first $(p\times p)$ sub-matrix of $\hat{\bm\Sigma}_{W_\beta \times w_\alpha}$ and $\boldsymbol\sigma_{W_\beta w_\alpha}$ consists of the first $p-1$ elements of the last column of $\hat{\bm\Sigma}_{W_\beta \times w_\alpha}$.
\begin{theorem}
	Suppose conditions (C1)-(C4) hold. Let $g_{ij}^{\{s\}}$ satisfy the condition (C1), E$(x_{i}^{[s]}|g_{ij}^{\{s\}})=0$ for all $1\leq s\leq p$, and E$(y_i^{[0]}|g_{ij}^{\{0\}})=0$. Then
	\[
	\|\bm\Sigma_{W_\beta \times w_\alpha}^{-\frac12}\hat{\bm\Sigma}_{W_\beta \times w_\alpha}\bm\Sigma_{W_\beta \times w_\alpha}^{-\frac12}-\mathbf I_{p+1}\|_2=O_P\bigg(\frac1{\sqrt M}\bigg),
	\]
	if  $n_{\rm min}$ and $M\to\infty$.
	\label{theorem4}
\end{theorem}
Theorem \ref{theorem4} shows that $\hat{\bm\Sigma}_{W_\beta \times w_\alpha}$ has a $O(\sqrt M)$ convergence rate after adjusting the scale of $\bm\Sigma_{W_\beta \times w_\alpha}$. As there may be more than 1 million independent variants in the whole genome,  $\hat{\bm\Sigma}_{W_\beta \times w_\alpha}$ has high precision. In addition, $n_0,n_1,...,n_p\to\infty$ are required such that  $\sqrt{n_0}\hat\alpha_j^*$ and $\sqrt{n_s}\hat{\beta}^*_{js}$ are asymptotically normally distributed. In addition, many popular GWAS methods such as cross-phenotype association analysis (CPASSOC, \citet{zhu2015meta}) and multi-trait analysis of GWAS (MTAG, \citet{turley2018multi}) need to estimate the covariance matrix of the estimation errors of GWAS summary statistics. As far as we are concerned, this theorem is the first one to theoretically guarantee that this covariance matrix can be consistently estimated from the GWAS insignificant statistics.
\subsection{Estimation of sandwich formula}
In this subsection, we illustrate how to estimate the covariance matrix of $\hat{\boldsymbol\theta}_{\rm BEE}$, i.e., cov($\hat{\boldsymbol\theta}_{\rm BEE}$)=$\bm\Sigma_{\rm BEE}(\boldsymbol\theta)$, through the famous sandwich formula \citep{liang1986longitudinal}:
\begin{align}
\bm\Sigma_{\rm BEE}(\boldsymbol\theta)=\textbf{F}_{\rm BEE}^{-1}\textbf{V}_{\rm BEE}(\boldsymbol\theta)\textbf{F}_{\rm BEE}^{-1}.
\end{align}
Here, the outer matrix $\textbf{F}_{\rm BEE}$ is the Fisher information matrix, i.e., the expectation of the Hessian matrix of $\boldsymbol{S}_{\rm BEE}(\boldsymbol\theta)$:
\begin{align}
\textbf{F}_{\rm BEE}=-\text{E}\bigg\{\frac{\partial \boldsymbol{S}_{\rm BEE}(\boldsymbol\theta)}{\partial \boldsymbol\theta^\top}\bigg\}=\bm\Sigma_{\beta\beta}.
\end{align}
The inner matrix $\textbf{V}_{\rm BEE}(\boldsymbol\theta)$ is the covariance matrix of $\boldsymbol{S}_{\rm BEE}(\boldsymbol\theta)$:
\begin{align}
\textbf{V}_{\rm BEE}(\boldsymbol\theta)=\text{E}\bigg\{\frac1{m}\sum_{j=1}^{m}\boldsymbol{S}_{j}(\boldsymbol\theta)\boldsymbol{S}_{j}(\boldsymbol\theta)^\top\bigg\},
\end{align}
where
\begin{align}
\boldsymbol{S}_{j}(\boldsymbol\theta)=-(\hat\alpha_j-\boldsymbol\theta^\top\hat{\boldsymbol\beta}_j)\hat{\boldsymbol\beta}_j-\bm\Sigma_{W_\beta W_\beta}\boldsymbol\theta+\boldsymbol\sigma_{W_\beta w_\alpha}.
\end{align}
A consistent estimate of $\bm\Sigma_{\rm BEE}(\boldsymbol\theta)$ is
\begin{align}
\hat{\bm\Sigma}_{\rm BEE}(\hat{\boldsymbol\theta}_{\rm BEE})=\hat{\textbf{F}}_{\rm BEE}^{-1}\hat{\textbf{V}}_{\rm BEE}(\hat{\boldsymbol\theta}_{\rm BEE})\hat{\textbf{F}}_{\rm BEE}^{-1},\end{align}
where
\begin{align}
\hat{\textbf{F}}_{\rm BEE}=\frac{\hat{\mathbf B}^\top\hat{\mathbf B}}m-\hat{\bm\Sigma}_{W_\beta W_\beta},\quad\hat{\textbf{V}}_{\rm BEE}(\hat{\boldsymbol\theta}_{\rm BEE})=\frac1{m}\sum_{j=1}^{m}\hat{\boldsymbol{S}}_{j}(\hat{\boldsymbol\theta}_{\rm BEE})\hat{\boldsymbol{S}}_{j}(\hat{\boldsymbol\theta}_{\rm BEE})^\top\notag\\
\hat{\boldsymbol{S}}_j(\hat{\boldsymbol\theta}_{\rm BEE})=-(\hat\alpha_j-\hat{\boldsymbol\theta}_{\rm BEE}^\top\hat{\boldsymbol\beta}_j)\hat{\boldsymbol\beta}_j-\hat{\bm\Sigma}_{W_\beta W_\beta}\hat{\boldsymbol\theta}_{\rm BEE}+\hat{\boldsymbol\sigma}_{W_\beta w_\alpha},\quad\quad
\end{align}
and $\hat{\bm\Sigma}_{W_\beta W_\beta}$ and $\hat{\boldsymbol\sigma}_{W_\beta w_\alpha}$ are estimated through (\ref{gwascorrelation}).
\begin{theorem}
Under the conditions of Theorem \ref{theorem4},
\[||\bm\Sigma_{\rm BEE}^{-\frac12}(\boldsymbol\theta)\hat{\bm\Sigma}_{\rm BEE}(\hat{\boldsymbol\theta}_{\rm BEE})\bm\Sigma^{-\frac12}_{\rm BEE}(\boldsymbol\theta)-\mathbf I_p||_2=O_P\bigg(\max\bigg\{\frac1{\sqrt n_{\rm min}},\frac{\sqrt m}{n_{\rm min}},\sqrt{\frac{\log m} m}\bigg\}\bigg)\]
if $n_{\rm min},m$ and $M\to\infty$ and $m/n_{\rm min}^2\to0$.
	\label{theorem5}
\end{theorem}
Theorem \ref{theorem5} shows that $\hat{\bm\Sigma}_{\rm BEE}(\boldsymbol\theta)$ has a $\min(\sqrt n_{\rm min},\sqrt{n^2_{\rm min}/m},\sqrt{m/\log m})$ convergence rate when $m/n^2_{\rm min}\to 0$. The first two convergence rates are brought by $||\hat{\mathbf F}_{\rm BEE}-\mathbf F_{\rm BEE}||_2$, while the third convergence rate is yielded by $||\hat{\textbf{V}}_{\rm BEE}(\hat{\boldsymbol\theta}_{\rm BEE})-\textbf{V}_{\rm BEE}(\boldsymbol\theta)||_2$, where the non-asymptotic analysis tool of random matrices are used to derive them \citep{vershynin2018high}. Note that the SE estimation should be of the same importance as the causal effect estimation.  Although the inference is made based on an unbiased estimate, it could still be invalid if the SE estimate is not reliable. Our simulations show that the vast majority of current univariable and multivariable MR approaches are unable to provide accurate SE estimates, e.g., MR-median consistently overestimates the SE and others have a tendency to underestimate it. In contrast, the sandwich formula, whose dependability has been extensively investigated empirically, is a reliable technique to obtain the SE estimate for MRBEE. This is yet another advantage of MRBEE over current approaches.
\subsection{Pleiotropy test}
Due to the complexity of GWAS data, we cannot completely rule out the possibility of the existence of UHP and CHP even in the case of modeling multiple exposures. Specifically, if UHP and CHP exist, 
\begin{align}
	\alpha_j=\boldsymbol\beta_j^\top\boldsymbol\theta+\gamma_{uj}+\gamma_{cj},
\end{align}
where $\gamma_{u_j}$ is a UHP satisfying E($\gamma_{u_j}\boldsymbol\beta_j)=\mathbf0$ and $\gamma_{c_j}$ is a CHP satisfying E($\gamma_{c_j}\boldsymbol\beta_j)\neq\mathbf0$. Conventional pleiotropy detection methods such as MR-Robust, MR-PRESSO, and IMRP do not distinguish between UHP and CHP as long as they resemble outliers. Recently, some novel methods such as CAUSE and MR-CUE have been developed to separate vertical pleiotropy, UHP and CHP by using a mixture model, allowing slightly larger proportions of UHP and CHP. However, both the conventional and novel methods only focus on one exposure, failing to realize that most CHP and UHP may disappear automatically after specifying all the relevant exposures.

In this paper, we assume that we have excluded all CHP by including all the relevant exposures and we adopt IMRP \citep{zhu2021iterative} to detect UHP. First, we define UHP as 
\begin{align}
	\gamma_j=\alpha_j-\boldsymbol\beta_j^\top\boldsymbol\theta.
\end{align}
In particular, we assume that $\gamma_j$ has a product structure $\gamma_j=\gamma_j^*b_j$,
where $\gamma^*_j$ is a fixed number and $b_j$ is a non-random binary indicator. Let $\mathcal O=\{j:\ b_j\neq0\}$ be the set of UHP. The number of elements in $\mathcal O$ (i.e., $|\mathcal O|$) should be relatively small, otherwise the UHP cannot be regarded as outliers. We specify the following variant-specific hypothesis test:
\begin{align}
	\textbf{H}_0:\ \gamma_j=0,\quad\quad\text{v.s.}\quad\quad\textbf{H}_1:\ \gamma_j\neq0.
	\label{test1}
\end{align} 
A natural estimate of $\gamma_j$ is
\begin{align}
\hat\gamma_j=\hat\alpha_j-\hat{\boldsymbol\beta}_j^\top\hat{\boldsymbol\theta}_{\rm BEE}=\gamma_j+\epsilon_j.
\end{align}
where $\epsilon_j=w_{\alpha_j}-{\boldsymbol w}_{\beta_j}^\top\boldsymbol\theta+{\boldsymbol w}_{\beta_j}^\top(\hat{\boldsymbol\theta}_{\rm BEE}-\boldsymbol\theta)$. 
It is easy to see that $\text{E}(\epsilon_j)=0$ and $
\text{var}(\epsilon_j)=\boldsymbol\theta^\top\mathbf\Sigma_{W_\beta w_\alpha}\boldsymbol\theta+\sigma_{\omega_\gamma\omega_\gamma}-2\boldsymbol\theta^\top\boldsymbol\sigma_{W_\beta w_\alpha}.$
As a result,  $t_{\gamma_j}=\hat\gamma^2_j/\text{var}(\epsilon_j)$ can be chosen as a feasible testing statistic for the hypothesis in (\ref{test1}), which follows a central $\chi^2_1$-distribution under the null hypothesis.  In practice, $\text{var}(\epsilon_j)$ can be estimated by
\begin{align}
	\widehat{\text{var}}(\epsilon_j)=\hat{\boldsymbol\vartheta}_{\rm BEE}^\top\mathbf{SE}_j\hat{\mathbf R}_{W_\beta \times w_\alpha}\mathbf{SE}_j\hat{\boldsymbol\vartheta}_{\rm BEE},
\end{align}
where $\hat{\boldsymbol\vartheta}_{\rm BEE}=(\hat{\boldsymbol\theta}^\top_{\rm BEE},-1)^\top$, $\mathbf{SE}_j=\text{diag}(\text{se}(\hat\beta_{j1}),\dots,\text{se}(\hat\beta_{jp}),\text{se}(\hat\alpha_j))$, and $\hat{\mathbf R}_{W_\beta \times w_\alpha}$ is the correlation matrix of $\hat{\mathbf\Sigma}_{W_\beta \times w_\alpha}$.  Then $\gamma_j$ is considered as an outlier if 
\begin{align}
	F_{\chi^2_1}(\hat t_{\gamma_j})>\kappa,
\end{align}
where $F_{\chi^2_1}(\cdot)$ is the CDF of $\chi^2_1$-distribution, $\hat t_{\gamma_j}=\hat\gamma_j^2/\widehat{\text{var}}(\epsilon_j)$, and $\kappa$ is a given threshold. 
\begin{theorem}
Assume that $|\mathcal O|$ is fixed and bounded and $\gamma_1*,\dots,\gamma^*_m$ are a series of non-random numbers. Then under the conditions of Theorem \ref{theorem5}, there exists a threshold $\kappa= F_{\chi^2_1}(C_0\log m)$ such that
\[
\Pr(\mathcal O=\hat{\mathcal O})\to1
\]
where $\hat{\mathcal O}=\{j:\ F_{\chi^2_1}(\hat t_{\gamma_j})>\kappa\}$ and $C_0$ is a sufficiently large constant.
\label{theorem6}
\end{theorem}
Theorem \ref{theorem6} indicates that there is a theoretical threshold $\kappa= F_{\chi^2_1}(C_0\log m)$ to consistently identify all UHP. This threshold increases with a rate $O(\log m)$ to reduce the false discovery rate (FDR) and its concrete value can be chosen by a FDR control method \citep{benjamini1995controlling}. In practice, MRBEE will iteratively apply the hypothesis test (\ref{test1}) to remove the outliers and use the remaining IVs to estimate $\boldsymbol\theta$. The stable estimate is regarded as $\hat{\boldsymbol\theta}_{\rm BEE}$. 
\section{Simulation}
In this section, we conduct numerical comparisons between MRBEE and existing MR methods. Full details of simulation settings and additional simulation results are shown in the supplementary material. 
\subsection{Univariable MR investigation}
We briefly introduce the simulation settings for univariable MR. First, we generate a binomial variable from $\text{Binom}(2,b_j)$ where $b_j\sim\text{Unif}(0.05,0.5)$ and standardize it as $g_{ij}$, the direct effect $\beta_j$ from $\mathcal{N}(0,1/m)$, and $u_i,v_i$ from a normal distribution with correlation coefficient $0.5$. The variances of $u_i$ and $v_i$ are chosen such that the IV-heritabilities are $\sigma_{\beta\beta}/\sigma_{xx}=0.3$ and $\theta^2\times(\sigma_{\beta\beta}/\sigma_{yy})=0.15$, respectively. We specify the causal effect $\theta=0.3/\sqrt{2}$. We compare MRBEE with IVW, DIVW, MR-Egger, MR-Lasso, MR-Median, IMRP, MR-ConMix, and MR-MiX, where most are implemented by using the R package \texttt{MendelianRandomization} \citep{yavorska2017mendelianrandomization}. Additionally, the IMRP procedure is incorporated into MRBEE in which the threshold $\kappa$ is chosen by R package \texttt{FDRestimation} \citep{murray2020false}. The so-called overlapping fraction is $n_{01}/n_0$, where the special fraction such that $\text{E}(S_{\rm IVW}(\theta))=0$ is $n_{01}/n_0\approx0.77.$ The number of independent replications is 1000. 

First, we study the influences of overlapping fraction $n_{01}/n_0$ and the number of IVs $m$, with the results displayed in Figure \ref{fig2}. Here, we fix $n_0=n_1=20000$, specify $n_{01}$ according to the overlapping fraction, and assume no UHP or CHP. It is easy to see that in general, only MRBEE is able to yield an unbiased estimate of $\theta$. For a special overlapping fraction (placed in the second column of Figure \ref{fig2}), all approaches become unbiased except DIVW. DIVW performs badly because it will further remove IVs based on their significance levels and consequently introduces an extra IV selection bias. In addition, the SE of causal effect estimate for all methods increases as the overlapping fraction decreases but remains unchanged by the increase of $m$. The results are consistent with our theoretical expectation and asymptotic properties of MRBEE.

\begin{figure}[p]
	\begin{center}
		\includegraphics[width=6.2in]{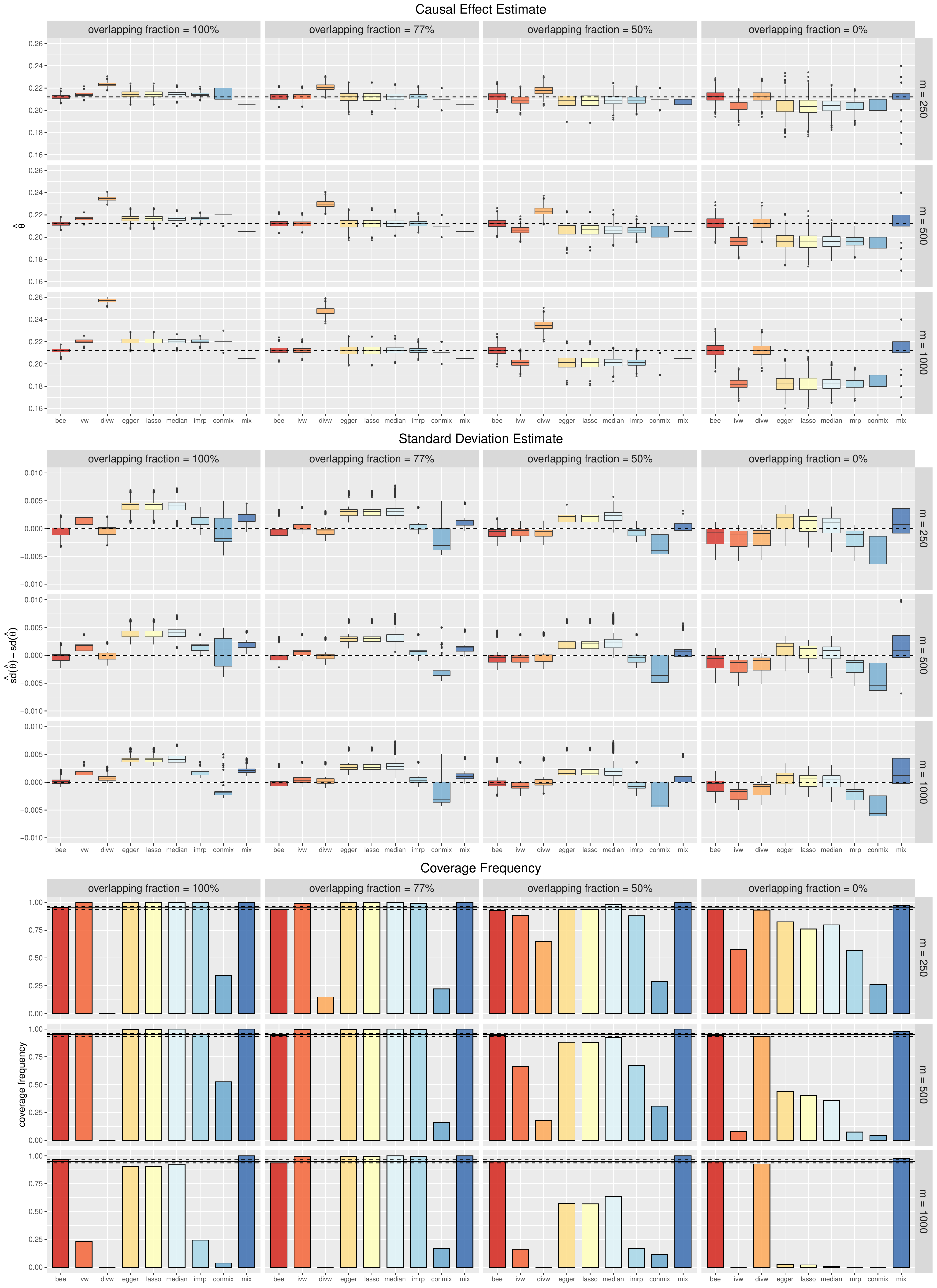}
	\end{center}
	\caption{\footnotesize Investigation of MR methods for univarate MR with sample sizes $n_0=n_1=20000$, in terms of overlapping fraction and number of instrumental variants.
		\label{fig2}}
\end{figure}
As for the standard error, we display the boxplot of $\hat{\text{se}}(\hat\theta)-\text{se}(\hat\theta)$ where $\text{se}(\hat\theta)$ is approximated by the empirical SE calculated from the independent replications. 
It is evident that the SE estimates produced by all approaches have reduced variances as $m$ grows. However, only MRBEE and DIVW can provide consistent SE estimates, confirming the accuracy of MRBEE and DIVW's SE formulas. Additionally, MR-ConMix is extremely likely to underestimate the standard error, while MR-Egger, MR-Lasso, MR-Median, and MR-Mix constantly overestimate it. As for IVW, it underestimates the SE when the fraction is large and overestimates it when the fraction is small.

The coverage frequency refers to the frequency that the confidence interval covers the true causal effect among simulations. Here, this confidence interval is constructed by doubling $\hat{\rm se}(\hat\theta)$, which means that the coverage frequency corresponding to neither an inflated type-I error nor an inflated type-II error should be around 0.95. We observed that only MRBEE enjoys a coverage frequency around 0.95. When $m=250$, MR-Egger, MR-Lasso, and MR-Median suffer from inflated type-II error rates, likely because these methods cannot estimate the SE properly. These approaches also result in inflated-type I error rates caused by weak instrument bias as $m$ increases. Additionally, because MR-Mix overestimates the SE, it consistently exhibits a substantially inflated type-II error rate. Furthermore, IMRP and MR-ConMix consistently have inflated type I error rates because they frequently underestimate the SE. 
 \begin{figure}[p]
 	\begin{center}
 		\includegraphics[width=6in]{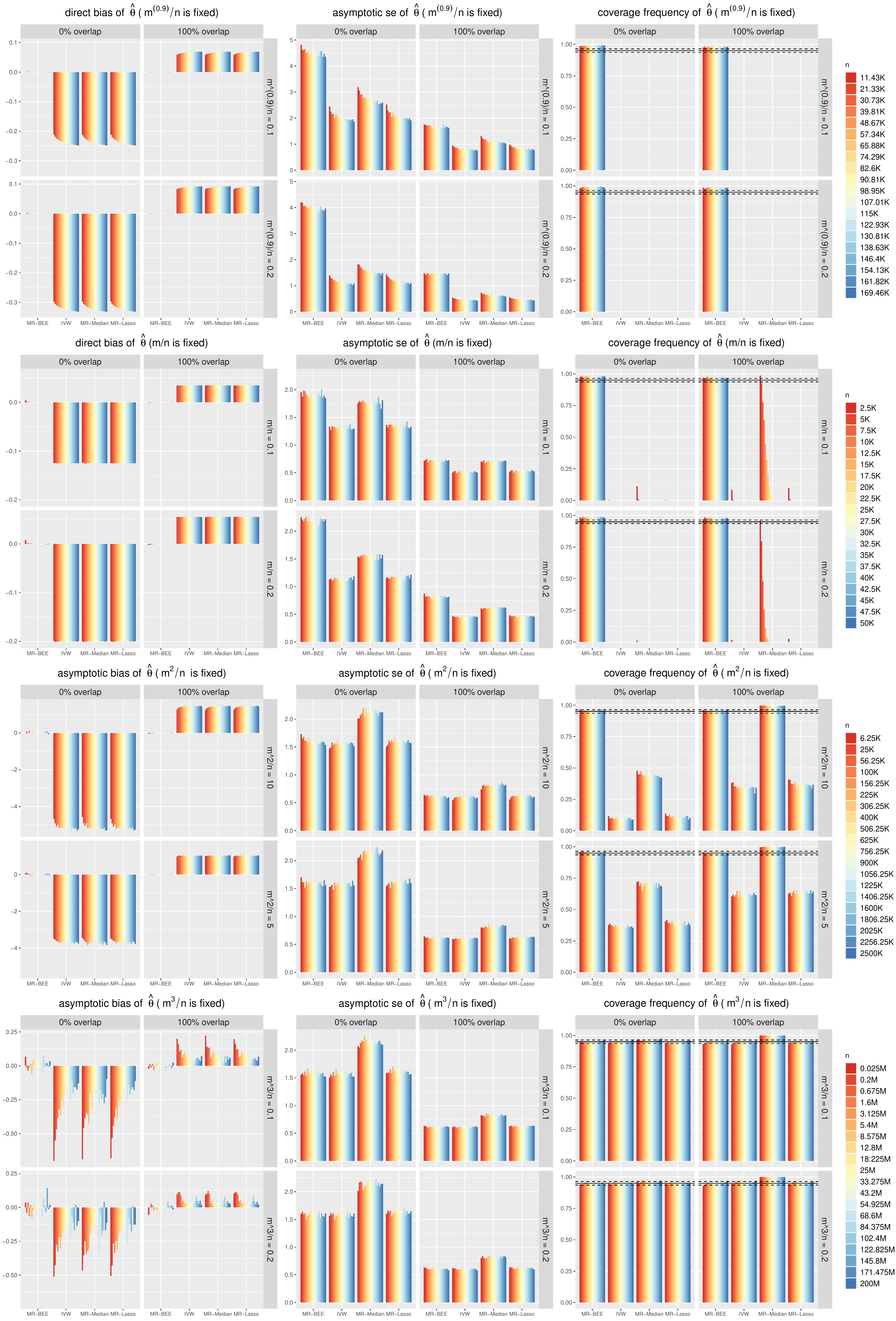}
 	\end{center}
 	\caption{\footnotesize  Investigations of MRBEE and IVW in terms of asymptotic bias and covariance matrix.
 		\label{fig3}}
\end{figure}
 
We next verify if the asymptotic normal distributions in Theorem \ref{theorem2} and Theorem \ref{theorem3} are correct. For a general estimate $\hat\theta$, the asymptotic bias and SE are $\sqrt s_n(\hat\theta-\theta)$ and $\sqrt s_n$\text{se}($\hat\theta$), respectively, where $\sqrt s_n$ is the convergence rate of $\hat\theta$. If this estimate is strongly asymptotically unbiased, the asymptotic bias $s_n(\hat\theta-\theta)$ should also be 0. Besides, if two estimates have equal asymptotic SEs, they are equally powerful in terms of statistical efficiency. We select MRBEE, IVW, MR-Median, and MR-Lasso to compare, only consider two overlapping fractions: 100\% and 0\%, set $n_0=n_1=n_{\rm min}$, and fix the causal effect $\theta=0.5$.  As for $m$ and $n_{\rm min}$, we focus on the following four cases:
\begin{itemize}
	\item[(1)] $m=2500,5000,\dots,50000$ and $m^{0.9}/n=c_0=0.1$ and $0.2$;  we examine the direct bias: $\hat\theta-\theta$, asymptotic SE:  $\sqrt{n^2_{\rm min}/m}~{\rm se}(\hat\theta)$, and coverage frequency;
	\item[(2)] $m=250,500,\dots,5000$ and $m/n=c_0=0.1$ and $0.2$;  we examine the direct bias:  $\hat\theta-\theta$, asymptotic SE:  $\sqrt n_{\rm min}~{\rm se}(\hat\theta)$, and coverage frequency;
	\item[(3)] $m=250,500,\dots,5000$ and $m^2/n=c_0=5$ and $10$;  we examine the asymptotic bias:  $\sqrt n_{\rm min}(\hat\theta-\theta)$, asymptotic SE:  $\sqrt n_{\rm min}~{\rm se}(\hat\theta)$, and coverage frequency;
	\item[(4)]  $m=250,500,\dots,5000$ and $m^3/n=c_0=5$ and $10$;  we examine the asymptotic bias:  $\sqrt n_{\rm min}(\hat\theta-\theta)$, asymptotic SE:  $\sqrt n_{\rm min}~{\rm se}(\hat\theta)$, and coverage frequency.
\end{itemize}
Note that we directly generate the estimation errors $\mathbf W_\beta$ and $\boldsymbol w_\alpha$ according to Theorem \ref{theorem1} because $n_{\rm min}$ in cases (3) and (4) can be larger than one million. The calculations involving individual-data are extremely time-consuming in these cases.

Figure \ref{fig3} demonstrates the simulation results. In case (1), $\hat\theta_{\rm BEE}$ is unbiased while the other three estimates suffer from non-removable biases. As for the asymptotic SE,  $\sqrt{n^2_{\rm min}/m}~{\rm se}(\hat\theta_{\rm BEE})$ remains unchanged when $n_{\rm min}$ and $m$ are sufficiently large (e.g., the bars colored in blue), verifying conclusion $(iii)$ in Theorem \ref{theorem3}. However, the coverage frequency of MRBEE is a little larger than 0.95, meaning that the SE of $\hat\theta_{\rm BEE}$ is overestimated in this extreme case. This phenomenon is reasonable because Theorem \ref{theorem4} points out that the convergence rate of the sandwich formula is $\min(\sqrt n_{\rm min}$,$\sqrt{n_{\rm min}^2/m},\sqrt{m/\log m})$, which slows down as $m$ increases.  In case (2), the direct bias of $\hat\theta_{\rm IVW}$ is unchanged as $n_{\rm min}$ tends to infinity, confirming conclusion $(iii)$ in Theorem \ref{theorem2}. As for $\hat\theta_{\rm BEE}$, its asymptotic SE is a little larger than $\hat\theta_{\rm IVW}$, verifying item $(ii)$ in Theorem \ref{theorem3}.

In case (3), the asymptotic bias of  $\hat\theta_{\rm IVW}$ is constant as $n_{\rm min}$ goes to infinity, illustrating that $\hat\theta_{\rm IVW}$ is not strongly asymptotically unbiased. As a result, the coverage frequencies of $\hat\theta_{\rm IVW}$ are significantly smaller than 0.95, confirming our claim that any inference made based on $\hat\theta_{\rm IVW}$ is invalid. Besides, the asymptotic SEs of $\hat\theta_{\rm BEE}$ and $\hat\theta_{\rm IVW}$ are essentially the same, indicating that  $\hat\theta_{\rm BEE}$ and $\hat\theta_{\rm IVW}$ are equally efficient as long as $m/n_{\rm min}\to0$. In case (4), the asymptotic bias of IVW, MR-Median, and MR-Lasso vanish as $n_{\rm min}$ increases and their coverage frequencies are around 0.95, which is consistent with conclusion $(i)$ in Theorem \ref{theorem2}. The equal asymptotic SEs also indicate that $\hat\theta_{\rm BEE}$ and $\hat\theta_{\rm IVW}$ are equally efficient in this scenario. In addition, IVW, MR-Median, and MR-Lasso suffer from the same degree of bias when there is no pleiotropy, while MR-Median not only suffers from a large asymptotic SE but also is likely to overestimate it. To understand why MR-Median is always less efficient than IVW when there is no pleiotropy, its asymptotic behavior is worthy of future investigation.
\subsection{Multivariable MR investigation}
For multivariable MR, we consider $p=6$ exposures and set the causal effect vector to be $\boldsymbol\theta=(0.3,0.3,-0.3,-0.3,0,0)^\top$. All of the exposures' IV-heritabilities are 0.3, while the outcome's IV-heritability is 0.15. We set an AR(1) structured genetic correlation matrix with coefficient $\rho=-0.5$ for the genetic effect $\boldsymbol\beta_j$, while considering a more intricate correlation structure for the noise terms $\boldsymbol u_i$ and $v_i$. In order to better mimic real data analysis, we take into account the scenario of completely overlapping GWAS samples (i.e., $n_{sk}=n_s=n_k$ for all $s,k$).  Other cases of sample overlaps and details of the simulation settings are present in the supplementary materials.
\begin{figure}[p]
	\begin{center}
		\includegraphics[width=6.2in]{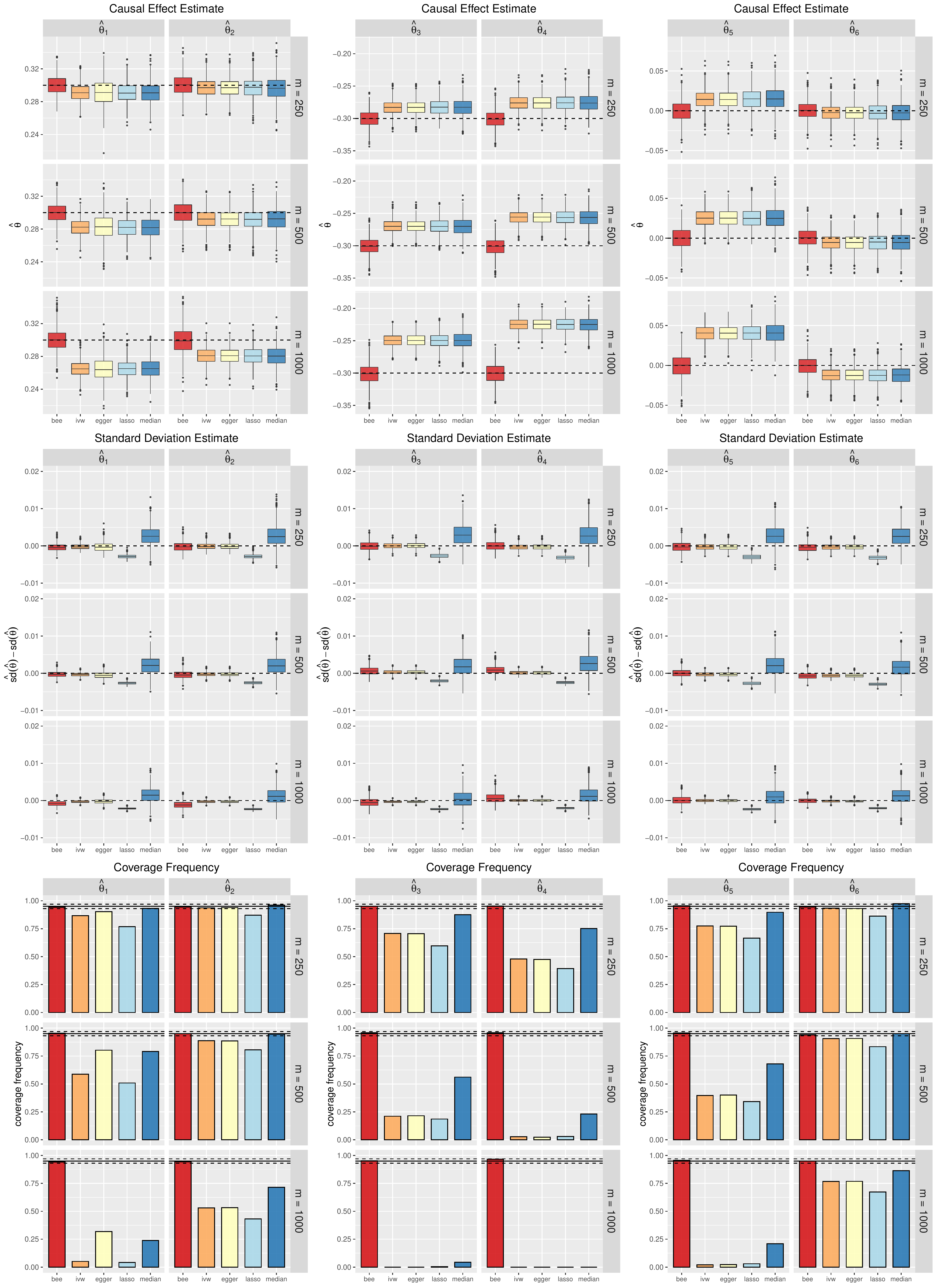}
	\end{center}
	\caption{\footnotesize  Investigation of MR methods for multivariable MR with sample sizes $n_0=\cdots=n_6=20000$ and overlap-sample sizes $n_{01}=\cdots=n_{65}=20000$, in terms of number of instrumental variants.
		\label{fig4}}
\end{figure}
\begin{figure}[t]
	\begin{center}
		\includegraphics[width=6.2in]{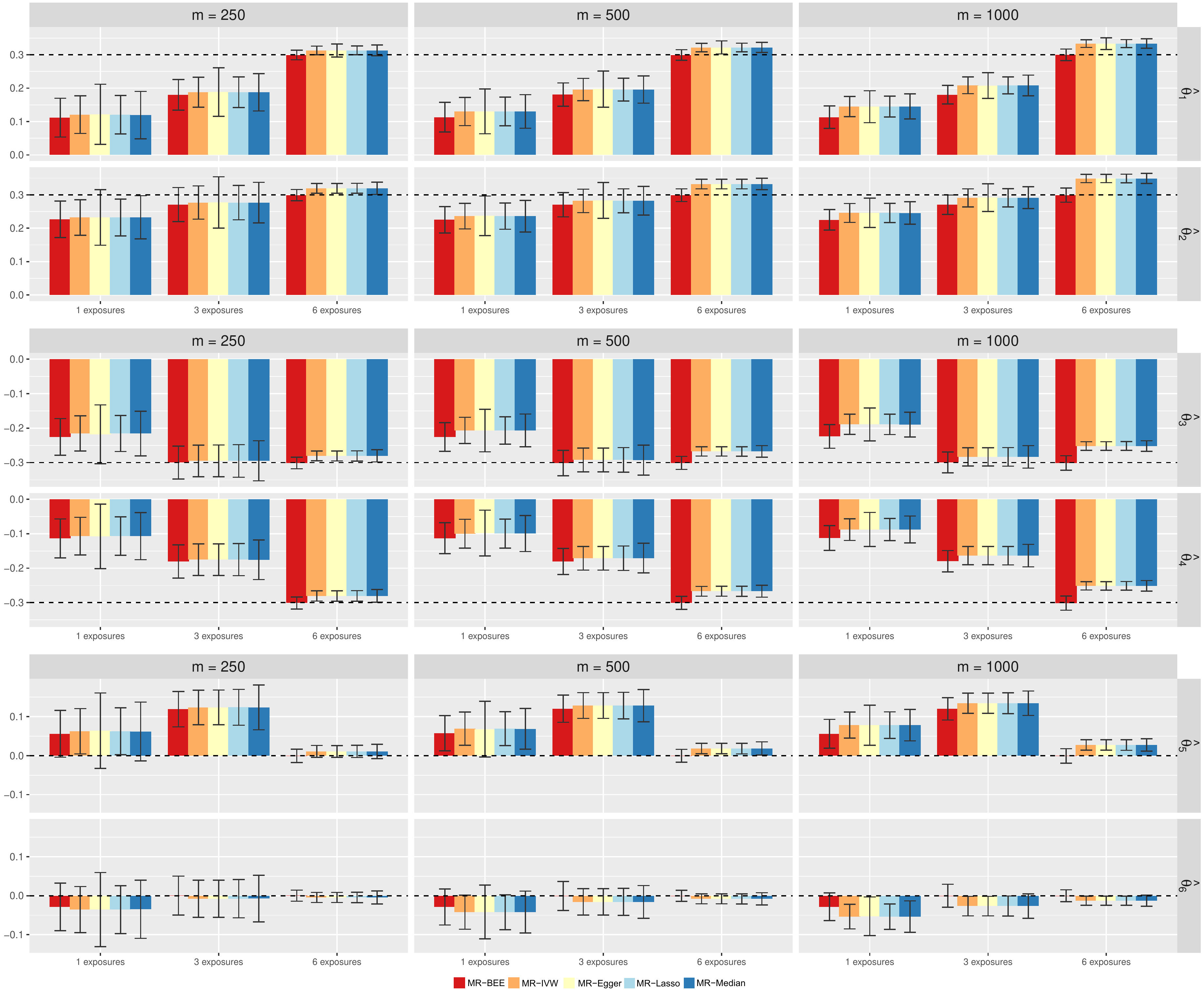}
	\end{center}
	\caption{\footnotesize  Investigation of MR methods for multivariable MR with sample sizes $n_0=\cdots=n_6=20000$ and overlap-sample sizes $n_{01}=\cdots=n_{65}=20000$, in terms of number of specified exposures.
		\label{fig5}}
\end{figure}

Figure \ref{fig4} presents the comparison between the multivariable versions of IVW, MR-Egger, MR-Lasso, MR-Median, and MRBEE. In general, MRBEE is the only method that can produce unbiased causal effect estimates in all cases. As $m$ increases, the SE of $\hat{\boldsymbol\theta}_{\rm BEE}$ remains the same, while the estimation error of the SE estimate becomes smaller. However, a very large $m$ may conversely reduce the accuracy of the SE estimate in multivariable MR. For example, the SE estimates of all approaches in the cases of $m=1000$ have larger empirical variances than those in the cases of $m=500$. This phenomenon can be explained by Theorem \ref{theorem5}, which indicates that the convergence rate of the sandwich formula is $\min(\sqrt n_{\rm min}$,$\sqrt{n_{\rm min}^2/m},\sqrt{m/\log m})$. Hence, a larger $m$ may result in a worse SE estimate if $n_{\rm min}$ is not increased as $m$.

All the multivariable MR methods except MRBEE suffer from larger weak instrument biases with the increase of $m$. The SE estimates provided by these methods, in particular MR-Median, are less reliable than that of MRBEE. Thus, causal inferences based on the existing multivariable MR methods could be even more unreliable than univariable MR methods. In addition, $\hat{\boldsymbol\theta}_{\rm IVW}$ can have a bias toward any direction in multivariable MR. For example, the bias of $\hat\theta_{5,\rm IVW}$ is positive while the bias of $\hat\theta_{6,\rm IVW}$ is negative. The actual directions are jointly determined by the correlations of confounders and genetic effects.

We also examine the impact of omitting some important exposures. We conduct simulations when 1, 3, and all 6 exposures are included in the multivariable MR model, respectively. Figure \ref{fig5} illustrates the results of the simulations. We observed that if associated exposures are omitted, the causal effect estimates can suffer severe biases. The degree of the biases is jointly determined by the genetic covariance matrix and covariance matrix of confounders. In conclusion, even though MRBEE has eliminated the estimation error bias and weak instrument bias, OVB still exists if any relevant exposure is not specified in the multivariable MR model.

\subsection{Other Investigations}
For univariable MR, we also investigated the effects of sample sizes, type-I error, winner's curse, and outlier detection. Regarding multivariable MR, we investigated the impact of different sample overlaps. In addition, the precision of estimating $\bm\Sigma_{W_\beta w_\alpha}$ by insignificant GWAS statistics is also studied. Only by increasing the sample sizes of the exposure and outcome cohorts simultaneously, the accuracy of MRBEE can be improved. The traditional MR methods suffer from inflated type-I errors when the overlapping fraction is large. After accounting for the weak instrument bias and estimate error bias, MRBEE is almost free of the winner curse's bias when the overlapping fraction is high. Furthermore, by applying the iterative method in IMRP, MRBEE can efficiently eliminate pleiotropic outliers and produce an accurate causal effect estimate. 
In addition, the estimation error $\bm\Sigma_{W_\beta w_\alpha}$ decreases with the increase of the number of insignificant variants $M$. Finally, multivariable MRBEE is accurate regardless of sample overlap. We summarized the findings with the simulation details in the supplementary material.
\section{Real Data Analysis}
Cardiovascular disease including coronary artery disease (CAD) is one of the leading causes of death for both men and women worldwide. There are many epidemiological studies and MR analyses based on GWAS summary data dedicated to identifying the causal risk factors for CAD. However, the causal effects of the risk factors on CAD are less clear and the existing evidence can be contradictory. For example, elevated low-density lipoprotein cholesterol level  (LDL-C) is a well-established causal risk factor for CAD \citep{scandinavian1994randomised}, whereas \citet{wang2022mendelian}  concluded by multivariable MR analysis that LDL-C is not causally related to CAD in Europeans. Additionally, substantial observational analyses and molecular experiments have suggested that uric acid (UA) and red blood cell counts (RBC) contribute to the development of CAD \citep{bujak2015prognostic,yu2020uric}. Nevertheless, \citet{wang2022mendelian}  did not observe significant causal effects of the two risk factors on CAD in Europeans. Furthermore, numerous MR analyses have concluded that body mass index (BMI) has a positive causal effect on CAD \citep{zhu2020mendelian,wang2022mendelian}. However, recent literature indicates that BMI is likely to influence CAD through the mediation with diseases such as diabetes and hypertension \citep{gill2021risk}. These contradictions may be due to biases in MR methods, including OVB, weak instrument bias, estimation error bias, etc.
\begin{figure}[t]
	\begin{center}
		\includegraphics[width=6.2in]{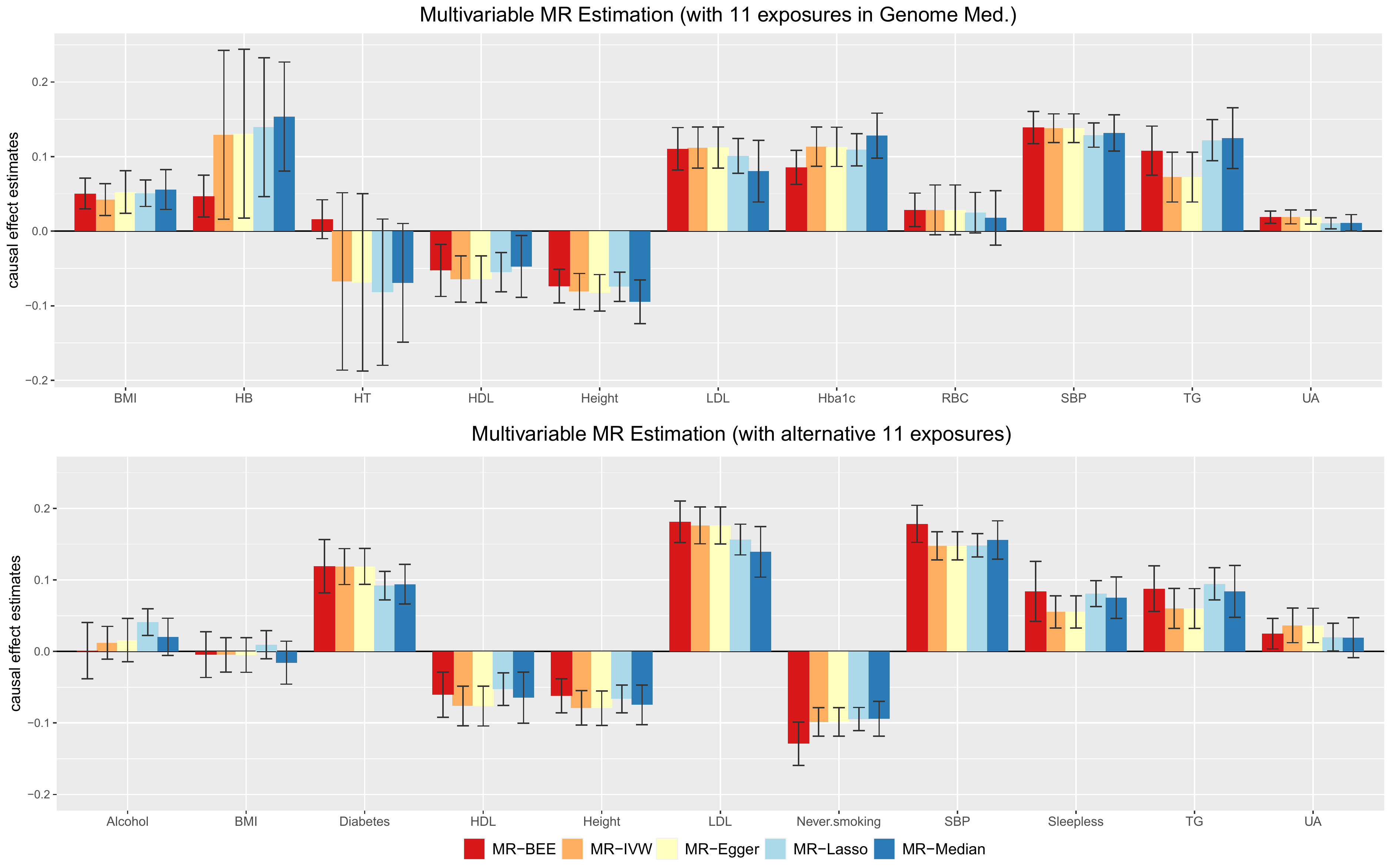}
	\end{center}
	\caption{\footnotesize  Causal effect estimates of CAD data. Confidence intervals are yielded by the double SE estimates. 
		\label{fig6}}
\end{figure}

We conducted two data analyses to estimate the causal effects of select risk factors on CAD. The first analysis uses the 11 exposures in \citet{wang2022mendelian}, including BMI, hemoglobin (HB),  hemoglobin a1c (Hba1c), hematocrit (HT), high-density lipoprotein cholesterol level (HDL-C), height, LDL-C, RBC, systolic blood pressure (SBP), triglycerides (TG), and UA. In \citet{wang2022mendelian}, these 11 exposures were divided into two groups and analyzed separately. In contrast, we analyzed them in one multivariable MR model to avoid the OVB. In the second analysis, we replace HB, Hba1c, HT, and RBC with alcohol consumption (alcohol), diabetes, lifetime never smoking status (never.smoking), and sleeplessness. All the GWAS summary statistics used in our analyses were downloaded from the Neale lab (\url{http://www.nealelab.is/uk-biobank/}). Quality controls (QCs) are presented in the supplementary material. The total numbers of instrumental variants for the first and second analyses are 5345 and 5301, respectively. 

Figure \ref{fig6} displays the causal effect estimates with 95\% confidence intervals. MRBEE confirms the causal effects of LDL-C, RBC, and UA on CAD. Here, HB, HT, and RBC have high mutual correlations: $\widehat{\text{cor}}(x_\text{HB},x_{\rm HT})=0.89$, $\widehat{\text{cor}}(x_\text{HB},x_{\rm RBC})=0.63$, and $\widehat{\text{cor}}(x_\text{HT},x_{\rm RBC})=0.72$, and thus the inferences obtained by the existing methods are not reliable. For example, the existing MR methods suggest that RBC is not significant, HB has a significant positive effect, and HT has a significant negative effect, which contradicts the fact that HT and CAD are positively associated \citep{sorlie1981hematocrit}. MRBEE corrects the estimation error bias and thus leads to a reasonable conclusion -- HB and RBC have positive causal effects on CAD while HT has a positive but insignificant causal effect on CAD. For the second analysis, MRBEE reveals that BMI is likely to affect CAD through the mediation of SBP and diabetes. In addition, MRBEE indicates that never.smoking is protective against CAD, whereas sleeplessness is associated with increasing CAD risk. Furthermore, due to the weak instrument bias and estimation error bias, the existing methods overestimate the effects of HDL-C and height and underestimate the causal effects of diabetes, LDL-C, never.smoking, SBP, and sleeplessness. By using MRBEE, we are able to obtain reliable causal effect estimates and therefore make valid inferences on the causal risk factors of CAD.
\section{Discussion}
In this paper, we first investigated the asymptotic behavior of the multivariable IVW estimate. Since almost all MR methods are based on the IVW method, understanding the asymptotic behavior of the IVW estimate has very far-reaching implications for the theoretical and empirical studies of MR methods. We found that the bias of the multivariable IVW estimate is the product of weak instrument bias and estimation error bias. Also, we revealed that estimation error bias is a linear combination of measurement error bias and confounder bias, in which the sample overlaps trade off the proportion of these two components of estimation error bias. In the literature, although the phenomenon that the IVW estimate suffers from bias has been observed, a quantitative explanation for its existence is still absent. Our work fills the gap, which is a significant theoretical contribution to MR.

Subsequently, in this paper, we describe MRBEE that can yield the unbiased causal effect estimate $\hat{\boldsymbol\theta}_{\rm BEE}$. We point out that $\hat{\boldsymbol\theta}_{\rm BEE}$ is strongly asymptotically unbiased in all scenarios, indicating that $\hat{\boldsymbol\theta}_{\rm BEE}$ is asymptotically valid when making causal inferences.  We also discuss how to perform MRBEE in practice, including how to estimate the bias-correction terms, how to estimate the sandwich formula, and how to identify possible UHP when multiple exposures are included. We present corresponding theorems to confirm that the estimates involved in the implementation of MRBEE are consistent in theory. In simulations, we show that MRBEE simultaneously estimates causal effects and the SE unbiasedly, and identifies UHP consistently. In section 5 and also in \citep{lorincz-comi2022mrbee}, the practical advances of MRBEE are further demonstrated.

It is worth offering guidance on how to properly perform MR analysis from our perspective. First, we suggest applying the multivariable MR approach instead of the univariable MR approach because the causal effect estimates obtained by the univariable MR approach are unreliable due to OVB, regardless of the presence of UHP and CHP in the model. Second, rather than selecting the optimal number of instrumental variants such that the F statistics and conditional F statistics are larger than 10 \citep{burgess2011avoiding,sanderson2021testing}, we advise including all the independent instrumental variants that are significantly associated with one or more exposures. Our theory illustrates that the asymptotic variance of a causal effect estimate is related to the cumulative variance explained by all specified IVs instead of the average variance explained by each IV. In particular, there is no need to worry about the issue of weak IVs because MRBEE has demonstrated efficiency to eliminate weak instrument bias through our simulations and theory. Third, when performing multivariable MR analysis, it is not necessary to remove variants that are pleiotropic between the exposures. For example, \citet{wang2022mendelian} observed that LDL-C was insignificantly associated with CAD in Europeans, which is unlikely to be true because this risk causality has been well established in randomized clinical trials \citep{scandinavian1994randomised}. The potential reason for this false negative is that \citet{wang2022mendelian} excluded the IVs associated with RBC, HB, HT, and UA in their multivariable MR analysis. We believe that the proper way to perform multivariable MR analysis is to simultaneously include all the relevant exposures, as the multivariable regression can automatically account for the pleiotropic variants shared by the specified exposures. Fourth, among the existing multivariable MR approaches including IVW, MR-Egger, and MR-Lasso, we recommend MRBEE as the primary analysis approach because it has been proven to be the only one that enjoys strongly asymptotic unbiasedness in the presence of many weak IVs.

\appendix
\section{Proof}
\subsection{Preliminary lemmas}
In this subsection, we specify some lemmas that can facilitate the proofs, most of which can be found in the existing papers. We first discuss the equivalent characterizations of sub-Gaussian and sub-exponential variables.
\begin{lemma}[Equivalent characterizations of sub-Guassian variables] Given any random variable $X$, the following properties are equivalent:
	\begin{itemize}
		\item[(I)] there is a constant $K_1\geq0$ such that
		\[
		\Pr(|X|\geq t)\leq 2\exp(-t^2/K_1^2),\quad\text{for all }t\geq0,
		\]
		\item[(II)] the moments of $X$ satisfy 
		\[
		||X||_{L_p}=(\text{E}(|X|^p))^{\frac1p}\leq K_2\surd p,\quad\text{for all }p\geq1,
		\]
		\item[(III)] the moment generating function (MGF) of $X^2$ satisfies:
		\[
		\text{E}\{\exp(\lambda^2X^2)\}\leq\exp(K_3^2\lambda^2),\quad\text{for all $\lambda$ staisfying }|\lambda|\leq K_3^{-1},
		\]
		\item[(IV)] the MGF of $X^2$ is bounded at some point, namely
		\[
		\text{E}\{\exp(X^2/K_4^2)\}\leq2,
		\]
		\item[(V)] if E$(X)=0$, the MGF of $X$ satisfies
		\[
		\text{E}\{\exp(\lambda X)\}\leq\exp(K_5^2\lambda^2),\quad\text{for all }\lambda\in\mathbb{R},
		\]
	\end{itemize}
	where $K_1,\dots,K_5$ are certain strictly positive constants.
\end{lemma}
This lemma summarizes some well-known properties of sub-Guassian and can be found in \citet[Proposition 2.5.2]{vershynin2018high}.
\begin{lemma}[Equivalent characterizations of sub-exponential variables] Given any random variable $X$, the following properties are equivalent:
	\begin{itemize}
		\item[(I)] there is a constant $K_1\geq0$ such that
		\[
		\Pr(|X|\geq t)\leq 2\exp(-t/K_1),\quad\text{for all }t\geq0,
		\]
		\item[(II)] the moments of $X$ satisfy 
		\[
		||X||_{L_p}=(\text{E}(|X|^p))^{\frac1p}\leq K_2 p,\quad\text{for all }p\geq1,
		\]
		\item[(III)] the moment generating function (MGF) of $|X|$ satisfies:
		\[
		\text{E}\{\exp(\lambda|X|)\}\leq\exp(K_3\lambda),\quad\text{for all $\lambda$ staisfying }0\leq\lambda\leq K_3^{-1},
		\]
		\item[(IV)] the MGF of $|X|$ is bounded at some point, namely
		\[
		\text{E}\{\exp(|X|/K_4)\}\leq2,
		\]
		\item[(V)] if E$(X)=0$, the MGF of $X$ satisfies
		\[
		\text{E}\{\exp(\lambda X)\}\leq\exp(K_5^2\lambda^2),\quad\text{for all }\lambda\leq K_5^{-1},
		\]
	\end{itemize}
	where $K_1,\dots,K_5$ are certain strictly positive constants.
\end{lemma}
This lemma summarizes some well-known properties of sub-exponential and can be found in \citet[Proposition 2.7.1]{vershynin2018high}.

\begin{lemma}[Product of sub-Gaussian variable is sub-exponential] Suppose that $X,Z$ are two sub-Gaussian variable, then $Y=XZ$ is a sub-exponential variable. Besides, if $X$ is a bounded sub-Gaussian variable, then then $Y=XZ$ is a sub-Gaussian variable.
\end{lemma}
The first claim of this lemma is provided by \citet[Proposition 2.7.7]{vershynin2018high}. The second claim of this lemma is a direct inference of \citet[Lemma A.2]{fan2011high}.

\begin{lemma}[$\ell_2$-norm of matrices with sub-Gaussian entries] Let $\boldsymbol X_1,\dots,\boldsymbol X_n$ be $n$ $(p\times1)$ independent identically distributed random vector with entries $x_{i1},\dots,x_{ip}$ are sub-Gaussian with zero-mean. Besides, define the covariance matrix of $\boldsymbol X_i$ as
	\[
	\boldsymbol\Sigma=\text{E}(\boldsymbol X_i\boldsymbol X_i^\top)
	\]
	and the related  sample covariance matrix 
	\[
	\hat{\boldsymbol\Sigma}=\frac1n\sum_{i=1}^n\boldsymbol X_i\boldsymbol X_i^\top.
	\]
	Then for every positive integer $n$,
	\[
	\text{E}(||\hat{\boldsymbol\Sigma}-\boldsymbol\Sigma||_2)\leq C\bigg(\frac{p}{n}+\sqrt{\frac pn}\bigg)||\boldsymbol\Sigma||_2,
	\]
	where $C$ is certain positive constant.
\end{lemma}
This lemma is provided by \citet[Theorem 4.7.1]{vershynin2018high}. It shows the convergence rate of sample covariance matrix is $\surd(n/m)$.
\begin{lemma}[$\ell_2$-norm of matrices with sub-exponential entries] Let $\boldsymbol X_1,\dots,\boldsymbol X_n$ be $n$ $(p\times1)$ independent identically distributed random vector with entries $x_{i1},\dots,x_{ip}$ are sub-exponential with zero-mean. Besides, define the covariance matrix of $\boldsymbol X_i$ as
	\[
	\boldsymbol\Sigma=\text{E}(\boldsymbol X_i\boldsymbol X_i^\top)
	\]
	and the related  sample covariance matrix 
	\[
	\hat{\boldsymbol\Sigma}=\frac1n\sum_{i=1}^n\boldsymbol X_i\boldsymbol X_i^\top.
	\]
	Then for ever $t\geq0$, the following inequality holds with probability at least $1-p\exp(-ct^2)$: 
	\[
	||\hat{\boldsymbol\Sigma}-\boldsymbol\Sigma||_2\leq \max(||\boldsymbol\Sigma||_2\delta,\delta^2),
	\]
	where $c$ is certain positive constant and $\delta=t\sqrt{p/n}$.
\end{lemma}
This lemma is the direct inference of \citet[Theorem 5.44]{vershynin2010introduction}. Besides, by letting $t=\sqrt{p\log n}$ we further obtain
\[
\text{E}(||\hat{\boldsymbol\Sigma}-\boldsymbol\Sigma||_2)=O\bigg(\sqrt{\frac{p\log n}n}\bigg)||\boldsymbol\Sigma||_2,
\]
if $\hat{\boldsymbol\Sigma}$ is the sample covariance matrix of sub-exponential vector. Note that in our method, the dimension $p$ is fixed and hence we cannot chose $t=\sqrt{p\log p}$ such that the estimation bound becomes $\sqrt{(p\log p)/n}||\boldsymbol\Sigma||_2$.
\begin{lemma}[Asymptotic normal distribution of Wishart matrix]
	Suppose $\boldsymbol X_1,\boldsymbol X_2,\dots,\boldsymbol X_n$ are $n$ IID relaxation of the $p$-dimensional variable $\boldsymbol X\sim\mathcal N(\mathbf 0,\bm\Sigma)$ with a well-conditioned covariance matrix $\bm\Sigma$. Besides, define the sample covariance matrix of $\bm\Sigma$ as
	\[
	\hat{\mathbf\Sigma}=\frac1n\sum_{i=1}^n\boldsymbol X_i\boldsymbol X_i^\top.
	\]
	If $p$ is a fixed number, then as $n\to\infty$,
	\[
	\sqrt n(\text{vec}(\hat{\mathbf\Sigma})-\text{vec}(\bm\Sigma))\stackrel{D}{\longrightarrow}\mathcal N\bigg(\mathbf0,(\mathbf I_{p^2}+\mathbf K_{p^2})(\bm\Sigma\otimes\bm\Sigma)\bigg),
	\]
	where $\mathbf K_{p^2}$ is the so-called commutation matrix, which is able to ensure $\mathbf K_{p^2}\text{vec}(\mathbf A)=\text{vec}(\mathbf A')$ for all $(p\times p)$ matrix.
\end{lemma}
This lemma can be found in \citet[equation (5), p90]{muirhead2009aspects}.

\subsection{Specific Lemmas}
In this subsection, we specify the following lemmas that are made based on the preliminary lemmas.
\begin{lemma}[Asymptotic normal distribution of sub-Gaussian and sub-exponential variables] Suppose $X_1,\dots,X_n$ are $n$ independent sub-Gaussian or sub-exponential variables with mean-zero and variance $\sigma_1^2,\dots,\sigma_n^2$ . Then
	\[
	\lim_{n\to\infty}\frac1{\sqrt n}\sum_{i=1}^nX_i\stackrel{D}{\longrightarrow}\mathcal N(0,\sigma_x^2),
	\]
	where 
	\[
	\sigma_x^2=\lim_{n\to\infty}\frac1n\sum_{i=1}^n\sigma_i^2.
	\]
\end{lemma}
\begin{proof}[Proof of Lemma A.7] It is easy to verify the Lyapunov’s condition: for all fixed $\delta>0$,
	\[
	\lim_{n\to\infty}\frac1{n^{1+\delta}}\sum_{i=1}^n\text{E}(|X_i|^{2+2\delta})\leq\frac{\sqrt{2K_2+2K_2\delta}^{2+2\delta}}{n^{\delta}}\to0
	\]
	by the (II) of Lemma A.1, if $X_1,\dots,X_n$ are sub-Gaussian variables;
	\[
	\lim_{n\to\infty}\frac1{n^{1+\delta}}\sum_{i=1}^n\text{E}(|X_i|^{2+2\delta})\leq\frac{(2K_2+2K_2\delta)^{2+2\delta}}{n^{\delta}}\to0
	\]
	by the (II) of Lemma A.2, if $X_1,\dots,X_n$ are sub-exponential variables. And hence the asymptotic normal distribution holds.
\end{proof}
\begin{lemma}[Asymptotic normal distribution of estimation error]Let 
	\[
	\xi_j^{[s]}=\frac1{\sqrt{n_s}}\sum_{i=1}^{n_s}g^{[s]}_{ij}x^{[s]}_{i,-j},
	\]
	where 
	\[
	x^{[s]}_{i,-j}=x^{[s]}_{i}-\beta_{js}g^{[s]}_{i,j},
	\]
	$s=0,1,\dots,p$, $x_{i,-j}^{[0]}$ represents $y_{i,-j}^{[0]}$ and $\beta{j0}$ represent $\alpha_j$. Then
	\[
	\xi^{[s]}_j\xrightarrow{D}\mathcal N(0,\sigma_{x_sx_s}-\sigma_{\beta_s\beta_s}),
	\]
	where $\sigma_{x_0x_0}$ represents $\sigma_{yy}$ and $\sigma_{\beta_0\beta_0}$ represents $\boldsymbol\theta^\top\bm\Sigma_{\beta\beta}\theta$.
\end{lemma}
\begin{proof}[Proof of Lemma A.8]
	Note that both $g^{[s]}_{ij}$ and $x^{[s]}_{i,-j}$ are sub-Gaussian ($x^{[s]}_{i,-j}$ is the product of a sub-Gaussian variable and a bounded sub-Gaussian variable), and it holds E$(g^{[s]}_{ij}x^{[s]}_{i,-j})=0$ and
	\begin{align}
		\text{var}(g^{[s]}_{ij}x^{[s]}_{i,-j})=\text{var}(g^{[s]}_{ij})\times\text{var}(x^{[s]}_{i,-j})=\sigma_{x_sx_s}-\sigma_{\beta_s\beta_s}.
	\end{align}
	As a result,
	\begin{align}
		\xi^{[s]}_j=\frac1{\sqrt{n_s}}\sum_{i=1}^{n_s}g^{[s]}_{ij}x^{[s]}_{i,-j}\xrightarrow{D}\mathcal N(0,\sigma_{x_sx_s}-\sigma_{\beta_s\beta_s}),
	\end{align}
	according Lemma A.7.
\end{proof}
\begin{lemma}[Asymptotic normality of bias-correction terms] Let 
	\[\boldsymbol\zeta_j=\bigg(\frac{n_{\rm min}}{n_1}\xi^{[1]}_{j},\frac{n_{\rm min}}{n_2}\xi^{[2]}_{j},\dots,\frac{n_{\rm min}}{n_p}\xi^{[p]}_{j},\frac{n_{\rm min}}{n_0}\xi^{[0]}_{j}\bigg)^\top.\] Under the conditions (C1)-(C4), 
	\[
	\lim_{m\to\infty}\frac1{\sqrt m}\sum_{j=1}^m(\text{vec}(\boldsymbol\zeta_j\boldsymbol\zeta_j^\top)-\text{vec}(\bm\Psi_{W_\beta\times w_\alpha}))\xrightarrow{D}\mathcal{N}\bigg(\mathbf0,(\mathbf I_{p^2}+\mathbf K_{p^2})(\bm\Psi_{W_\beta\times w_\alpha}\otimes \bm\Psi_{W_\beta\times w_\alpha})\bigg).
	\]
	as $n_{\rm min},m\to\infty$.
\end{lemma}
\begin{proof}[Proof of Lemma A.9]
	By using Lemma A.7, $\boldsymbol\zeta_j$ follows $\mathcal{N}(0,\bm\Psi_{W_\beta\times w_\alpha})$ as $n_{\rm min}\to\infty$. Then by using Lemma A.6, this lemma holds.
\end{proof}
\begin{lemma}[Asymptotic normality of residual term] Under the conditions (C1)-(C4), 
	\[
	\lim_{m\to\infty}\frac1{\sqrt m}\sum_{j=1}^m\sqrt m\boldsymbol\beta_j\xi_j^{[s]}\xrightarrow{D}\mathcal{N}(0,\sigma_{x_sx_s}\bm\Sigma_{\beta\beta}),
	\]
	and
	\[
	\lim_{m\to\infty}\frac1m\sum_{j=1}^m\sqrt m\boldsymbol\beta_j\sqrt m\boldsymbol\beta_j^\top\xi_j^{[s]}\xi_j^{[k]}\xrightarrow{P}\frac{n_{sk}}{\sqrt{n_sn_k}}\sigma_{x_sx_k}\bm\Sigma_{\beta\beta},
	\]
	for $s=0,\dots,p$, where $\sigma_{x_0x_k}$ represents $\sigma_{yx_k}=\sum_{l=1}^p\theta_{l}\sigma_{x_lx_k}$.
\end{lemma}
\begin{proof}[Proof of Lemma A.10]
	By condition (C4), $\sqrt m\boldsymbol\beta_j$ is independent of $\xi_j^{[s]}$. By Lemma A.3, $\sqrt m\boldsymbol\beta_j\xi_j^{[s]}$ is sub-exponential with mean $\mathbf 0$ and covariance matrix
	\begin{align}
		\text{cov}(\sqrt m\boldsymbol\beta_j\xi_j^{[s]})&=\text{cov}(\sqrt m\boldsymbol\beta_j)\times\text{var}(\xi_j^{[s]})\notag\\
		&= (\sigma_{x_sx_s}-\sigma_{\beta_s\beta_s})\bm\Sigma_{\beta\beta}.
	\end{align}
	Hence, by Lemma A.6,
	\[
	\lim_{m\to\infty}\frac1{\sqrt m}\sum_{j=1}^m\sqrt m\boldsymbol\beta_j\xi_j^{[s]}\xrightarrow{D}\mathcal{N}(0,\sigma_{x_sx_s}\bm\Sigma_{\beta\beta}).
	\]
	On the other hand, $\boldsymbol\beta_j\xi_j^{[s]}$ is sub-exponential variable according to Lemma A.3, and
	\begin{align}
		\text{cov}(\sqrt m\boldsymbol\beta_j\xi_j^{[s]},\sqrt m\boldsymbol\beta_j\xi_j^{[k]})&=\text{cov}(\xi_j^{[s]},\xi_j^{[k]})\times\bm\Sigma_{\beta\beta}\notag\\
		&=\frac{n_{sk}}{\sqrt{n_sn_k}}(\sigma_{x_sx_k}-\sigma_{\beta_s\beta_k})\bm\Sigma_{\beta\beta}.
	\end{align}
	Hence, by using Lemma A.5
	\[
	\lim_{m\to\infty}\frac1m\sum_{j=1}^m\sqrt m\boldsymbol\beta_j\sqrt m\boldsymbol\beta_j^\top\xi_j^{[s]}\xi_j^{[k]}\xrightarrow{P}\frac{n_{sk}}{\sqrt{n_sn_k}}\sigma_{x_sx_k}\bm\Sigma_{\beta\beta}.
	\]
\end{proof}

\subsection{Proofs of theorems in section 2}
\begin{proof}[Proof of Theorem 1]~
	As for the estimation error $\boldsymbol \omega_\alpha$, we have
	\begin{align}
		w_{\alpha_j}=\frac{\boldsymbol g_j^{[0]\top}\boldsymbol y^{[0]}}{n_0}-\alpha_j=\frac{\boldsymbol g_j^{[0]\top}\boldsymbol y^{[0]}_{-j}}{n_0},
	\end{align}
	where
	\begin{align}
		\boldsymbol y^{[0]}_{-j}=\boldsymbol y^{[0]}-\alpha_j\boldsymbol g_j^{[0]}=\sum_{s\neq j}^m\alpha_t\boldsymbol g_t^{[0]}+\mathbf U^{[0]}\boldsymbol\theta+\boldsymbol v^{[0]},
	\end{align}
	and $\mathbf U^{[0]}$ and $\boldsymbol v^{[0]}$ are the corresponding noise terms in the outcome GWAS cohort.
	According to Lemma A.8,
	\begin{align}
		\xi^{[0]}_j=\frac1{\sqrt{n_0}}\sum_{i=1}^{n_0}g^{[0]}_{ij}y^{[0]}_{i,-j}\xrightarrow{D}\mathcal N(0,\sigma_{yy}-\boldsymbol\theta^\top\bm\Sigma_{\beta\beta}\boldsymbol\theta).
	\end{align}
	As for the estimation error $w_{\beta_{js}}$, we have
	\begin{align}
		w_{\beta_{js}}=\frac{\boldsymbol g_j^{[s]\top}\boldsymbol x^{[s]}}{n_s}-\beta_{js}=\frac{\boldsymbol g_j^{[s]\top}\boldsymbol x^{[s]}_{-j}}{n_s},
	\end{align}
	where
	\begin{align}
		\boldsymbol x^{[s]}_{-j}=\boldsymbol x^{[s]}-\boldsymbol g_j^{[s]}\beta_{js}=\sum_{t\neq j}\beta_{ts}\boldsymbol g_t^{[s]}+\boldsymbol u^{[s]}.
	\end{align}
	Let 
	\begin{align}
		\xi_j^{[s]}=\frac{\boldsymbol g_j^{[s]\top}\boldsymbol x^{[s]}_{-j}}{\sqrt{n_s}}=\frac1{\sqrt n_s}\sum_{i=1}^{n_s}g_{ij}^{[s]}x_{i,-j}^{[s]},
	\end{align}
	where $x_{i,-j}^{[s]}$ is the $i$th element in vector $\boldsymbol x^{[s]}_{-j}$. According to Lemma A.8,
	\begin{align}
		\xi_j^{[s]}=\frac1{\sqrt n_s}\sum_{i=1}^{n_s}g_{ij}^{[s]}x_{i,-j}^{[s]}\xrightarrow{D}\mathcal N(0,\sigma_{x_sx_s}-\sigma_{\beta_s\beta_s}).
	\end{align}
	
	Now we show the covariance between $\xi^{[s]}_j$ and $\xi^{[k]}_j$:
	\begin{align}
		\text{cov}(\xi^{[s]}_j,\xi^{[k]}_j)=\text{E}\bigg(\frac{\boldsymbol x_{-j}^{[s]\top}\boldsymbol g_j^{[s]}\boldsymbol g_j^{[k]\top}\boldsymbol x_{-j}^{[k]}}{\sqrt{n_sn_k}}\bigg),
	\end{align}
	where $\boldsymbol x_{-j}^{[0]}$ represents $\boldsymbol y_{-j}^{[0]}$ for simplicity. Denote $\mathbf Q^{[sk]}=(Q_{it}^{[sk]})$ being a $(n_s\times n_k)$ matrix whose $(i,t)$th element is 
	\begin{align}
		Q^{[sk]}_{it}=\text{E}(g_{ij}^{[s]}g_{tj}^{[k]})=\begin{cases}
			1,& (i,t)\in\mathcal{Q}^{[sk]},\\
			0,& (i,t)\notin\mathcal{Q}^{[sk]},
		\end{cases}
	\end{align}
	where
	\begin{align}
		\mathcal Q^{[sk]}=\{(i,t):\ g_{ij}^{[s]}\text{ and }g_{tj}^{[k]}\text{ come from the same individual}\}.
	\end{align}
	As a result,
	\begin{align}
		\text{cov}(\xi^{[s]}_j,\xi^{[k]}_j)&=\text{E}\bigg(\frac{\boldsymbol x_{-j}^{[s]\top}\mathbf Q^{[sk]}\boldsymbol x_{-j}^{[k]}}{\sqrt{n_sn_k}}\bigg)=\frac1{\sqrt{n_sn_k}}\sum_{(i,t)\in\mathcal{Q}^{[sk]}}\text{E}(x_{i,-j}^{[s]}x_{t,-j}^{[k]})\notag\\
		&=\frac{n_{sk}}{\sqrt{n_sn_k}}\bigg(\sigma_{x_sx_k}-\sigma_{\beta_s\beta_k}\bigg),
	\end{align}
	where $\sigma_{x_0x_k}$ represents $\sigma_{yx_k}$ for simplicity, and $\sigma_{\beta_0\beta_k}$ represents
	\begin{align}
		\sigma_{\beta_0\beta_k}=\text{cov}(\sqrt m\boldsymbol\beta_j^\top\boldsymbol\theta,\sqrt m\beta_{jk})=\sum_{l=1}^p\theta_l\sigma_{\beta_l\beta_k}.
	\end{align}
	\par
	
	Finally, we show $\xi^{[s]}_j$ is uncorrelated with $\xi^{[s]}_t$ for all $t\neq j$ and $s=0,\dots,p$. Specifically,
	\begin{align}
		\text{cov}(\xi^{[s]}_j,\xi^{[s]}_t)=\text{E}\bigg(\frac{\boldsymbol x_{-j}^{[s]\top}\boldsymbol g_j^{[s]}\boldsymbol g_t^{[s]\top}\boldsymbol  x_{-j}^{[s]}}{n_s}\bigg).
	\end{align}
	According the model setting, $\boldsymbol g_j^{[s]}$ is independent of $\boldsymbol g_t^{[s]}$ for all $t\neq s$. Therefore, $\text{cov}(\xi^{[s]}_j,\xi^{[s]}_t)=0$.
	
	Note that if $m\to\infty$, $\bm\Sigma_{\beta\beta}=\frac1m\mathbf\Psi_{\beta\beta}$ vanishes. And so Theorem 1 is proved.
\end{proof}

\begin{proof}[Proof of Theorem 2]~
	The score function of IVW is 
	\begin{align}
		-\frac 1m\hat{\mathbf B}^\top(\boldsymbol{\hat a}-\hat{\mathbf B}\hat{\boldsymbol\theta}_{\rm IVW})=-\frac 1m\hat{\mathbf B}^\top(\boldsymbol{\hat a}-\hat{\mathbf B}\boldsymbol\theta)+\frac1m\hat{\mathbf B}^\top\hat{\mathbf B}(\hat{\boldsymbol\theta}_{\rm IVW}-\boldsymbol\theta)
	\end{align}
	which leads to 
	\begin{align}
		\mathbf H_{\rm IVW}(\hat{\boldsymbol\theta}_{\rm IVW}-\boldsymbol\theta)=-\boldsymbol S_{\rm IVW}(\boldsymbol\theta),
	\end{align}
	where 
	\begin{align}
		\mathbf H_{\rm IVW}=\frac1m\hat{\mathbf B}^\top\hat{\mathbf B},\quad \boldsymbol S_{\rm IVW}(\boldsymbol\theta)=-\frac 1m\hat{\mathbf B}^\top(\boldsymbol{\hat a}-\hat{\mathbf B}\boldsymbol\theta).
	\end{align}
	
	We first work with the Hessian matrix $\mathbf H_{\rm IVW}$:
	\begin{align}
		m\mathbf H_{\rm IVW}&=\hat{\mathbf B}^\top\hat{\mathbf B}=\mathbf B^\top\mathbf B+\mathbf B^\top\mathbf W_\beta+\mathbf W_\beta^\top\mathbf B+\mathbf W_\beta^\top\mathbf W_\beta\notag\\
		&=\mathbf J_1+\mathbf J_2+\mathbf J_3+\mathbf J_4.
	\end{align}
	As for $\mathbf J_1$,
	\begin{align}
		\mathbf J_1=\sum_{j=1}^m\boldsymbol\beta_j\boldsymbol\beta_j^\top\stackrel{P}{\longrightarrow}\bm\Psi_{\beta\beta}.
	\end{align}
	As for $\mathbf J_2$,
	\begin{align}
		\|\sqrt n_{\rm min}\mathbf J_2\|_2&=\bigg\|\frac1{\sqrt{m}}\sum_{j=1}^m(\sqrt n_{\rm min}\boldsymbol w_{\beta_j})(\sqrt m\boldsymbol\beta_j)^\top\bigg\|_2\notag\\
		&\leq\sqrt{\bigg\|\frac1{m}\sum_{j=1}^m(\sqrt n_{\rm min} \boldsymbol w_{\beta_j})(\sqrt n_{\rm min} \boldsymbol w_{\beta_j})^\top\bigg\|_2}\times\sqrt{\bigg\|\frac1{m}\sum_{j=1}^m(\sqrt{m}\boldsymbol\beta_j)(\sqrt m\boldsymbol\beta_j)^\top\bigg\|_2}\notag\\
		&\leq \lambda_{\max }^{\frac12}(\bm\Psi_{W_\beta W_\beta})\times \lambda_{\max }^{\frac12}(\bm\Psi_{\beta\beta}),
	\end{align}
	which means 
	\begin{align}
		\|\mathbf J_2\|_2=O_P(1/\sqrt n_{\rm min}).
	\end{align}
	As for $\mathbf J_3$, it has the same order as $\mathbf J_2$. As for $\mathbf J_4$,
	\begin{align}
		\frac{n_{\rm min}}{m}\mathbf J_4=\frac1m\sum_{j=1}^m(\sqrt n_{\rm min}\boldsymbol w_{\beta_j})(\sqrt n_{\rm min}\boldsymbol w_{\beta_j})^\top\stackrel{P}{\longrightarrow}\bm\Psi_{W_\beta W_\beta}
	\end{align}
	Hence: 
	\begin{itemize}
		\item[(1)] If $m/n_{\rm min}\to0$,
		\begin{align}
			\|\mathbf J_4\|_2\leq \lambda_{\max }(\bm\Psi_{W_\beta W_\beta})\times\frac{m}{n_{\rm min}}\to0.
		\end{align}
		Therefore,
		\begin{align}
			m\mathbf H_{\rm IVW}\stackrel{P}{\longrightarrow}\bm\Psi_{\beta\beta}.
		\end{align}
		
		\item[(2)] If $m/n_{\rm min}\to c_0\in(0,\infty)$, then
		\begin{align}
			\mathbf J_4=\frac{m}{n_{\rm min}}\times \frac1m\sum_{j=1}^m(\sqrt n_{\rm min}\boldsymbol w_{\beta_j})(\sqrt n_{\rm min}\boldsymbol w_{\beta_j})^\top\stackrel{P}{\longrightarrow} c_0\bm\Psi_{W_\beta W_\beta}.
		\end{align}
		Therefore,
		\begin{align}
			m\mathbf H_{\rm IVW}\stackrel{P}{\longrightarrow}\bm\Psi_{\beta\beta}+c_0\bm\Psi_{W_\beta W_\beta}.
		\end{align}
		
		\item[(3)] If $m/n_{\rm min}\to\infty$ and $m/n_{\rm min}^{1+\tau}\to c_0\in(0,+\infty)$ with certain constant $\tau>0$, then
		\begin{align}
			\frac1{n_{\rm min}^{\tau}}\mathbf J_4=\frac{m}{n_{\rm min}^{1+\tau}}\times \frac1m\sum_{j=1}^m(\sqrt n_{\rm min}\boldsymbol w_{\beta_j})(\sqrt n_{\rm min}\boldsymbol w_{\beta_j})^\top\stackrel{P}{\longrightarrow} c_0\bm\Psi_{W_\beta W_\beta}.
		\end{align}
		Therefore,
		\begin{align}
			\frac{m}{n_{\rm min}^{\tau}}\mathbf H_{\rm IVW}=c_0n_{\rm min}\mathbf H_{\rm IVW}\stackrel{P}{\longrightarrow} c_0\bm\Psi_{W_\beta W_\beta}.
		\end{align}
	\end{itemize}
	
	We then work with $\boldsymbol S_{\rm IVW}(\theta)$:
	\begin{align}
		m\boldsymbol S_{\rm IVW}(\theta)&=-\mathbf B^\top\boldsymbol w_\alpha-\mathbf W_\beta^\top\boldsymbol w_\alpha+\mathbf B^\top\mathbf W_\beta\boldsymbol\theta+\mathbf W_\beta^\top\mathbf W_\beta\boldsymbol\theta\notag\\
		&=\boldsymbol K_1+\boldsymbol K_2+\boldsymbol K_3+\boldsymbol K_4.
	\end{align}
	As for $\boldsymbol K_1+\boldsymbol K_3$,
	\begin{align}
		\sqrt n_{\rm min}(\boldsymbol K_1+\boldsymbol K_3)=\frac1{\sqrt m}\sum_{j=1}^m(-\sqrt n_{\rm min} w_{\alpha_j}+\sqrt n_{\rm min}\boldsymbol w_{\beta_j}^\top\boldsymbol\theta)(\sqrt m\boldsymbol\beta_j)\stackrel{D}{\longrightarrow}\mathcal N(\mathbf 0,\psi_\theta\bm\Psi_{\beta\beta}),
	\end{align}
	where
	\begin{align}
		\psi_\theta=\psi_{w_\alpha w_\alpha}+\boldsymbol\theta^\top\bm\Psi_{W_\beta W_\beta}\boldsymbol\theta-2\boldsymbol\theta^\top\boldsymbol\psi_{W_\beta w_\alpha}.\label{barsigma}
	\end{align}
	As for $\boldsymbol K_2$,
	\begin{align}
		\frac {n_{\rm min}}{m} \boldsymbol K_2=-\frac 1m\sum_{j=1}^m(\sqrt n_{\rm min}w_{\alpha_j})(\sqrt n_{\rm min}\boldsymbol w_{\beta_j})\stackrel{P}{\longrightarrow}-\boldsymbol\psi_{W_\beta w_\alpha}.
	\end{align}
	As for $\boldsymbol K_4$,
	\begin{align}
		\frac {n_{\rm min}}{m} \boldsymbol K_4=\bigg(\frac 1m\sum_{j=1}^m(\sqrt n_{\rm min}\boldsymbol w_{\beta_j}\sqrt n_{\rm min}\boldsymbol w_{\beta_j}\bigg)\boldsymbol\theta\stackrel{P}{\longrightarrow}\bm\Psi_{W_\beta W_\beta}\boldsymbol\theta,
	\end{align}
	
	Jointing these results, we summary the asymptotic behavior of $\hat{\boldsymbol\theta}_{\rm IVW}$:
	\begin{itemize}
		\item[(1)] If  $m/\sqrt n_{\rm min}\to 0$, then
		\begin{align}
			\sqrt n_{\min }||\boldsymbol K_2+\boldsymbol K_4||=O_P\bigg(\frac{m}{\sqrt n_{\rm min}}\bigg)=o_P(1).
		\end{align}
		Therefore,
		\begin{align}
			\sqrt n_{\rm min} \times m\boldsymbol S_{\rm IVW}(\boldsymbol\theta)=\sqrt n_{\min }(\boldsymbol K_1+\boldsymbol K_3)+o_P(1)\stackrel{D}{\longrightarrow}\mathcal N(\mathbf0,\psi_\theta\bm\Psi_{\beta\beta}).
		\end{align}
		Note that when $m/n_{\rm min}\to0$, $m\mathbf H_{\rm IVW}\stackrel{P}{\longrightarrow}\bm\Psi_{\beta\beta}.$
		Therefore,
		\begin{align}
			\sqrt n_{\rm min}(\hat{\boldsymbol\theta}_{\rm IVW}-\boldsymbol\theta)=-\sqrt n_{\rm min}(m\mathbf H_{\rm IVW})^{-1}(m\boldsymbol S_{\rm IVW}(\theta))
			\stackrel{D}{\longrightarrow}\mathcal N(\mathbf 0,\psi_\theta\bm\Psi_{\beta\beta}^{-1}),
		\end{align}
		
		\item[(2)] If  $m/\sqrt n_{\rm min}\to c_0$, then
		\begin{align}
			\sqrt n_{\rm min}(\boldsymbol K_2+\boldsymbol K_4)\to -c_0\boldsymbol\psi_{W_\beta w_\alpha}+c_0\bm\Psi_{W_\beta W_\beta}\boldsymbol\theta,
		\end{align}
		and hence 
		\begin{align}
			\sqrt n_{\rm min} \times m\boldsymbol S_{\rm IVW}(\boldsymbol\theta)\stackrel{D}{\longrightarrow}\mathcal N(-c_0(\boldsymbol\psi_{W_\beta w_\alpha}+\bm\Psi_{W_\beta W_\beta}\boldsymbol\theta),\psi_\theta\bm\Psi_{\beta\beta}).
		\end{align}
		Note that when $m/n_{\rm min}\to0$, $m\mathbf H_{\rm IVW}\stackrel{P}{\longrightarrow}\bm\Psi_{\beta\beta}.$ 
		Therefore,
		\begin{align}
			\sqrt n_{\rm min}(\hat{\boldsymbol\theta}_{\rm IVW}-\boldsymbol\theta)&=-\sqrt n_{\rm min}(m\mathbf H_{\rm IVW})^{-1}(m\boldsymbol S_{\rm IVW}(\theta))\notag\\
			&\stackrel{D}{\longrightarrow}\mathcal N(c_0\bm\Psi_{\beta\beta}^{-1}(\boldsymbol\psi_{W_\beta w_\alpha}-\bm\Psi_{W_\beta W_\beta}\boldsymbol\theta),\psi_\theta\bm\Psi_{\beta\beta}^{-1}).
		\end{align}
		
		\item[(3)] If $m/\sqrt n_{\rm min}\to\infty$ and $m/n_{\rm min}\to c_0$, then $||\boldsymbol K_1+\boldsymbol K_3||_2=O_P(1/\sqrt n_{\rm min})$, 
		\begin{align}
			\boldsymbol K_2+\boldsymbol K_4\stackrel{P}{\longrightarrow}-c_0\boldsymbol\psi_{W_\beta w_\alpha}+c_0\bm\Psi_{W_\beta W_\beta}\boldsymbol\theta,
		\end{align}
		and 
		\begin{align}
			m\mathbf H_{\rm IVW}\stackrel{P}{\longrightarrow}\bm\Psi_{\beta\beta}+c_0\bm\Psi_{W_\beta W_\beta}.
		\end{align}
		Hence, 
		\begin{align}
			\hat{\boldsymbol\theta}_{\rm IVW}-\boldsymbol\theta\stackrel{P}{\longrightarrow}c_0(\bm\Psi_{\beta\beta}+c_0\bm\Psi_{W_\beta W_\beta})^{-1}(\boldsymbol\psi_{W_\beta w_\alpha}-\bm\Psi_{W_\beta W_\beta}\boldsymbol\theta).
		\end{align}
		
		\item[(4)] If $m/n_{\rm min}\to\infty$ and $m/n_{\rm min}^{1+\tau}\to c_0$, then 
		\begin{align}
			\frac1{n_{\rm min}^{\tau}}(\boldsymbol K_2+\boldsymbol K_4)\stackrel{P}{\longrightarrow} -c_0\boldsymbol\psi_{W_\beta w_\alpha}+c_0\bm\Psi_{W_\beta W_\beta}\boldsymbol\theta
		\end{align}
		and
		\begin{align}
			\frac{m}{n_{\rm min}^{\tau}}\mathbf H_{\rm IVW}\stackrel{P}{\longrightarrow} c_0\bm\Psi_{W_\beta W_\beta}.
		\end{align}
		Therefore,
		\begin{align}
			\hat{\boldsymbol\theta}_{\rm IVW}\stackrel{P}{\longrightarrow}\bm\Psi_{W_\beta W_\beta}^{-1}\boldsymbol\psi_{W_\beta w_\alpha}.
		\end{align}
	\end{itemize}
	Now Theorem 2 is proved.
\end{proof}

\subsection{Proofs of theorems in section 3}
\begin{proof}[Proofs of Theorem 3]~
	Note that 
	\begin{align}
		\mathbf 0=\boldsymbol S_{\rm BEE}(\hat{\boldsymbol\theta}_{\rm BEE})=\boldsymbol S_{\rm BEE}(\boldsymbol\theta)+\mathbf H_{\rm BEE}(\hat{\boldsymbol\theta}_{\rm BEE}-\boldsymbol\theta),
	\end{align}
	where
	\begin{align}
		\boldsymbol S_{\rm BEE}(\boldsymbol\theta)=-\frac1m\hat{\mathbf B}^\top(\hat{\boldsymbol\alpha}-\hat{\mathbf B}\boldsymbol\theta)-\mathbf\Sigma_{W_\beta W_\beta}\boldsymbol\theta+\boldsymbol\sigma_{W_\beta w_\alpha},
	\end{align}
	and
	\begin{align}
		\mathbf H_{\rm BEE}=\frac1m\hat{\mathbf B}^\top\hat{\mathbf B}-\mathbf\Sigma_{W_\beta W_\beta}.
	\end{align}
	As for $\boldsymbol S_{\rm BEE}(\boldsymbol\theta)$,
	\begin{align}
		m\boldsymbol S_{\rm BEE}(\boldsymbol\theta)&=-(\mathbf B+\mathbf W_\beta)^\top(\boldsymbol\alpha+\boldsymbol w_\alpha-\mathbf B\boldsymbol\theta-\mathbf W_\beta\boldsymbol\theta)-m\mathbf\Sigma_{W_\beta W_\beta}+m\boldsymbol\sigma_{W_\beta w_\alpha}\notag\\
		&=-\bigg\{\mathbf B^\top(\boldsymbol w_\alpha-\mathbf W_\beta\boldsymbol\theta)\bigg\}+\bigg\{\bigg(\mathbf W_\beta^\top\mathbf W_\beta-m\mathbf\Sigma_{W_\beta W_\beta}\bigg)\boldsymbol\theta\bigg\}-\bigg\{\mathbf W_\beta^\top\boldsymbol w_\alpha-m\boldsymbol\sigma_{W_\beta w_\alpha}\bigg\}\notag\\
		&=\boldsymbol K_1+\boldsymbol K_2+\boldsymbol K_3.
	\end{align}
	Here, we define a new vector $\boldsymbol\vartheta=(\boldsymbol\theta^\top,1)^\top$, an alternative vector \[\boldsymbol\zeta_j=\bigg(\frac{n_{\rm min}}{n_1}\xi^{[1]}_{j},\frac{n_{\rm min}}{n_2}\xi^{[2]}_{j},\dots,\frac{n_{\rm min}}{n_p}\xi^{[p]}_{j},\frac{n_{\rm min}}{n_0}\xi^{[0]}_{j}\bigg)^\top,\] where
	\[
	\xi_j^{[s]}=\frac1{\sqrt n_s}\sum_{i=1}^{n_s}g_{ij}^{[s]}x_{is}^{[s]},\quad s=0,1,\dots,p,
	\]
	and a new covariance matrix
	\begin{align}
		\text{cov}(\boldsymbol\zeta_j)= \bm\Psi_{W_\beta\times w_\alpha}=\begin{pmatrix}\bm\Psi_{W_\beta W_\beta}&\boldsymbol\psi_{W_\beta w_\alpha}\\\boldsymbol\psi_{W_\beta w_\alpha}^\top&\psi_{w_\alpha w_\alpha}\end{pmatrix}.
	\end{align}
	As for $\boldsymbol K_1$, it can be rewritten as
	\begin{align}
		\sqrt n_{\rm min}\boldsymbol K_1&=-\sum_{j=1}^m\sqrt n_{\rm min}(w_{\alpha_j}-\boldsymbol w_{\beta_j}^\top\boldsymbol\theta)\boldsymbol\beta_j=
		\frac1{\sqrt m}\sum_{j=1}^m(\sqrt n_{\rm min} \boldsymbol\zeta_j^\top\boldsymbol\vartheta)(\sqrt m\boldsymbol \beta_j)\notag\\
		&\stackrel{D}{\longrightarrow}\mathcal N(\mathbf0,\psi_{\theta}\bm\Psi_{\beta\beta}),
	\end{align}
	where $\psi_{\theta}$ defined in (\ref{barsigma}) can be rewritten as
	\begin{align}
		\psi_{\theta}=\boldsymbol\vartheta^\top\bm\Psi_{W_\beta\times w_\alpha}\boldsymbol\vartheta.
	\end{align}
	As for $\boldsymbol K_2+\boldsymbol K_3$, it can be rewritten as
	\begin{align}
		\boldsymbol K_2+\boldsymbol K_3&=\mathbf I_{p+1}^{1:p}\begin{pmatrix}
			\mathbf W_\beta^\top\mathbf W_\beta-m\mathbf\Sigma_{W_\beta W_\beta}&\mathbf W_\beta^\top\boldsymbol w_\alpha-m\boldsymbol\sigma_{W_\beta w_\alpha}\\
			\boldsymbol w_\alpha^\top\mathbf W_\beta-m\boldsymbol\sigma_{W_\beta w_\alpha}^\top&\boldsymbol w_\alpha^\top\boldsymbol w_\alpha-m\sigma_{w_\alpha w_\alpha}
		\end{pmatrix}\begin{pmatrix}\boldsymbol\theta\\-1\end{pmatrix}\notag\\
		&=\frac{\sqrt{m}}{n_{\rm min}}\mathbf I_{p+1}^{1:p}\bigg(\frac1{\sqrt m}\sum_{j=1}^m\boldsymbol\zeta_j\boldsymbol\zeta_j^\top-\bm\Psi_{W_\beta\times w_\alpha}\bigg)\boldsymbol\vartheta\notag\\
		&=\frac{\sqrt{m}}{n_{\rm min}}\mathbf I_{p+1}^{1:p}\mathbf K_4\boldsymbol\vartheta,
	\end{align}
	where $\mathbf I_{p+1}^{1:p}$ is a $(p\times (p+1))$ matrix consisting of the first $p$ row of $\mathbf I_{p+1}$ and
	\begin{align}
		\mathbf K_4=\frac1{\sqrt m}\sum_{j=1}^m\boldsymbol\zeta_j\boldsymbol\zeta_j^\top-\bm\Psi_{W_\beta\times w_\alpha}.
	\end{align}
	According to Lemma A.6,
	\begin{align}
		\text{vec}(\mathbf K_4)\stackrel{D}{\longrightarrow}\mathcal N\bigg(\mathbf 0,(\mathbf I_{(p+1)^2}+\mathbf K_{(p+1)^2})(\bm\Psi_{W_\beta\times w_\alpha}\otimes \bm\Psi_{W_\beta\times w_\alpha})\bigg).
	\end{align}
	As a result,
	\begin{align}
		\frac{n_{\rm min}}{\sqrt m}(\boldsymbol K_2+\boldsymbol K_3)\stackrel{D}{\longrightarrow}\mathcal N(\mathbf 0,\mathbf\Sigma_{\rm BC})
	\end{align}
	where 
	\begin{align}
		\mathbf\Sigma_{\rm BC}=\underbrace{\bigg[\boldsymbol\vartheta^\top\otimes\mathbf I_{p+1}^{1:p}\bigg]}_{p\times (p+1)^2}\underbrace{\bigg[(\mathbf I_{(p+1)^2}+\mathbf K_{(p+1)^2})(\bm\Psi_{W_\beta\times w_\alpha}\otimes \bm\Psi_{W_\beta\times w_\alpha})\bigg]}_{(p+1)^2\times(p+1)^2}\underbrace{\bigg[\boldsymbol\vartheta^\top\otimes\mathbf I_{p+1}^{1:p}\bigg]^\top}_{(p+1)^2\times p}.
		\label{SigmaBC}
	\end{align}
	So far, we can obtain:
	\begin{itemize}
		\item[(1)] If $m/n_{\rm min}\to0$, 
		\begin{align}
			\sqrt n_{\rm min}\times m\boldsymbol S_{\rm BEE}(\boldsymbol\theta)=\sqrt n_{\rm min}\boldsymbol K_1+o_P(1)\stackrel{D}{\longrightarrow}\mathcal N(\mathbf0,\psi_{\theta}\bm\Psi_{\beta\beta}).
		\end{align}
		\item[(2)] If $m/n_{\rm min}\to c_0$,
		\begin{align}
			\sqrt n_{\rm min}\times m\boldsymbol S_{\rm BEE}(\boldsymbol\theta)=\sqrt n_{\rm min}\boldsymbol K_1+\sqrt n_{\rm min}(\boldsymbol K_2+\boldsymbol K_3)\stackrel{D}{\longrightarrow}\mathcal N(\mathbf0,\psi_{\theta}\bm\Psi_{\beta\beta}+c_0\mathbf{\Sigma}_{\rm BC}).
		\end{align}
		\item[(3)] If $m/n_{\rm min}\to\infty$ and $\sqrt m/n_{\rm min}\to0$,
		\begin{align}
			\frac{n_{\rm min}}{\sqrt m}\times m\boldsymbol S_{\rm BEE}(\boldsymbol\theta)=\frac{n_{\rm min}}{\sqrt m}(\boldsymbol K_2+\boldsymbol K_3)+\frac{n_{\rm min}}{\sqrt m}\boldsymbol K_1\stackrel{D}{\longrightarrow}\mathcal N(\mathbf 0,\mathbf\Sigma_{\rm BC}),
		\end{align}
		where
		\begin{align}
			\frac{n_{\rm min}}{\sqrt m}\boldsymbol K_1=\sqrt{\frac{n_{\rm min}}m}\times\sqrt n_{\min }\boldsymbol K_1=O_P\bigg(\sqrt{\frac{n_{\rm min}}m}\bigg)=o_P(1).
		\end{align}
	\end{itemize}
	
	Now we move to $\mathbf H_{\rm BEE}$:
	\begin{align}
		m\mathbf H_{\rm BEE}&=\mathbf B^\top\mathbf B+\bigg(\mathbf W_\beta^\top\mathbf W_\beta-m\mathbf\Sigma_{W_\beta W_\beta}\bigg)+\mathbf B^\top\mathbf W_\beta+\mathbf W_\beta^\top\mathbf B\notag\\
		&=\mathbf J_1+\mathbf J_2+\mathbf J_3+\mathbf J_4.
	\end{align}
	As for $\mathbf J_1=\mathbf B^\top\mathbf B$, we have
	\begin{align}
		||\mathbf J_1-\bm\Psi_{\beta\beta}||_2&=\bigg\|\frac1m\sum_{j=1}^m\sqrt m\boldsymbol\beta_j\sqrt m\boldsymbol\beta_j^\top-\bm\Psi_{\beta\beta}\bigg\|_2\notag\\
		&=O_P\bigg(\frac1{\sqrt m}\bigg).
	\end{align}
	As for  $\mathbf J_2=\mathbf W_\beta^\top\mathbf W_\beta-m\bm\Sigma_{W_\beta W_\beta}$, we have
	\begin{align}
		\mathbf J_2=\sum_{j=1}^{m}\bigg(\boldsymbol w_{\beta_j}\boldsymbol w_{\beta_j}^\top-\mathbf\Sigma_{W_\beta W_\beta}\bigg)=\frac{\sqrt m}{n_{\rm min}}\frac1{\sqrt m}\sum_{j=1}^m\bigg(\boldsymbol\xi_j\boldsymbol\xi_j^\top-\bm\Psi_{W_\beta W_\beta}\bigg).
	\end{align}
	As a result,
	\begin{align}
		\frac{n_{\rm min}}{\sqrt m}\text{vec}(\mathbf J_2)\stackrel{D}{\longrightarrow}\mathcal N(\mathbf0,(\mathbf I_{p^2}+\mathbf K_{p^2})(\bm\Psi_{W_\beta W_\beta}\otimes \bm\Psi_{W_\beta W_\beta})),
	\end{align}
	which means $||\mathbf J_2||=O_P(\sqrt m/n_{\rm min})$. As for $\mathbf J_3=\mathbf B^\top\mathbf W_{\beta}$,
	\begin{align}
		\sqrt n_{\min }||\mathbf J_3||_2&=\bigg\|\frac1{\sqrt m}\sum_{j=1}^m\sqrt m\boldsymbol\beta_j\sqrt n_{\min }\boldsymbol\omega_{\beta_j}^\top\bigg\|_2\notag\\
		&\leq\sqrt{\bigg\|\frac1m\sum_{j=1}^m\sqrt m\boldsymbol\beta_j\sqrt m\boldsymbol\beta_j^\top\bigg\|_2}\sqrt{\bigg\|\frac1m\sum_{j=1}^m\sqrt n_{\min }\boldsymbol\omega_{\beta_j}\sqrt n_{\min }\boldsymbol\omega_{\beta_j}^\top\bigg\|_2}\notag\\
		&\leq\lambda_{\max }^{\frac12}(\bm\Psi_{\beta\beta})\times\lambda^{\frac12}_{\max }(\bm\Psi_{W_\beta W_\beta}),
	\end{align}
	which means
	\begin{align}
		||\mathbf J_3||_2=O_P\bigg(\frac1{\sqrt n_{\min }}\bigg)
	\end{align}
	As for $\mathbf J_4$, it is easy to see $||\mathbf J_4||_2^2=||\mathbf J_3||_2^2$.
	Hence, for all three scenarios in Theorem 3,
	\begin{align}
		||m\mathbf H_{\rm BEE}-\bm\Psi_{\beta\beta}||_2=O_P\bigg\{\max\bigg(\frac1{\sqrt m},\frac1{\sqrt n_{\rm min}},\frac{\sqrt{m}}{n_{\rm min}}\bigg)\bigg\}.
	\end{align}
	
	And hence, according to the Slutsky's theorem, 
	\begin{enumerate}
		\item[(1)] If $m/n_{\min }\to0$, 
		\begin{align}
			\sqrt n_{\rm min}(\hat{\boldsymbol\theta}_{\rm BEE}-\boldsymbol\theta)=-\sqrt n_{\rm min}\mathbf\Psi_{\beta\beta}^{-1}\boldsymbol K_1\stackrel{D}{\longrightarrow}\mathcal N(\mathbf 0,\psi_\theta\bm\Psi_{\beta\beta}^{-1}).
		\end{align}
		\item[(2)] If $m/n_{\min }\to c_0$,
		\begin{align}
			\sqrt n_{\rm min}(\hat{\boldsymbol\theta}_{\rm BEE}-\boldsymbol\theta)=-\sqrt n_{\rm min}\mathbf\Psi_{\beta\beta}^{-1}(\boldsymbol K_1+\boldsymbol K_2+\boldsymbol K_3)\stackrel{D}{\longrightarrow}\mathcal N(\mathbf 0,\psi_\theta\bm\Psi_{\beta\beta}^{-1}+c_0\bm\Psi_{\beta\beta}^{-1}\bm\Psi_{\rm BC}\bm\Psi_{\beta\beta}^{-1}).
		\end{align}
		\item[(2)] If $m/n_{\rm min}\to\infty$ and $m/n_{\min }^2\to0$, 
		\begin{align}
			\sqrt{n_{\rm min}^2/m}(\hat{\boldsymbol\theta}_{\rm BEE}-\boldsymbol\theta)
			=-\frac{n_{\rm min}}{\sqrt m}\mathbf\Psi_{\beta\beta}^{-1}(\boldsymbol K_2+\boldsymbol K_3)
			\stackrel{D}{\longrightarrow}\mathcal N(\mathbf 0,\bm\Psi_{\beta\beta}^{-1}\bm\Psi_{\rm BC}\bm\Psi_{\beta\beta}^{-1}).
		\end{align}
	\end{enumerate}
	Thus, Theorem 3 is proved.
\end{proof}

\begin{proof}[Proof of Theorem 4]~
	Similar to $\xi_j^{[s]}$, we define $\eta_j^{\{s\}}$ as
	\begin{align}
		\eta_j^{\{s\}}=\frac{\boldsymbol g_j^{\{s\}\top}\boldsymbol x^{[s]}}{\sqrt {n_s}}=\frac1{\sqrt{n_s}}\sum_{i=1}^{n_s}g_{ij}^{\{s\}}x_{i}^{[s]}.
	\end{align}
	By using similar deduction as which in the proof of Theorem 1,
	\begin{align}
		\eta_j^{\{s\}}\xrightarrow{D}\mathcal{N}(0,\sigma_{x_sx_s})
	\end{align}
	and 
	\begin{align}
		\text{cov}(\eta_j^{\{s\}},\eta_j^{\{k\}})=\frac{n_{sk}}{\sqrt{n_sn_k}}\sigma_{x_sx_k}.
	\end{align}
	Denote $\boldsymbol\eta_j=(\eta_j^{\{1\}},\dots,\eta_j^{\{p\}},\eta_j^{\{0\}})$ where $\eta_j^{\{0\}}$ represents $\frac1{\sqrt {n_0}}\boldsymbol g_j^{\{s\}\top}\boldsymbol y^{[0]}$. Then we have
	\begin{align}
		\text{cov}(\boldsymbol\eta_j)=\mathbf D_\eta^{-1}\bm\Sigma_{W_\beta\times w_\alpha}\mathbf D_\eta^{-1},
	\end{align}
	where 
	\begin{align}
		\mathbf D_\eta=\text{diag}\bigg(\frac1{\sqrt{n_1}},\dots,\frac1{\sqrt{n_p}},\frac1{\sqrt{n_0}}\bigg).
	\end{align}
	By using Lemma A.4,
	\begin{align}
		\bigg\|\frac1M\sum_{j=1}^M\boldsymbol\eta_j\boldsymbol\eta_j^\top-\text{cov}(\boldsymbol\eta_j)\bigg\|_2=O_P\bigg(\frac1{\sqrt M}\bigg),
	\end{align}
	and hence
	\begin{align}
		\|\bm\Sigma_{W_\beta\times w_\alpha}^{-\frac12}\hat{\bm\Sigma}_{W_\beta\times w_\alpha}\bm\Sigma_{W_\beta\times w_\alpha}^{\frac12}-\mathbf I_{p+1}\|_2&\leq\lambda_{\min }^{-1}(\text{cov}(\boldsymbol\eta_j)) \bigg\|\frac1M\sum_{j=1}^M\boldsymbol\eta_j\boldsymbol\eta_j^\top-\text{cov}(\boldsymbol\eta_j)\bigg\|_2\notag\\
		&=O_P\bigg(\frac1{\sqrt M}\bigg).	
	\end{align}
	Thus, Theorem 4 is proved.
\end{proof}

\begin{proof}[Proof of Theorem 5]~
	Note that
	\begin{align}
		\boldsymbol{S}_{j}(\boldsymbol\theta)&=-(\hat\alpha_j-\boldsymbol\theta^\top\hat{\boldsymbol\beta}_j)\hat{\boldsymbol\beta}_j-\bm\Sigma_{W_\beta W_\beta}\boldsymbol\theta+\boldsymbol\sigma_{W_\beta w_\alpha}\notag\\
		&=(w_{\alpha_j}-\boldsymbol\theta^\top\boldsymbol w_{\beta_j})\boldsymbol\beta_j+\bigg\{(w_{\alpha_j}-\boldsymbol\theta^\top\boldsymbol w_{\beta_j})\boldsymbol w_{\beta_j}-\bm\Sigma_{W_\beta W_\beta}\boldsymbol\theta+\boldsymbol\sigma_{W_\beta w_\alpha}\bigg\}\notag\\
		&=\boldsymbol J_{1j}+\boldsymbol J_{2j}.
	\end{align}
	Note that both $\boldsymbol J_{1j}$ and $\boldsymbol J_{2j}$ are sub-exponential variables with zero mean and covariance matrix
	\begin{align}
		\text{cov}(\boldsymbol J_{1j})=\frac1{mn_{\rm min}}\psi_\theta\bm\Psi_{\beta\beta},\quad \text{cov}(\boldsymbol J_{2j})=\frac1{n^2_{\rm min}}\bm\Sigma_{\rm BC}.
	\end{align}
	Therefore, we obtain 
	\begin{align}
		\text{cov}(\boldsymbol{S}_{j}(\boldsymbol\theta))=\bm\Sigma_S=\begin{cases}
			\frac1{mn_{\rm min}}\psi_\theta\bm\Psi_{\beta\beta},&\text{ if $m/n_{\rm min}\to0$},\\
			\frac1{mn_{\rm min}}\psi_\theta\bm\Psi_{\beta\beta}+\frac{c_0}{mn_{\rm min}}\bm\Sigma_{\rm BC},&\text{ if $m/n_{\rm min}\to c_0$},\\
			\frac1{n_{\rm min}^2}\bm\Sigma_{\rm BC},,&\text{ if $m/n_{\rm min}\to\infty$ and $\sqrt m/n_{\rm min}\to0$}.
		\end{cases}
	\end{align}
	Then by using Lemma A.5,
	\begin{align}
		\bigg\|\frac1m\sum_{j=1}^m\boldsymbol{S}_{j}(\boldsymbol\theta)\boldsymbol{S}_{j}(\boldsymbol\theta)^\top-\bm\Sigma_S\bigg\|_2=O_P\bigg(\sqrt{\frac{\log m} m}\bigg)||\bm\Sigma_S||_2.
	\end{align}
	By using the Slutsky's theorem,
	\begin{align}
		\bigg\|\frac1m\sum_{j=1}^m\hat{\boldsymbol{S}}_j(\hat{\boldsymbol\theta}_{\rm BEE})\hat{\boldsymbol{S}}_j(\hat{\boldsymbol\theta}_{\rm BEE})^\top-\bm\Sigma_S\bigg\|_2=O_P\bigg(\sqrt{\frac{\log m} m}\bigg)||\bm\Sigma_S||_2.
	\end{align}
	where
	\begin{align}
		\hat{\boldsymbol{S}}_j(\hat{\boldsymbol\theta}_{\rm BEE})&=-(\hat{\boldsymbol\theta}_{\rm BEE}^\top\hat{\boldsymbol\beta}_j-\hat\alpha_j)\hat{\boldsymbol\beta}_j+\hat{\bm\Sigma}_{W_\beta W_\beta}\hat{\boldsymbol\theta}_{\rm BEE}-\hat{\boldsymbol\sigma}_{W_\beta w_\alpha}
	\end{align}
	
	On the other hand, according to the proof of Theorem 3,
	\begin{align}
		\|m\hat{\mathbf F}_{\rm BEE}-\bm\Psi_{\beta\beta}||_2=O_P\bigg\{\max\bigg(\frac1{\sqrt m},\frac1{\sqrt n_{\rm min}},\frac{\sqrt{m}}{n_{\rm min}}\bigg)\bigg\}.
	\end{align}
	Note that \citet[A22(p223)]{Bickel2008} illustrates
	\begin{align}
		\|\mathbf A_1\mathbf A_2\mathbf A_3-\mathbf B_1\mathbf B_2\mathbf B_3\|_2=O_P\bigg\{\max\bigg(||\mathbf A_1-\mathbf B_1||_2,||\mathbf A_2-\mathbf B_2||_2,||\mathbf A_3-\mathbf B_3||_2\bigg)\bigg\},
	\end{align}
	where $\mathbf A_1,\mathbf A_2,\mathbf A_3,\mathbf B_1,\mathbf B_2,\mathbf B_3$ are six matrices with non-diverging maximum singular values. Hence,
	\begin{align}
		||\hat{\bm\Sigma}_{\rm BEE}(\hat{\boldsymbol\theta}_{\rm BEE})-\bm\Sigma_{\rm BEE}(\boldsymbol\theta)||_2&=\bigg\|(m\hat{\mathbf F}_{\rm BEE})^{-1}\bigg(\sum_{j=1}^m\hat{\boldsymbol{S}}_j(\hat{\boldsymbol\theta}_{\rm BEE})\hat{\boldsymbol{S}}_j(\hat{\boldsymbol\theta}_{\rm BEE})^\top\bigg)(m\hat{\mathbf F}_{\rm BEE})^{-1}-m\mathbf\Psi_{\beta\beta}^{-1}\bm\Sigma_S\mathbf\Psi_{\beta\beta}^{-1}\bigg\|_2\notag\\
		&=O_P\bigg\{\max\bigg(\sqrt{\frac{\log m} m},\frac1{\sqrt n_{\rm min}},\frac{\sqrt{m}}{n_{\rm min}}\bigg)\bigg\}||m\bm\Sigma_S||_2,
	\end{align}
	and consequently
	\begin{align}
		||\bm\Sigma_{\rm BEE}^{-\frac12}(\boldsymbol\theta)\hat{\bm\Sigma}_{\rm BEE}(\boldsymbol\theta)\bm\Sigma^{-\frac12}_{\rm BEE}(\boldsymbol\theta)-\mathbf I_p||_2=O_P\bigg\{\max\bigg(\sqrt{\frac{\log m} m},\frac1{\sqrt n_{\rm min}},\frac{\sqrt{m}}{n_{\rm min}}\bigg)\bigg\}.
	\end{align}
	Thus, Theorem 5 is proved.
\end{proof}

\begin{proof}[Proof of Theorem 6]~
	Note that $||\hat{\boldsymbol\theta}_{\rm BEE}-\boldsymbol\theta||_2=O_P(n_{\min }^{-\frac12})$ and hence $\hat\alpha_j-\hat{\boldsymbol\beta}_j^\top\hat{\boldsymbol\theta}_{\rm BEE}$ and $\hat\alpha_j-\hat{\boldsymbol\beta}_j^\top\boldsymbol\theta$ have the same distribution. For $j\in\mathcal{O}^c$,
	\begin{align}
		\hat\gamma_j=\varepsilon_j&=\hat\alpha_j-\hat{\boldsymbol\beta}_j^\top\hat{\boldsymbol\theta}_{\rm BEE}=w_{\alpha_j}-\boldsymbol w_{\beta_j}^\top\boldsymbol\theta+\boldsymbol w_{\beta_j}^\top(\hat{\boldsymbol\theta}_{\rm BEE}-\boldsymbol\theta)\notag\\
		&\sim\mathcal{N}(0,\sigma_{\varepsilon\varepsilon}),
	\end{align}
	where
	\begin{align}
		\sigma_{\varepsilon\varepsilon}=\boldsymbol\theta^\top\mathbf\Sigma_{W_\beta w_\alpha}\boldsymbol\theta+\sigma_{\omega_\gamma\omega_\gamma}-2\boldsymbol\theta^\top\boldsymbol\sigma_{W_\beta w_\alpha}.
	\end{align}
	As a result,
	\begin{align}
		\frac{\hat\gamma_j^2}{\sigma_{\varepsilon\varepsilon}}\sim\chi^2_1.
	\end{align}
	Denote $\kappa^*=F_{\chi^2_1}^{-1}(\kappa)$. Then by using Lemma A.1 of \citet{Huang2012},
	\begin{align}
		\Pr\bigg(\max_{j\in\mathcal{O}^c}\frac{\hat\gamma_j^2}{\sigma_{\varepsilon\varepsilon}}\leq \kappa^*\bigg)&=1-\Pr\bigg(\max_{j\in\mathcal{O}^c}\frac{\hat\gamma_j^2}{\sigma_{\varepsilon\varepsilon}}> \kappa^*\bigg)\geq 1-(m-|\mathcal O|)\Pr\bigg(\frac{\hat\gamma_j^2}{\sigma_{\varepsilon\varepsilon}}> \kappa^*\bigg)\notag\\
		&\geq 1-m\Pr\bigg(\frac{\hat\gamma_j^2}{\sigma_{\varepsilon\varepsilon}}> \kappa^*\bigg)\geq 1-m\exp\bigg(-\frac{(\sqrt{2\kappa^*-1}-1)^2}{4}\bigg).
	\end{align}
	By letting $\kappa^*=C_0\log m$ with $C_0$ being a sufficiently large constant,
	\begin{align}
		\Pr\bigg(\max_{j\in\mathcal{O}^c}\frac{\hat\gamma_j^2}{\sigma_{\varepsilon\varepsilon}}\leq \kappa^*\bigg)
		&\geq 1-\exp\bigg(\log m-\frac{2C_0\log m -2\sqrt{C_0\log m-1}}{4}\bigg)\notag\\
		&\geq 1-\exp\bigg(-\frac{(2C_0-4)\log m -2\sqrt{C_0\log m-1}}{4}\bigg)\to1,
	\end{align}
	if $m\to\infty$.\par
	On the other hand, for $j\in\mathcal{O}$, $\hat\gamma_j=\gamma_j+\varepsilon_j,$
	and hence
	\begin{align}
		\frac{\hat\gamma_j^2}{\sigma_{\varepsilon\varepsilon}}\sim\chi^2_1\bigg(\frac{\gamma_j^2}{\sigma_{\varepsilon\varepsilon}}\bigg),
	\end{align}
	where $\chi^2_1(\lambda)$ refers to the noncentral chi-squared distribution with degree of freedom 1 and  noncentrality parameter $\lambda$.
	Let $F_{\chi^2_1(\lambda)}(\cdot)$ be the CDF of this noncentral chi-squared distribution, which is indeed equal to
	\begin{align}
		F_{\chi^2_1(\lambda)}(x)=1-\bigg(Q(\sqrt{x}-\sqrt\lambda)+Q(\sqrt{x}+\sqrt\lambda)\bigg),
	\end{align}
	where $F_{\chi^2_1(\lambda)}(\cdot)$ be the CDF of $\chi^2_1(\lambda)$ and $Q(x)$ is the Gaussian Q-function, i.e., $Q(x)=1-\Phi(x)$ and $\Phi(x)$ is the CDF of standard normal distribution. 
	
	Note that there should exist a constant $D_0$ such that
	\begin{align}
		\frac{\gamma_j^2}{\sigma_{\varepsilon\varepsilon}}\geq D_0n_{\min }
	\end{align}
	where $D_0$ is a sufficient large constant. And
	\begin{align}
		\Pr\bigg(\min_{j\in\mathcal{O}}\frac{\hat\gamma_j^2}{\sigma_{\varepsilon\varepsilon}}\geq \kappa^*\bigg)&=1-\Pr\bigg(\min_{j\in\mathcal{O}^c}\frac{\hat\gamma_j^2}{\sigma_{\varepsilon\varepsilon}}< \kappa^*\bigg)\notag\\
		&\geq 1-\Pr\bigg(\frac{\hat\gamma_j^2}{\sigma_{\varepsilon\varepsilon}}<\kappa^*\bigg),\quad j\text{ is arbitrary element in }\mathcal O.
	\end{align}
	Hence,
	\begin{align}
		\Pr\bigg(\min_{j\in\mathcal{O}}\frac{\hat\gamma_j^2}{\sigma_{\varepsilon\varepsilon}}\geq \kappa^*\bigg)&\geq Q(\sqrt{\kappa^*}-\sqrt{D_0n_{\min }})+Q(\sqrt{\kappa^*}+\sqrt{D_0n_{\min }})\notag\\
		&\geq Q(\sqrt{C_0\log m}-\sqrt{D_0n_{\min }})+Q(\sqrt{C_0\log m}+\sqrt{D_0n_{\min }})\to1
	\end{align}
	if $m,n_{\min }\to\infty$. Thus, Theorem 6 is proved.
\end{proof}

\setstretch{1.05}
\bibliographystyle{Chicago}
\bibliography{mybib.bib}

\begin{thebibliography}{}

\bibitem[\protect\citeauthoryear{Benjamini and Hochberg}{Benjamini and
  Hochberg}{1995}]{benjamini1995controlling}
Benjamini, Y. and Y.~Hochberg (1995).
\newblock Controlling the false discovery rate: a practical and powerful
  approach to multiple testing.
\newblock {\em Journal of the Royal statistical society: series B
  (Methodological)\/}~{\em 57\/}(1), 289--300.

\bibitem[\protect\citeauthoryear{Bickel and Levina}{Bickel and
  Levina}{2008}]{bickel2008regularized}
Bickel, P.~J. and E.~Levina (2008).
\newblock Regularized estimation of large covariance matrices.
\newblock {\em The Annals of Statistics\/}~{\em 36\/}(1), 199--227.

\bibitem[\protect\citeauthoryear{Bowden, Davey~Smith, and Burgess}{Bowden
  et~al.}{2015}]{bowden2015mendelian}
Bowden, J., G.~Davey~Smith, and S.~Burgess (2015).
\newblock Mendelian randomization with invalid instruments: effect estimation
  and bias detection through egger regression.
\newblock {\em International Journal of Epidemiology\/}~{\em 44\/}(2),
  512--525.

\bibitem[\protect\citeauthoryear{Bowden, Davey~Smith, Haycock, and
  Burgess}{Bowden et~al.}{2016}]{bowden2016consistent}
Bowden, J., G.~Davey~Smith, P.~C. Haycock, and S.~Burgess (2016).
\newblock Consistent estimation in mendelian randomization with some invalid
  instruments using a weighted median estimator.
\newblock {\em Genetic Epidemiology\/}~{\em 40\/}(4), 304--314.

\bibitem[\protect\citeauthoryear{Bujak, Wasilewski, Osadnik, Jonczyk,
  Kolodziejska, Gierlotka, and Gkasior}{Bujak
  et~al.}{2015}]{bujak2015prognostic}
Bujak, K., J.~Wasilewski, T.~Osadnik, S.~Jonczyk, A.~Kolodziejska,
  M.~Gierlotka, and M.~Gkasior (2015).
\newblock The prognostic role of red blood cell distribution width in coronary
  artery disease: a review of the pathophysiology.
\newblock {\em Disease markers\/}~{\em 2015}.

\bibitem[\protect\citeauthoryear{Bulik-Sullivan, Finucane, Anttila, Gusev, Day,
  Loh, Duncan, Perry, Patterson, Robinson, et~al.}{Bulik-Sullivan
  et~al.}{2015}]{bulik2015atlas}
Bulik-Sullivan, B., H.~K. Finucane, V.~Anttila, A.~Gusev, F.~R. Day, P.-R. Loh,
  L.~Duncan, J.~R. Perry, N.~Patterson, E.~B. Robinson, et~al. (2015).
\newblock An atlas of genetic correlations across human diseases and traits.
\newblock {\em Nature Genetics\/}~{\em 47\/}(11), 1236--1241.

\bibitem[\protect\citeauthoryear{Bulik-Sullivan, Loh, Finucane, Ripke, Yang,
  Patterson, Daly, Price, and Neale}{Bulik-Sullivan et~al.}{2015}]{bulik2015ld}
Bulik-Sullivan, B.~K., P.-R. Loh, H.~K. Finucane, S.~Ripke, J.~Yang,
  N.~Patterson, M.~J. Daly, A.~L. Price, and B.~M. Neale (2015).
\newblock Ld score regression distinguishes confounding from polygenicity in
  genome-wide association studies.
\newblock {\em Nature Genetics\/}~{\em 47\/}(3), 291--295.

\bibitem[\protect\citeauthoryear{Burgess, Davies, and Thompson}{Burgess
  et~al.}{2016}]{burgess2016bias}
Burgess, S., N.~M. Davies, and S.~G. Thompson (2016).
\newblock Bias due to participant overlap in two-sample mendelian
  randomization.
\newblock {\em Genetic Epidemiology\/}~{\em 40\/}(7), 597--608.

\bibitem[\protect\citeauthoryear{Burgess, Foley, Allara, Staley, and
  Howson}{Burgess et~al.}{2020}]{burgess2020robust}
Burgess, S., C.~N. Foley, E.~Allara, J.~R. Staley, and J.~M. Howson (2020).
\newblock A robust and efficient method for mendelian randomization with
  hundreds of genetic variants.
\newblock {\em Nature Communications\/}~{\em 11\/}(1), 1--11.

\bibitem[\protect\citeauthoryear{Burgess and Thompson}{Burgess and
  Thompson}{2015}]{burgess2015multivariable}
Burgess, S. and S.~G. Thompson (2015).
\newblock Multivariable mendelian randomization: the use of pleiotropic genetic
  variants to estimate causal effects.
\newblock {\em American Journal of Epidemiology\/}~{\em 181\/}(4), 251--260.

\bibitem[\protect\citeauthoryear{Burgess and Thompson}{Burgess and
  Thompson}{2021}]{burgess2021mendelian}
Burgess, S. and S.~G. Thompson (2021).
\newblock {\em Mendelian randomization: methods for causal inference using
  genetic variants}.
\newblock Chapman and Hall/CRC.

\bibitem[\protect\citeauthoryear{Burgess, Thompson, and Collaboration}{Burgess
  et~al.}{2011}]{burgess2011avoiding}
Burgess, S., S.~G. Thompson, and C.~C.~G. Collaboration (2011).
\newblock Avoiding bias from weak instruments in mendelian randomization
  studies.
\newblock {\em International Journal of Epidemiology\/}~{\em 40\/}(3),
  755--764.

\bibitem[\protect\citeauthoryear{Calonico, Cattaneo, and Farrell}{Calonico
  et~al.}{2018}]{calonico2018effect}
Calonico, S., M.~D. Cattaneo, and M.~H. Farrell (2018).
\newblock On the effect of bias estimation on coverage accuracy in
  nonparametric inference.
\newblock {\em Journal of the American Statistical Association\/}~{\em
  113\/}(522), 767--779.

\bibitem[\protect\citeauthoryear{Cheng, Qiu, Chai, Sun, Xia, Shi, and
  Liu}{Cheng et~al.}{2022}]{cheng2022mr}
Cheng, Q., T.~Qiu, X.~Chai, B.~Sun, Y.~Xia, X.~Shi, and J.~Liu (2022).
\newblock Mr-corr2: a two-sample mendelian randomization method that accounts
  for correlated horizontal pleiotropy using correlated instrumental variants.
\newblock {\em Bioinformatics\/}~{\em 38\/}(2), 303--310.

\bibitem[\protect\citeauthoryear{Cheng, Zhang, Chen, and Liu}{Cheng
  et~al.}{2022}]{cheng2022mendelian}
Cheng, Q., X.~Zhang, L.~S. Chen, and J.~Liu (2022).
\newblock Mendelian randomization accounting for complex correlated horizontal
  pleiotropy while elucidating shared genetic etiology.
\newblock {\em Nature Communications\/}~{\em 13\/}(1), 1--13.

\bibitem[\protect\citeauthoryear{Diggle, Diggle, Heagerty, Liang, Zeger,
  et~al.}{Diggle et~al.}{2002}]{diggle2002analysis}
Diggle, P., P.~J. Diggle, P.~Heagerty, K.-Y. Liang, S.~Zeger, et~al. (2002).
\newblock {\em Analysis of longitudinal data}.
\newblock Oxford university press.

\bibitem[\protect\citeauthoryear{Ebrahim and Davey~Smith}{Ebrahim and
  Davey~Smith}{2008}]{ebrahim2008mendelian}
Ebrahim, S. and G.~Davey~Smith (2008).
\newblock Mendelian randomization: can genetic epidemiology help redress the
  failures of observational epidemiology?
\newblock {\em Human Genetics\/}~{\em 123\/}(1), 15--33.

\bibitem[\protect\citeauthoryear{Fan, Liao, and Mincheva}{Fan
  et~al.}{2011}]{fan2011high}
Fan, J., Y.~Liao, and M.~Mincheva (2011).
\newblock High dimensional covariance matrix estimation in approximate factor
  models.
\newblock {\em The Annals of Statistics\/}~{\em 39\/}(6), 3320.

\bibitem[\protect\citeauthoryear{Gill, Zuber, Dawson, Pearson-Stuttard, Carter,
  Sanderson, Karhunen, Levin, Wootton, Klarin, et~al.}{Gill
  et~al.}{2021}]{gill2021risk}
Gill, D., V.~Zuber, J.~Dawson, J.~Pearson-Stuttard, A.~R. Carter, E.~Sanderson,
  V.~Karhunen, M.~G. Levin, R.~E. Wootton, D.~Klarin, et~al. (2021).
\newblock Risk factors mediating the effect of body mass index and waist-to-hip
  ratio on cardiovascular outcomes: Mendelian randomization analysis.
\newblock {\em International Journal of Obesity\/}~{\em 45\/}(7), 1428--1438.

\bibitem[\protect\citeauthoryear{Grant and Burgess}{Grant and
  Burgess}{2021}]{grant2021pleiotropy}
Grant, A.~J. and S.~Burgess (2021).
\newblock Pleiotropy robust methods for multivariable mendelian randomization.
\newblock {\em Statistics in Medicine\/}~{\em 40\/}(26), 5813--5830.

\bibitem[\protect\citeauthoryear{Gresham, Dunham, and Botstein}{Gresham
  et~al.}{2008}]{gresham2008comparing}
Gresham, D., M.~J. Dunham, and D.~Botstein (2008).
\newblock Comparing whole genomes using dna microarrays.
\newblock {\em Nature Reviews Genetics\/}~{\em 9\/}(4), 291--302.

\bibitem[\protect\citeauthoryear{Group et~al.}{Group
  et~al.}{1994}]{scandinavian1994randomised}
Group, S. S. S.~S. et~al. (1994).
\newblock Randomised trial of cholesterol lowering in 4444 patients with
  coronary heart disease: the scandinavian simvastatin survival study (4s).
\newblock {\em The Lancet\/}~{\em 344\/}(8934), 1383--1389.

\bibitem[\protect\citeauthoryear{Hall}{Hall}{1992}]{hall1992effect}
Hall, P. (1992).
\newblock Effect of bias estimation on coverage accuracy of bootstrap
  confidence intervals for a probability density.
\newblock {\em The Annals of Statistics\/}, 675--694.

\bibitem[\protect\citeauthoryear{Jankova and Van De~Geer}{Jankova and Van
  De~Geer}{2018}]{jankova2018semiparametric}
Jankova, J. and S.~Van De~Geer (2018).
\newblock Semiparametric efficiency bounds for high-dimensional models.
\newblock {\em The Annals of Statistics\/}~{\em 46\/}(5), 2336--2359.

\bibitem[\protect\citeauthoryear{Klein, Zeiss, Chew, Tsai, Sackler, Haynes,
  Henning, SanGiovanni, Mane, Mayne, et~al.}{Klein
  et~al.}{2005}]{klein2005complement}
Klein, R.~J., C.~Zeiss, E.~Y. Chew, J.-Y. Tsai, R.~S. Sackler, C.~Haynes, A.~K.
  Henning, J.~P. SanGiovanni, S.~M. Mane, S.~T. Mayne, et~al. (2005).
\newblock Complement factor h polymorphism in age-related macular degeneration.
\newblock {\em Science\/}~{\em 308\/}(5720), 385--389.

\bibitem[\protect\citeauthoryear{Liang and Zeger}{Liang and
  Zeger}{1986}]{liang1986longitudinal}
Liang, K.-Y. and S.~L. Zeger (1986).
\newblock Longitudinal data analysis using generalized linear models.
\newblock {\em Biometrika\/}~{\em 73\/}(1), 13--22.

\bibitem[\protect\citeauthoryear{Lorincz-Comi, Yang, Li, and Zhu}{Lorincz-Comi
  et~al.}{2022}]{lorincz-comi2022mrbee}
Lorincz-Comi, N., Y.~Yang, G.~Li, and X.~Zhu (2022).
\newblock Mrbee: A novel bias-corrected multivariable mendelian randomization
  method.
\newblock {\em PubMed\/}.

\bibitem[\protect\citeauthoryear{MacArthur, Bowler, Cerezo, Gil, Hall,
  Hastings, Junkins, McMahon, Milano, Morales, et~al.}{MacArthur
  et~al.}{2017}]{macarthur2017new}
MacArthur, J., E.~Bowler, M.~Cerezo, L.~Gil, P.~Hall, E.~Hastings, H.~Junkins,
  A.~McMahon, A.~Milano, J.~Morales, et~al. (2017).
\newblock The new nhgri-ebi catalog of published genome-wide association
  studies (gwas catalog).
\newblock {\em Nucleic Acids Research\/}~{\em 45\/}(D1), D896--D901.

\bibitem[\protect\citeauthoryear{Morrison, Knoblauch, Marcus, Stephens, and
  He}{Morrison et~al.}{2020}]{morrison2020mendelian}
Morrison, J., N.~Knoblauch, J.~H. Marcus, M.~Stephens, and X.~He (2020).
\newblock Mendelian randomization accounting for correlated and uncorrelated
  pleiotropic effects using genome-wide summary statistics.
\newblock {\em Nature Genetics\/}~{\em 52\/}(7), 740--747.

\bibitem[\protect\citeauthoryear{Muirhead}{Muirhead}{2009}]{muirhead2009aspects}
Muirhead, R.~J. (2009).
\newblock {\em Aspects of multivariate statistical theory}.
\newblock John Wiley \& Sons.

\bibitem[\protect\citeauthoryear{Murray and Blume}{Murray and
  Blume}{2020}]{murray2020false}
Murray, M.~H. and J.~D. Blume (2020).
\newblock False discovery rate computation: Illustrations and modifications.
\newblock {\em arXiv preprint arXiv:2010.04680\/}.

\bibitem[\protect\citeauthoryear{Qi and Chatterjee}{Qi and
  Chatterjee}{2019}]{qi2019mendelian}
Qi, G. and N.~Chatterjee (2019).
\newblock Mendelian randomization analysis using mixture models for robust and
  efficient estimation of causal effects.
\newblock {\em Nature Communications\/}~{\em 10\/}(1), 1--10.

\bibitem[\protect\citeauthoryear{Rees, Wood, and Burgess}{Rees
  et~al.}{2017}]{rees2017extending}
Rees, J.~M., A.~M. Wood, and S.~Burgess (2017).
\newblock Extending the mr-egger method for multivariable mendelian
  randomization to correct for both measured and unmeasured pleiotropy.
\newblock {\em Statistics in Medicine\/}~{\em 36\/}(29), 4705--4718.

\bibitem[\protect\citeauthoryear{Rees, Wood, Dudbridge, and Burgess}{Rees
  et~al.}{2019}]{rees2019robust}
Rees, J.~M., A.~M. Wood, F.~Dudbridge, and S.~Burgess (2019).
\newblock Robust methods in mendelian randomization via penalization of
  heterogeneous causal estimates.
\newblock {\em PloS One\/}~{\em 14\/}(9), e0222362.

\bibitem[\protect\citeauthoryear{Sadreev, Elsworth, Mitchell, Paternoster,
  Sanderson, Davies, Millard, Smith, Haycock, Bowden, et~al.}{Sadreev
  et~al.}{2021}]{sadreev2021navigating}
Sadreev, I.~I., B.~L. Elsworth, R.~E. Mitchell, L.~Paternoster, E.~Sanderson,
  N.~M. Davies, L.~A. Millard, G.~D. Smith, P.~C. Haycock, J.~Bowden, et~al.
  (2021).
\newblock Navigating sample overlap, winner's curse and weak instrument bias in
  mendelian randomization studies using the uk biobank.
\newblock {\em medRxiv\/}.

\bibitem[\protect\citeauthoryear{Sanderson, Davey~Smith, Windmeijer, and
  Bowden}{Sanderson et~al.}{2019}]{sanderson2019examination}
Sanderson, E., G.~Davey~Smith, F.~Windmeijer, and J.~Bowden (2019).
\newblock An examination of multivariable mendelian randomization in the
  single-sample and two-sample summary data settings.
\newblock {\em International Journal of Epidemiology\/}~{\em 48\/}(3),
  713--727.

\bibitem[\protect\citeauthoryear{Sanderson, Spiller, and Bowden}{Sanderson
  et~al.}{2021}]{sanderson2021testing}
Sanderson, E., W.~Spiller, and J.~Bowden (2021).
\newblock Testing and correcting for weak and pleiotropic instruments in
  two-sample multivariable mendelian randomization.
\newblock {\em Statistics in Medicine\/}~{\em 40\/}(25), 5434--5452.

\bibitem[\protect\citeauthoryear{Schwarz}{Schwarz}{1978}]{schwarz1978estimating}
Schwarz, G. (1978).
\newblock Estimating the dimension of a model.
\newblock {\em The Annals of Statistics\/}, 461--464.

\bibitem[\protect\citeauthoryear{Sorlie, Garcia-Palmieri, Costas~Jr, and
  Havlik}{Sorlie et~al.}{1981}]{sorlie1981hematocrit}
Sorlie, P.~D., M.~R. Garcia-Palmieri, R.~Costas~Jr, and R.~J. Havlik (1981).
\newblock Hematocrit and risk of coronary heart disease: the puerto rico heart
  health program.
\newblock {\em American heart journal\/}~{\em 101\/}(4), 456--461.

\bibitem[\protect\citeauthoryear{Stock, Wright, and Yogo}{Stock
  et~al.}{2002}]{stock2002survey}
Stock, J.~H., J.~H. Wright, and M.~Yogo (2002).
\newblock A survey of weak instruments and weak identification in generalized
  method of moments.
\newblock {\em Journal of Business \& Economic Statistics\/}~{\em 20\/}(4),
  518--529.

\bibitem[\protect\citeauthoryear{Sudlow, Gallacher, Allen, Beral, Burton,
  Danesh, Downey, Elliott, Green, Landray, et~al.}{Sudlow
  et~al.}{2015}]{sudlow2015uk}
Sudlow, C., J.~Gallacher, N.~Allen, V.~Beral, P.~Burton, J.~Danesh, P.~Downey,
  P.~Elliott, J.~Green, M.~Landray, et~al. (2015).
\newblock Uk biobank: an open access resource for identifying the causes of a
  wide range of complex diseases of middle and old age.
\newblock {\em PLoS Medicine\/}~{\em 12\/}(3), e1001779.

\bibitem[\protect\citeauthoryear{Tam, Patel, Turcotte, Boss{\'e}, Par{\'e}, and
  Meyre}{Tam et~al.}{2019}]{tam2019benefits}
Tam, V., N.~Patel, M.~Turcotte, Y.~Boss{\'e}, G.~Par{\'e}, and D.~Meyre (2019).
\newblock Benefits and limitations of genome-wide association studies.
\newblock {\em Nature Reviews Genetics\/}~{\em 20\/}(8), 467--484.

\bibitem[\protect\citeauthoryear{Turley, Walters, Maghzian, Okbay, Lee,
  Fontana, Nguyen-Viet, Wedow, Zacher, Furlotte, et~al.}{Turley
  et~al.}{2018}]{turley2018multi}
Turley, P., R.~K. Walters, O.~Maghzian, A.~Okbay, J.~J. Lee, M.~A. Fontana,
  T.~A. Nguyen-Viet, R.~Wedow, M.~Zacher, N.~A. Furlotte, et~al. (2018).
\newblock Multi-trait analysis of genome-wide association summary statistics
  using mtag.
\newblock {\em Nature Genetics\/}~{\em 50\/}(2), 229--237.

\bibitem[\protect\citeauthoryear{Van~de Geer, B{\"u}hlmann, Ritov, and
  Dezeure}{Van~de Geer et~al.}{2014}]{van2014asymptotically}
Van~de Geer, S., P.~B{\"u}hlmann, Y.~Ritov, and R.~Dezeure (2014).
\newblock On asymptotically optimal confidence regions and tests for
  high-dimensional models.
\newblock {\em The Annals of Statistics\/}~{\em 42\/}(3), 1166--1202.

\bibitem[\protect\citeauthoryear{Verbanck, Chen, Neale, and Do}{Verbanck
  et~al.}{2018}]{verbanck2018detection}
Verbanck, M., C.-Y. Chen, B.~Neale, and R.~Do (2018).
\newblock Detection of widespread horizontal pleiotropy in causal relationships
  inferred from mendelian randomization between complex traits and diseases.
\newblock {\em Nature Genetics\/}~{\em 50\/}(5), 693--698.

\bibitem[\protect\citeauthoryear{Vershynin}{Vershynin}{2010}]{vershynin2010introduction}
Vershynin, R. (2010).
\newblock Introduction to the non-asymptotic analysis of random matrices.
\newblock {\em arXiv preprint arXiv:1011.3027\/}.

\bibitem[\protect\citeauthoryear{Vershynin}{Vershynin}{2018}]{vershynin2018high}
Vershynin, R. (2018).
\newblock {\em High-dimensional probability: An introduction with applications
  in data science}, Volume~47.
\newblock Cambridge University Press.

\bibitem[\protect\citeauthoryear{Visscher, Wray, Zhang, Sklar, McCarthy, Brown,
  and Yang}{Visscher et~al.}{2017}]{visscher201710}
Visscher, P.~M., N.~R. Wray, Q.~Zhang, P.~Sklar, M.~I. McCarthy, M.~A. Brown,
  and J.~Yang (2017).
\newblock 10 years of gwas discovery: biology, function, and translation.
\newblock {\em The American Journal of Human Genetics\/}~{\em 101\/}(1), 5--22.

\bibitem[\protect\citeauthoryear{Wang, Shi, Zhu, Hao, Chen, Cheng, Foo, and
  Wang}{Wang et~al.}{2022}]{wang2022mendelian}
Wang, K., X.~Shi, Z.~Zhu, X.~Hao, L.~Chen, S.~Cheng, R.~S. Foo, and C.~Wang
  (2022).
\newblock Mendelian randomization analysis of 37 clinical factors and coronary
  artery disease in east asian and european populations.
\newblock {\em Genome Medicine\/}~{\em 14\/}(1), 1--15.

\bibitem[\protect\citeauthoryear{Wijmenga and Zhernakova}{Wijmenga and
  Zhernakova}{2018}]{wijmenga2018importance}
Wijmenga, C. and A.~Zhernakova (2018).
\newblock The importance of cohort studies in the post-gwas era.
\newblock {\em Nature Genetics\/}~{\em 50\/}(3), 322--328.

\bibitem[\protect\citeauthoryear{Xue, Shen, and Pan}{Xue
  et~al.}{2021}]{xue2021constrained}
Xue, H., X.~Shen, and W.~Pan (2021).
\newblock Constrained maximum likelihood-based mendelian randomization robust
  to both correlated and uncorrelated pleiotropic effects.
\newblock {\em The American Journal of Human Genetics\/}~{\em 108\/}(7),
  1251--1269.

\bibitem[\protect\citeauthoryear{Yavorska and Burgess}{Yavorska and
  Burgess}{2017}]{yavorska2017mendelianrandomization}
Yavorska, O.~O. and S.~Burgess (2017).
\newblock Mendelianrandomization: an r package for performing mendelian
  randomization analyses using summarized data.
\newblock {\em International Journal of Epidemiology\/}~{\em 46\/}(6),
  1734--1739.

\bibitem[\protect\citeauthoryear{Ye, Shao, and Kang}{Ye
  et~al.}{2021}]{ye2021debiased}
Ye, T., J.~Shao, and H.~Kang (2021).
\newblock Debiased inverse-variance weighted estimator in two-sample
  summary-data mendelian randomization.
\newblock {\em The Annals of Statistics\/}~{\em 49\/}(4), 2079--2100.

\bibitem[\protect\citeauthoryear{Yi}{Yi}{2017}]{yi2017statistical}
Yi, G.~Y. (2017).
\newblock {\em Statistical analysis with measurement error or
  misclassification: strategy, method and application}.
\newblock Springer.

\bibitem[\protect\citeauthoryear{Yu and Cheng}{Yu and Cheng}{2020}]{yu2020uric}
Yu, W. and J.-D. Cheng (2020).
\newblock Uric acid and cardiovascular disease: an update from molecular
  mechanism to clinical perspective.
\newblock {\em Frontiers in Pharmacology\/}~{\em 11}, 582680.

\bibitem[\protect\citeauthoryear{Yuan, Liu, Guo, Yan, Xue, and Zhou}{Yuan
  et~al.}{2022}]{yuan2022likelihood}
Yuan, Z., L.~Liu, P.~Guo, R.~Yan, F.~Xue, and X.~Zhou (2022).
\newblock Likelihood-based mendelian randomization analysis with automated
  instrument selection and horizontal pleiotropic modeling.
\newblock {\em Science Advances\/}~{\em 8\/}(9), eabl5744.

\bibitem[\protect\citeauthoryear{Zhao, Wang, Hemani, Bowden, and Small}{Zhao
  et~al.}{2020}]{zhao2020statistical}
Zhao, Q., J.~Wang, G.~Hemani, J.~Bowden, and D.~S. Small (2020).
\newblock Statistical inference in two-sample summary-data mendelian
  randomization using robust adjusted profile score.
\newblock {\em The Annals of Statistics\/}~{\em 48\/}(3), 1742--1769.

\bibitem[\protect\citeauthoryear{Zhu}{Zhu}{2020}]{zhu2020mendelian}
Zhu, X. (2020).
\newblock Mendelian randomization and pleiotropy analysis.
\newblock {\em Quantitative Biology\/}, 1--11.

\bibitem[\protect\citeauthoryear{Zhu, Feng, Tayo, Liang, Young, Franceschini,
  Smith, Yanek, Sun, Edwards, et~al.}{Zhu et~al.}{2015}]{zhu2015meta}
Zhu, X., T.~Feng, B.~O. Tayo, J.~Liang, J.~H. Young, N.~Franceschini, J.~A.
  Smith, L.~R. Yanek, Y.~V. Sun, T.~L. Edwards, et~al. (2015).
\newblock Meta-analysis of correlated traits via summary statistics from gwass
  with an application in hypertension.
\newblock {\em The American Journal of Human Genetics\/}~{\em 96\/}(1), 21--36.

\bibitem[\protect\citeauthoryear{Zhu, Li, Xu, and Wang}{Zhu
  et~al.}{2021}]{zhu2021iterative}
Zhu, X., X.~Li, R.~Xu, and T.~Wang (2021).
\newblock An iterative approach to detect pleiotropy and perform mendelian
  randomization analysis using gwas summary statistics.
\newblock {\em Bioinformatics\/}~{\em 37\/}(10), 1390--1400.

\bibitem[\protect\citeauthoryear{Zhu, Zhu, Wang, Cooper, and Chakravarti}{Zhu
  et~al.}{2022}]{zhu2022genome}
Zhu, X., L.~Zhu, H.~Wang, R.~S. Cooper, and A.~Chakravarti (2022).
\newblock Genome-wide pleiotropy analysis identifies novel blood pressure
  variants and improves its polygenic risk scores.
\newblock {\em Genetic Epidemiology\/}~{\em 46\/}(2), 105--121.

\end{thebibliography}

\end{document}